\documentclass[12pt]{article}
\usepackage{amsmath}
\usepackage{mathtools} 
\usepackage[T1]{fontenc}
\usepackage{amssymb}
\usepackage{amsfonts}
\usepackage{amsthm}
\usepackage{mathrsfs}
\usepackage[all]{xy}
\usepackage{tensor}
\usepackage{comment}
\usepackage{stmaryrd}
\usepackage{bm}
\usepackage[normalem]{ulem}
\usepackage[usenames,dvipsnames]{xcolor}
\usepackage[colorlinks=true, citecolor=Blue, linkcolor=Blue]{hyperref}
\usepackage{ragged2e}
\usepackage{tikz}
    \usetikzlibrary{decorations.markings}

\usepackage{diagbox}  
\usepackage{makecell} 
\usepackage{hhline} 
    
\usepackage[labelfont={small, bf, sf}, textfont={small,sf}]{caption}
\numberwithin{equation}{section} 
\usepackage{authblk}
\usepackage[in]{fullpage}
\usepackage{cite}

\newcommand{\ups}{\upsilon}
\newcommand{\lie}{\pounds}

\newcommand{\brmod}[2]{\llbracket #1, #2 \rrbracket}

\newcommand{\h}[1]{{\hat{#1}}}
\newcommand{\wh}[1]{{\widehat{#1}}} 
\newcommand{\fs}{\mathscr{F}}
\newcommand{\sls}{\mathscr{S}}
\newcommand{\beom}{\mathcal{E}} 
\newcommand{\stm}{\mathcal{M}}
\newcommand{\sr}{\mathcal{U}}
\newcommand{\ns}{\mathcal{N}}
\newcommand{\cn}{\mathcal{C}}
\newcommand{\slsu}{\sls_{\sr}}
\newcommand{\flx}{\mathcal{F}}
\newcommand{\n}[1]{\mathscr{#1}}
\newcommand{\ps}{\mathscr{P}}
\newcommand{\psu}{\ps_{\sr}}
\newcommand{\anom}{\mathcal{A}}

\newcommand{\ja}{a}
\newcommand{\jl}{\nu}
\newcommand{\ra}{\rightarrow}
\newcommand{\cov}[1]{\overset{c}{#1}}
\newcommand{\mcov}[1]{\overset{vc}{#1}} 
\newcommand{\cb}{r}
\newcommand{\cgamma}{\rho}
\newcommand{\cflx}{\varepsilon}

\newtheorem{lemma}{Lemma}

\newcommand{\p}[1]{  {\underline{ #1} }   }

\newcommand{\un}[1]{ \underline{#1} }

\newcommand{\scri}{\mathscr{I}}
\newcommand{\scp}{\scri^+}

\newcommand{\ind}[1]{\indices{#1}}

\newcommand{\heq}{\mathrel{\mathop {\widehat=} }}

\newcommand{\ms}{\mathscr}
\newcommand{\mf}{\mathfrak}
\newcommand{\bb}{\mathbb}

\newcommand{\beq}{\begin{equation}}
\newcommand{\eeq}{\end{equation}}
\newcommand{\bes}{\begin{subequations}}
\newcommand{\ees}{\end{subequations}}
\newcommand{\bea}{\begin{eqnarray}}
\newcommand{\eea}{\end{eqnarray}}

\newcommand{\be}{\begin{equation}}
\newcommand{\ee}{\end{equation}}
\newcommand{\conf}{\Phi}
\newcommand{\gunphys}{g}
\newcommand{\gphys}{{\tilde g}}
\newcommand{\hateq}{\mathrel{\mathop {\widehat=} }}
\newcommand{\volume}{\eta}
\newcommand{\volumesmall}{\mu}
\newcommand{\nonaffinity}{\kappa}
\newcommand{\shape}{W}
\newcommand{\expansion}{\Theta}
\newcommand{\inducedmetric}{q}

\title{A general framework for gravitational charges and holographic renormalization}

\author[1,2]{Venkatesa Chandrasekaran\thanks{venchandrasekaran@ias.edu}}

\author[3]{\'Eanna \'E. Flanagan\thanks{eef3@cornell.edu}}
\author[3]{Ibrahim Shehzad\thanks{is354@cornell.edu}}
\author[4,5]{Antony J. Speranza\thanks{asperanz@gmail.com}}
	
\affil[1]{\small \it Berkeley Center for Theoretical Physics, Berkeley, CA, 94720, USA}
\affil[2]{\small \it Institute for Advanced Study, Princeton, NJ, 08540, USA}
\affil[3]{\small \it Department of Physics, Cornell University, Ithaca, NY, 14853, USA}
\affil[4]{\small \it Perimeter Institute for Theoretical Physics, 31 Caroline St. N, Waterloo, ON N2L 2Y5, Canada}
\affil[5]{\small \it Department of Physics, University of Illinois, Urbana-Champaign, Urbana IL 61801, USA}

\date{}

\begin{document}

\maketitle

\begin{abstract}

We develop a general framework for constructing charges associated with 
diffeomorphisms in gravitational theories using covariant phase space techniques.  
This framework encompasses both localized charges associated with spacetime subregions, as well as
global conserved charges 
of the full spacetime.  Expressions for 
the charges include contributions from the boundary and corner terms in the subregion action,
and are rendered unambiguous by appealing to the variational principle for the subregion, which selects
a preferred form of the symplectic flux through the boundaries.  
The Poisson brackets of the charges on the subregion phase space are 
shown to reproduce the bracket of Barnich and Troessaert for open subsystems, thereby giving a novel derivation
of this bracket from first principles.  
In the context of asymptotic boundaries, we show that the procedure of holographic renormalization
can be always applied to obtain finite charges and fluxes once suitable counterterms have been
found to ensure a finite action.  This enables the study of larger asymptotic symmetry groups
by loosening the boundary conditions imposed at infinity.
We further present an algorithm for explicitly computing the counterterms that renormalize the 
action and symplectic potential, and, as an application of our framework, demonstrate that it reproduces known expressions for 
the charges of the generalized Bondi-Metzner-Sachs algebra.

\end{abstract}

\flushbottom

\newpage
    
\tableofcontents

\newpage

\section{Introduction and summary}

Canonical methods in general relativity and other gravitational theories provide 
an important tool for understanding the theory's observables and degrees of freedom.
These methods are particularly well-suited for characterizing the subtle 
role played by diffeomorphisms, which serve as the gauge symmetries of these theories.
The gauge nature of diffeomorphisms is captured by the fact that,
in the absence of boundaries, they generate transformations 
on the gravitational phase space corresponding to degenerate directions of the presymplectic form;
equivalently, the Hamiltonians generating these diffeomorphisms vanish on-shell.  
Introducing boundaries, either at infinity or finite locations in spacetime, 
partially breaks the full diffeomorphism invariance of theory, and  results in nontrivial
charges associated with the broken gauge symmetries.  The nonzero contribution to the charges
comes purely from an integral over the boundary of the spacetime region,
which is a manifestation of the familiar fact that the on-shell Hamiltonian is a pure boundary
term in diffeomorphism-invariant theories. 

An important technical tool for investigating properties of diffeomorphism invariance is the 
covariant phase space formalism \cite{Witten:1986qs, Crnkovic1987, Crnkovic:1987tz, Ashtekar1991,
LeeWald1990, Wald:1993nt, Iyer:1994ys}.  Its advantage over other constructions of 
gravitational phase spaces is the fact that covariance is maintained throughout.
This allows the
consequences of diffeomorphism invariance to be easily discerned, the most important of which
is the localization of diffeomorphism charges to contributions from the boundary.  
These boundary charges find applications in a number of questions in gravitational physics,
including black hole entropy
\cite{Wald:1993nt, Iyer:1994ys, Strominger1998, 
Carlip_1999, Hawking_2016, Haco:2018ske,Chen:2020nyh,
Chandrasekaran:2020wwn}, 
asymptotic symmetries 
\cite{Stro-lectures,Compere:2018aar}, entanglement and edge modes 
\cite{Donnelly2016a, Speranza2018a, Freidel:2020xyx, Donnelly:2020xgu}, and 
holography 
\cite{Hollands:2005ya, Papadimitriou:2005ii,Hollands:2005wt}. 
Given the breadth of scenarios in which boundary charges find use, 
it is important to have a well-defined framework that constructs these 
charges in an unambiguous manner.  Unfortunately, there are a number of 
complications that arise  related to ambiguities in the formalism,  
renormalization at asymptotic boundaries, 
and equivocal definitions of charges, that have lead to 
differing results and conclusions regarding boundary charges in 
various contexts.  The goal of the present work is to develop a general
framework that addresses these complications and sharply characterizes
the choices that must be made to resolve the various ambiguities.

One major motivation for having such a framework is its 
applications to holography in asymptotically flat spacetime, an arena in which the Hamiltonian formulation can provide important insights
\cite{Aneesh:2021uzk}. One can approach holography in a bottom-up manner, wherein one uses knowledge of the symmetries and charges of the theory at asymptotic boundaries to extrapolate properties of a putative dual theory. The classic example of this is the discovery by Brown and Henneaux that the asymptotic charge algebra of $\text{AdS}_3$ gravity coincides with the Virasoro algebra of $\text{CFT}_2$ \cite{Brown-Hennaux}. In a similar manner, systematically understanding the symmetries and charges at null infinity could help characterize the structure of the boundary theory. In particular, motivated by the UV/IR correspondence of the standard AdS/CFT dictionary, one might hope that a prescription for IR renormalization of classical observables using the Hamiltonian formalism leads to insights on universal properties of the putative boundary theory in the UV. 
This procedure is known as holographic renormalization \cite{Witten:1998qj, Henningson1998a,
Balasubramanian:1999re,
DeHaro2001, Papadimitriou:2005ii},\footnote{The name holographic renormalization arose because the formalism originated in the context of holographic dualities between bulk and boundary theories.  However the formalism itself as used here does not require any such dualities and can be defined in purely classical contexts.}
and its use has been expanded to 
asymptotically flat applications; for example, it is needed in order to obtain
finite charges associated with the generalized BMS group \cite{CL, Campiglia:2014yka,
Compere:2018ylh}.

When describing boundary charges,
it is often useful to distinguish between
{\it global} charges and {\it
  localized} charges \cite{CFP}. Given a set of boundary conditions that define a phase space, global charges are given by integrals over a complete Cauchy surface.  They include contributions from 
 all the degrees of freedom of the theory, 
and  generate the corresponding symmetry on the global phase space \cite{CFP}. They are integrals over the codimension-two boundaries of the Cauchy surface, which will typically be a sum over cross-sections of all the codimension-one boundaries in the spacetime that the Cauchy slice intersects.

Localized charges instead 
arise when defining a phase space associated with a subsystem of the 
full theory, such
as when considering a subregion of spacetime.
Standard examples include: the interior of a timelike tube in spacetime, as
occurs the Brown-York quasilocal charge construction \cite{Brown:1992br}; 
the exterior region of a finite null hypersurface \cite{CFP, Chandrasekaran:2020wwn};
and the domain of dependence of a partial Cauchy surface ending on a cut of 
$\scp$ in asymptotically flat spacetimes \cite{Wald:1999wa}.  Such  
subsystems are fundamentally  open Hamiltonian systems, which interact through
their boundary with degrees of freedom of the complementary region.
Because of this interaction, the subregion symplectic form is not
conserved under evolution along the boundary, and hence there is no
integrable charge generating the diffeomorphism associated with this evolution
on the subsystem phase space.  
Instead, localized charges are defined
as a best approximation for the generator of the diffeomorphism on the subregion.

A procedure for defining localized charges in the covariant phase space was
put forward by Wald and Zoupas
\cite{Wald:1999wa}, and subsequently developed in \cite{CFP, Chandrasekaran:2020wwn}.
These charges satisfy a modification of Hamilton's equation in which
the symplectic form evaluated on a diffeomorphism variation
yields the variation of the charge, plus an additional term representing the 
flux.  In order to produce unambiguous results, a criterion must be given
for separating the charge from the flux in this equation, determining
this criterion is the main challenge in obtaining well-defined localized 
charges.  An additional set of independent ambiguities, known as
Jacobson-Kang-Myers (JKM) ambiguities
\cite{Jacobson:1993vj}, arise in the definitions of the 
theory's Lagrangian and symplectic potential, and naively affect both
localized and global charges. 
One would like to have a coherent framework in which all the ambiguities are resolved through a single unified principle.

We will show that the crucial ingredient is the choice of action for the subregion,
including boundary and corner terms.
Equivalently, this  can be viewed as a preferred choice for the symplectic flux
at each of the boundaries, which appear as the boundary terms in the 
variational principle for the chosen action.  
The idea to use the action principle to resolve ambiguities in the covariant phase 
was first proposed in \cite{Compere:2008us}, motivated by holographic considerations 
in asymptotically AdS spacetimes
\cite{Henningson1998a, Balasubramanian:1999re, DeHaro2001,
Papadimitriou:2005ii}.  This principle is also
partially inspired by Euclidean gravity, wherein one takes the action to be the fundamental object from which all other observables are computed.
Furthermore, it ties in with the Brown-York construction of quasilocal charges
\cite{Brown:1992br}, in which the subregion
action plays a central role, and one can show that  these quasilocal
charges agree with the canonical charges constructed when utilizing the action principle
to fix their ambiguities \cite{Chandrasekaran:2021hxc}.  
This perspective based on the full subregion action will allow us to resolve both sets of localized charge ambiguities in one fell swoop. 
Moreover, it will enable us to give a simple general argument that holographic renormalization can always be performed to obtain finite charges and fluxes, without imposing {\it any}
boundary conditions on the field variations beyond those contained in the equations of 
motion.  Indeed, as explained in \cite{Skenderis:2009kd},
such generality is one of the main
novelties that the holographic approach 
brings to the study
of gravitational charges.   
Thus, our framework unifies many different aspects of 
gravitational charges in diffeomorphism-invariant theories.

In what follows, we give a detailed summary of each of our main results

\subsection{Extended summary of results}
We begin in section \ref{sec:charges} by presenting the general framework for 
utilizing the covariant phase space in constructing gravitational 
charges.  While much of the material in this section is review, 
we present a number of results for handling 
background structures in the theory, which modify a number of formulas 
by noncovariant contributions.\footnote{Noncovariant corrections to covariant
phase space quantities have also been explored in \cite{Freidel:2021cjp}, which
contains some overlap with the results of section \ref{sec:charges}.}  The reasons for allowing noncovariances are twofold.
First, as was shown in \cite{Chandrasekaran:2020wwn}, central extensions in
 gravitational charge algebras arise due to noncovariant boundary terms
in the action, and
such extensions often contain critical information about properties of the theory.
Second, allowing for noncovariance 
extends the applicability of the covariant phase space to noncovariant
formulations of the theory, such as the ADM formulation \cite{Arnowitt:1962hi}, 
facilitating a straightforward
comparison between the formulations.  

The main objective of section \ref{sec:charges} is to arrive at unambiguous
expressions for the gravitational charges.  Ambiguities can arise in two
related but conceptually distinct ways.  The first are the JKM ambiguities
\cite{Jacobson:1993vj}, which occur in the formulation of the covariant 
phase space by Wald and collaborators 
\cite{Wald:1999wa, LeeWald1990, Wald:1993nt, Iyer:1994ys}
due to the fact that 
various quantities, such as the Lagrangian or the symplectic current,
are defined only up to addition of exact differential forms.  
We demonstrate in section \ref{sec:JKM} that the gravitational
charges can be defined in such a way as to be completely invariant
under the JKM transformations, including transformations involving 
noncovariant quantities.  This provides a powerful link between covariant
and noncovariant formulations of the theory, since any two formulations can be
viewed as being related by a JKM transformation.  This then demonstrates that 
the charges are not sensitive to the specific choices made in setting up the 
canonical framework.  

The second set of ambiguities occurs for localized charges constructed 
via the Wald-Zoupas procedure \cite{Wald:1999wa}.  These charges depend 
on the form of the flux through the boundary of the subregion, and 
a prescription is needed to fix the expression for the flux.  
Wald and Zoupas
gave a proposal called the {\it stationarity requirement}
 for 
fixing the ambiguity, which requires that the
decomposition of the symplectic potential
be chosen such that the 
flux vanishes identically in stationary spacetimes.  This condition, along with 
a requirement on the covariance properties of the flux, was shown to yield unambiguous
localized charges for BMS generators in 4D asymptotically flat spacetime
\cite{Wald:1999wa}.  On the other hand, there has been much recent interest in 
extended symmetry algebras at null infinity 
\cite{Barnich2009, Barnich:2011ct, Cachazo:2014fwa, Kapec:2014opa,
Campiglia:2014yka, CL, Compere:2018ylh, freidel2021weyl}, 
which were missed in older analyses due to imposition of
diffeomorphism-freedom conditions
at the boundary that do not
correspond to degeneracy directions of the symplectic form and are
thus not true gauge degrees of freedom.
One can demonstrate that the stationarity
and covariance requirements do not produce finite charges associated with 
these extended symmetries \cite{Flanagan:2019vbl}, and for sufficiently permissive
boundary conditions, the stationarity requirement may either fail, or not fully
fix all possible ambiguities in the flux.  This motivates finding 
an alternative for fixing the flux ambiguities.

We therefore focus in this work
on a different resolution that is more closely tied to the 
variational principle associated with the subregion.  This resolution 
was first proposed by Comp\`ere and Marolf \cite{Compere:2008us} 
(see also \cite{Andrade:2015gja, Andrade:2015fna}),
motivated 
by the covariant Peierls bracket construction that far predates the more modern treatments
of the covariant phase space \cite{Peierls:1952cb, DeWitt:1962cg, Marolf:1993zk, Marolf:1993af,
Hollands:2005ya}.
These ideas were subsequently expanded upon and formalized in the work of Harlow and Wu \cite{Harlow:2019yfa} and the extension of this construction to Wald-Zoupas localized charges was 
recently described by two of us \cite{Chandrasekaran:2020wwn}.  
It has also 
been employed in applications of extended symmetries of asymptotically AdS spaces and 
their flat space limits in \cite{Compere:2020lrt,Fiorucci:2020xto}.\footnote{A related
approach described in \cite{Freidel:2020xyx, Freidel:2021cjp} 
absorbs all boundary terms in the action into
a bulk Lagrangian.  Often, this produces results consistent with the action
variational principle, but it lacks some of the flexibility of the present formulation, requires arbitrary choices in how to extend the boundary term into the bulk,
and cannot handle the corner improvements described in section
\ref{sec:corners}.  
}  
The variational principle pertains to the 
full action for the subsystem, involving an integral of the Lagrangian in the bulk plus additional
boundary terms, which are chosen to ensure the action is stationary for a given choice of 
boundary conditions.  For a closed system, the boundary conditions  are essential in determining the 
dynamics of the theory.  Localized subregions instead behave like  open systems due to the presence
of symplectic flux through the boundary, and in this case boundary conditions should not
be imposed, as they would unnecessarily constrain the dynamics.  Nevertheless, the boundary
contribution in the variation of the action is used to describe the flux through
the boundary, and hence the form of the flux is largely determined by the choice 
of boundary condition one would have to impose if viewing the subregion as a closed system.

From the viewpoint of the variational principle, 
resolving the ambiguities 
in the covariant phase space
formalism thus amounts to finding a preferred form for the flux, or, equivalently, to a preferred
boundary condition
one would impose if treating the system as closed.  A particularly natural choice is to require that 
the flux be of Dirichlet form, meaning it depends algebraically on variations of the intrinsic
variables on the boundary.  
For example, at a timelike boundary in theories where the only dynamical field 
is the metric,  the Dirichlet condition
implies that the flux take the form 
$\beom = \pi^{ij}\delta h_{ij}$, where $h_{ij}$ is the induced metric and $\pi^{ij}$ can involve both
intrinsic and extrinsic quantities.  Similarly, on a null surface, the Dirichlet form of the flux
is $\beom = \pi^{ij}\delta q_{ij} + \pi_i \delta n^i$, where $q_{ij}$ is the degenerate 
induced metric and $n^i$ is the null generator.  
Arguments in favor of the Dirichlet form of the flux were presented in 
\cite{Chandrasekaran:2020wwn}, and 
include the connection to junction conditions at a surface, 
the semiclassical description of the path integral when gluing subregions, 
and a straightforward relation to the Brown-York and holographic constructions.  
For most of this work, we focus on the Dirichlet
form of the flux, but emphasize that most of the formal  constructions work for other choices
corresponding different boundary conditions, although these other choices yield different 
values of the charges and can affect their algebra.   

The demonstration in section \ref{sec:JKM} that the action, symplectic form, and 
localized charges are all insensitive to ambiguities is then performed by working 
out how the individual contributions to each of these quantities change under 
JKM transformations once the expression for the flux has been fixed.  
We also introduce a class of {\it boundary canonical transformations}, which resemble
the JKM transformations, but act nontrivially on the form of the flux, and hence 
change expressions for the charges.  Because these boundary canonical transformations
change the subregion action, this emphasizes that different choices of action
generically produce different charges.
A careful treatment of the definition of all quantities involved in constructing
the localized charges reveals an additional set of corner ambiguities in the 
charges described in sections \ref{subsec:covphasespace} and \ref{sec:corners},
that naively affect the values of the charge.  We further demonstrate 
in section \ref{sec:corners} that a corner 
improvement term in the localized charges fixes this ambiguity as well.

Having obtained ambiguity-free expressions for the charges, we proceed 
in section \ref{sec:algebra} 
to determine
the algebra they satisfy.  
This algebra can be defined by way of the bracket introduced by 
Barnich and Troessaert
in \cite{Barnich:2011mi} (henceforward referred to as the BT bracket), 
where it was postulated as a sensible choice that reproduces the 
algebra satisfied by the vector fields generating the diffeomorphisms on spacetime, up to 
extensions.  We present a new result deriving this bracket from first principles
by identifying it as the 
Poisson bracket of the localized charges on the subregion phase space.
This derivation relies on the flux being of Dirichlet form,
but the arguments continue to hold for a class of alternative forms of the 
flux, subject to certain conditions.    
The bracket of the localized charges in general does not close,
 but instead produces additional generators $K_{\xi,\zeta}$ that yield an extension
of the algebra satisfied by the spacetime vector fields.  
Explicit expressions for the extension terms are given in equations
(\ref{eqn:Kxizeta}) and (\ref{eqn:tilK}), 
which are consistent with the expressions originally
derived in \cite{Chandrasekaran:2020wwn}, suitably generalized to allow noncovariances
in the bulk Lagrangian.  We further show that the brackets 
between the new generators $K_{\xi,\zeta}$ and the localized charges $H_\xi$ coincides with 
the bracket postulated by Barnich and Troessaert, as long as the generators $K_{\xi,\zeta}$ 
depend only on intrinsic variables at the surface when employing the Dirichlet flux condition.  
This requirement is nontrivially satisfied for charges constructed at null surfaces in 
general relativity, which serves as a consistency check on the use of the BT bracket.

The final sections of this paper are devoted to charges constructed at asymptotic boundaries.
In section \ref{sec:compendium}, 
as a segue into holographic renormalization,
we review a number of asymptotic symmetry algebras that have been
proposed for 4D asymptotically flat space.  Our presentation focuses on the different
universal structures each algebra preserves, and we specifically analyze the cases of 
the standard BMS group, the generalized BMS group \cite{Campiglia:2014yka, CL}, and the recently 
proposed Weyl BMS group \cite{freidel2021weyl}, 
which in fact coincides with the symmetry group obtained in
\cite{CFP} for finite null boundaries.  Detailed derivations of these universal 
structures and their associated symmetry groups are given in appendix \ref{sec:derivegroups}.

We then turn to an analysis of the 
the holographic renormalization procedure that is needed to obtain finite results for 
asymptotic charges and their fluxes.  This procedure can be viewed 
as finding a boundary canonical transformation that renders the action finite,
after which all JKM-invariant quantities are finite as well.  
We further show that a JKM transformation can be performed to make each
individual term in the expressions for the charges finite as well. 
It has often been remarked that one reason for imposing boundary conditions on fields at asymptotic
boundaries is to ensure that the charges and fluxes have a finite limit to the boundary. The framework of holographic renormalization instead provides
a different perspective
\cite{Skenderis:2009kd}: one should allow for the most general asymptotic
expansion of the dynamical fields that are consistent with the 
equations of motion, and handle any divergences using the 
counterterms that renormalize the action.  
It was first demonstrated by Comp\`ere and Marolf that in 
asymptotically AdS space, the resulting symplectic structure obtained 
via the holographic renormalization procedure is 
finite for  all fluctuations
of the dynamical fields, 
which further implies the charges and fluxes are finite as well,
consistent with previous results on holographic asymptotic charges
\cite{Henningson1998a, Balasubramanian:1999re, DeHaro2001,
Papadimitriou:2005ii}. 
In section \ref{genarg}, we show that 
this argument applies quite generally to any asymptotic boundary, and give a general argument 
that the fluxes and charges are finite once a set of boundary terms that renormalize the action
have been found. In section \ref{exparg} we show that holographic
renormalization can always be successfully carried out,
by giving an algorithm for computing the terms that one must add to the symplectic potential and
Lagrangian to obtain finite renormalized 
quantities.  It is impossible to
simultaneously maintain covariance and achieve finiteness, so our
renormalized quantities break covariance through dependence on a
choice of background structure. This is entirely analogous to the situation in AdS/CFT, where renormalized asymptotic charges necessarily depend on the choice of radial cutoff surface, which translates into the appearance of the Weyl anomaly on the boundary
\cite{Henningson1998a, Balasubramanian:1999re,
Skenderis:2000in}.
Finally, in section \ref{expcalc}, we apply the formalism described in section \ref{exparg} to explicitly compute the renormalized symplectic potential and the localized charges associated with the generalized BMS group in vacuum general relativity in 4D asymptotically flat spacetimes.

We conclude in section \ref{sec:discussion} 
with several points of discussion and avenues for future work.

\subsection{Notation}  \label{sec:notation}

 Unless otherwise stated, we will work in $d+1$ spacetime dimensions with metric signature $(-,+,+,\cdots)$. We will use
the indices $a,b,c$ for $d+1$ dimensional tensors in spacetime and $i,j,k$ for $d$ dimensional tensors intrinsic to a surface embedded in spacetime. The conformal factor in our notation will be denoted by $\Phi$ (instead of the more commonly used symbol, $\Omega$, which we will reserve for the symplectic form). We will use $\scri$ to denote null infinity in asymptotically flat spacetimes, $\scp$ where we specialize to future null infinity,
and $\hateq$  to denote equality on $\scp$ (or more generally on a null surface). The null normal to a null surface will be denoted by $n_a$,
and the auxiliary null vector on a null surface will be denoted by $l^{a}$. $\kappa$ is used to denote the inaffinity associated with a null vector and is defined by $n^a \nabla_a n^b \hateq \nonaffinity n^b$. Often an index free notation will be used to denote differential forms, 
although the indices will be made explicit where convenient. 
For example, $\eta \equiv \eta_{i_{1} i_{2} \cdots i_{d}}$ and $\mu \equiv \mu_{i_{1} i_{2} \cdots i_{d-1}}$ will denote the volume forms on codimension-1 and codimension-2 surfaces respectively. 
We will use $i_{v} \eta$ to denote the inner product of a vector field, $v^{a}$, with a differential form (in this case $\eta$). On occasion, the contracted indices will be displayed while the uncontracted indices will be left implicit. In other places, where convenient, all of the indices will be made explicit. In summary, we will 
freely use any of the expressions $i_{v}\eta$, $v^{i} \eta_{i}$, $v^{i} \eta_{i j_{2} \cdots j_{d}}$ to denote the contraction of $v^i$ into the form $\eta$.

Pullbacks to surfaces will be denoted using underlines, i.e., the pullback of 
$\theta$ to a surface will be denoted by $\underline{\theta}$.
 When working with the covariant phase space,
 $\fs$ will be used to denote the field configuration space of a theory, while $\sls$ will represent the space of field configurations that satisfy the equations of motion. Operations on it including $L_{\hat{\xi}}$, $\delta$, $I_{\hat{\xi}}$, and $\Delta_{\hat{\xi}}$ will be defined in section~\ref{subsec:covphasespace}, and capitalized calligraphic letters $\n A, \n B,\ldots$ will be used as abstract indices on $\sls$.
Note also that for simplicity, we will not distinguish between ``pre-symplectic'' and ``symplectic'' for quantities defined on the pre-phase space and the true phase space (see the second paragraph of section~\ref{subsec:covphasespace} for detail). Finally, Table \ref{tab:notation} lists various differential forms used in this paper along with their degrees on phase space and on spacetime, and the equations where they first appear.
\newline
\begin{table*}[t]
\centering
\footnotesize
\begin{tabular}{| c || c | c | c | c |}
\hline
\diagbox{$\sls$ degree}{Spacetime \\degree}&$(d+1)$ & $d$ & $(d-1)$ & $(d-2)$\\ 
\hhline{| = # = | = | = | = |}

0 & \makecell{$L'\eqref{eqn:dL}$, \\$\cov{L}\eqref{eqn:covL}$} 
& \makecell*[l]{ $b'\,\eqref{eqn:b'}$, $\cov{r}\,\eqref{eqn:b'ambig}$, \\
$\ell '\,\eqref{eqn:thdecomp}$, $J_{\xi}'\,\eqref{eqn:noethercurrent}$, \\
$\mcov{J}_{\xi} \eqref{eq:Jvc}$,  $a\,\eqref{eqn:L'shift}$, \\ $B\,\eqref{eq:boundarycanonicaltransformation}$ } 
&  \makecell[l]{ $e\,\eqref{eqn:b'ambig}$, $\mcov{Q}_{\xi}\eqref{eq:Qdefn}$, \\
$Q_{\xi}'\eqref{eq:Qprimexi}$, $h_{\xi}\eqref{eqn:Hxi}$, \\
$f\,\eqref{eqn:l'shift}$, $c'\,\eqref{eqn:bldecomp}$, \\
$\tilde{h}_{\xi} \eqref{eqn:hxiimproved}$ }
& 

\\
\hline
1 &  
& \makecell[l]{$\theta\,\eqref{eqn:dL}$, $\cov{\theta}\,\eqref{eqn:covth}$,\\
$\beom\,\eqref{eqn:thdecomp}$ }
& \makecell*[l]{$\lambda'\,\eqref{eqn:noncovtheta}$, $\cov{\rho}\,\eqref{eqn:lambda'ambig}$, \\
$\nu\,\eqref{eqn:thprime}$, $\beta'\, \eqref{eqn:thdecomp}$, \\
$\Lambda\,\eqref{eq:boundarycanonicaltransformation}$, 
$\varepsilon\,\eqref{eqn:bldecomp}$ }
& \makecell[l]{$\chi\,\eqref{eqn:lambda'ambig}$,
  $\gamma'\,\eqref{eqn:bldecomp}$,\\ $\mu_{\xi}\eqref{conj2}$, $\zeta\eqref{cct}$}
\\
\hline
\makecell*{2} &  & \makecell[l]{$\omega'\, \eqref{eqn:omega}$} &  &  \\
\hline 
\end{tabular}
\caption{
A summary of the
  various
  differential forms that are defined in our covariant phase space formalism,
 showing their spacetime degrees and phase space ($\sls$) degrees,
along with the equations where
 they are first introduced. We  
generally employ a convention where Greek or calligraphic letters denote forms with
 phase space degree greater than zero, and Latin letters denote forms of
 phase space degree
 zero.  See the paragraph above \eqref{eqn:dL} for the meaning of the prime notation,
 the paragraph below \eqref{eqn:noncov-defn} for the meaning of $\cov{}$, and footnote \ref{footnote:vc} for the meaning of $\mcov{}$.
} 
\label{tab:notation}
\end{table*}

\normalsize

Finally, when dealing with subregions, it is important to keep track of the orientations
of the various components of its boundary, for which we 
follow the conventions of \cite{Harlow:2019yfa}.  
Beginning with the codimension-$0$
subregion $\sr$ with $\ns$ a null or timelike component of the boundary, we choose 
the orientation of $\ns$ to be that induced as part of $\partial\sr$.  
The orientation of a spatial surface $\Sigma$ inside of $\sr$ whose boundary
intersects $\ns$ will be oriented as part of the boundary of its past, and the 
codimension-$2$ surface $\partial\Sigma$ defining a cut of $\ns$ will inherit
the induced orientation as a boundary of $\Sigma$.  Note that this means that 
$\partial\Sigma$ has the opposite orientation as that induced as part 
of the boundary of its past in $\ns$.  We define the volume form $\eta$ 
on $\ns$ to be one consistent with this choice of orientation, and similarly 
define $\mu$ on $\partial\Sigma$ to be consistent with its orientation.  
See appendix \ref{app:nullsurf} for the details of these volume forms 
when $\ns$ is a null surface.

\section{Gravitational charges at finite boundaries }
\label{sec:charges}

In any gravitational theory defined on a spacetime region with boundary, there are 
nonzero charges associated with diffeomorphisms that act near the boundary.  Depending on the 
context, one can distinguish between two related notions of charges, namely, {\it global} charges and 
{\it localized} charges.  Global charges are defined when the spacetime region under consideration can 
be viewed as a closed system, which occurs when considering the entire spacetime, or else working with 
a subregion of spacetime on which boundary conditions are imposed to prevent any interaction with 
the complementary region.  These charges generate the symmetry transformation of their associated
diffeomorphism on phase space via Hamilton's equation, and are conserved under time evolution.  
On the other hand, localized charges are defined for a subregion of spacetime, which is not assumed to 
be isolated from its complement.  Such charges need not be conserved due to the presence of nonzero
fluxes through the boundary, and in general will not faithfully generate the transformation
associated with the diffeomorphism.  Nevertheless, these localized charges provide useful
notions of quasilocal energy and momentum for subregions in phase space, and, as we will 
discuss, satisfy an algebra that closely resembles the diffeomorphism algebra of their corresponding
vector fields.  

Despite the distinctions, the two notions of charges are not entirely independent of each other.
Instead, a global charge can be viewed as a special case of a localized charge, in which the spacetime region is specialized to a closed system and the fluxes of the charge vanish.  
For this reason, we will focus in this work on the more general construction of localized
charges, and simply mention at various points how the construction can be specialized to 
global charges.

This section reviews the construction of localized gravitational charges
using covariant phase space techniques.  
The procedure was initially developed 
by Wald and Zoupas \cite{Wald:1999wa}, and in the present work  
we specifically focus on a number of recent developments 
on the handling of boundaries in the covariant phase space that have led to 
resolutions of the various ambiguities that can appear in the formalism
\cite{Compere:2008us, Andrade:2015gja, Andrade:2015fna, Harlow:2019yfa, Chandrasekaran:2020wwn}.  
The resolution comes from demanding that the symplectic potential $\beom$ 
describing
the flux through the subregion's boundary be of Dirichlet form.  
We will demonstrate explicitly that this fixes both the 
standard JKM ambiguities present in the covariant phase space formalism
\cite{Jacobson:1993vj, Iyer:1994ys}, as well as the additional ambiguity in identifying the flux 
when employing the Wald-Zoupas procedure.  
In fact, we will see that the formalism is invariant under generically {\it noncovariant} JKM
transformations, which, in particular, allows for formulations involving a bulk Lagrangian that is 
not spacetime-covariant, such as in the ADM formulation of the theory 
\cite{Arnowitt:1962hi}.
This provides maximal flexibility in identifying charges, allowing one to switch between a covariant
or noncovariant formulation depending on the application; invariance under JKM transformations 
ensures that the final result for the charges will not depend on this intermediate choice.  
We also describe in section \ref{sec:corners} a resolution of an additional set of ambiguities
involving corner contributions to the action, leading to an improved set of localized
charges.  These corner-improved charges generalize the proposal of reference 
\cite{Chandrasekaran:2020wwn} to allow for a noncovariant bulk Lagrangian and symplectic potential.

Throughout this section, we assume that boundaries are at finite locations in spacetime, and that all quantities have finite limits to the boundaries.  This assumption excludes asymptotic boundaries such as spatial infinity or future null infinity in asymptotically flat spacetimes, which can brought to a finite location in spacetime via conformal compactification, 
at the expense of having some of the dynamical fields diverge on the boundary.  
Later in Sec.\ \ref{sec:hr}, 
we will discuss the modifications and generalizations of the formalism that are necessary to handle asymptotic boundaries, based on the technique of holographic renormalization.

\subsection{Covariant  phase space}\label{subsec:covphasespace}
We begin with a brief review of the covariant phase space construction
\cite{Witten:1986qs, Crnkovic1987, Crnkovic:1987tz, Ashtekar1991, LeeWald1990, Wald:1993nt, Iyer:1994ys}
 in order to 
establish notation, which largely coincides with that used 
in \cite{Chandrasekaran:2020wwn}, and 
to point to places where we generalize the standard treatments.  
For recent reviews and more in-depth discussions 
of the covariant phase space, see \cite{Harlow:2019yfa,CFP}.

The idea behind the covariant phase space is to provide a canonical description of 
a field theory defined on a manifold
$\stm$
without breaking covariance by singling out a foliation of constant-time slices, as it done in more 
standard phase space constructions.  This is achieved by working with
the space $\sls$ of all field configurations satisfying the equations of motion, viewed as 
a subspace of the space $\fs$ of all field configurations.  
In a globally hyperbolic spacetime,  each
solution in $\sls$ can be identified, up to gauge transformations,
with its initial data defined on a Cauchy slice $\Sigma$, and 
since this initial data comprises the usual phase space of the theory, we see that there is a canonical
identification between $\sls$ modulo gauge transformations and the standard noncovariant phase space.\footnote{We will later consider subregions of spacetime 
which are not globally hyperbolic,
so this identification will not hold in those cases, but the construction nevertheless 
will allow us to define a sensible notion of phase space for the subregion.}
Since the phase space arises as a quotient of $\sls$ by the action of the gauge group, we will find
that $\sls$ has the structure of a pre-phase space, on which we will construct a pre-symplectic form
that has degenerate directions.  Most calculations will be done on $\sls$, bearing in mind that
eventually the quotient must be taken to arrive at expressions for the true phase space.  Throughout this work, we will drop the ``pre'' label for objects 
defined on $\sls$, and simply point out where it is important to distinguish
between the pre-phase space and true phase space.

The spaces $\fs$ and $\sls$ are infinite-dimensional manifolds, on which certain standard
differential geometry concepts are well-defined.  The dynamical fields $\phi$ (which will later be 
taken to consist of the metric and any matter fields) define a collection of functions
on field space, and the gradients of these functions are denoted $\delta\phi$.  
Differential forms of higher degree on field space can then be constructed 
by taking wedge products, and 
we will employ the notation where the product $\alpha \beta$ of two field-space differential
forms is always assumed to be a field-space wedge product, and hence satisfies
$\alpha\beta = (-1)^{ab} \beta\alpha$, where $a$ and $b$ are the respective form degrees 
of $\alpha$ and $\beta$. The operator $\delta$ then defines an exterior derivative on the 
space of field-space differential forms in the usual way. Vector fields are  defined 
by infinitesimal variations of the field configuration, and since vectors tangent to 
solution space $\sls$ must preserve the equations of motion, they are parametrized by 
solutions of the linearized field equations.  Given a vector field $V$ on $\sls$, we 
denote the operation of contraction with a differential form by $I_V$, so that  in
particular $I_V\delta\phi$ gives a phase space function that returns the 
linearized solution corresponding to $V$ around each background solution.
We can also take Lie derivatives along a given vector field $V$ in field space,
which we denote $L_V$, and its action on differential forms can be computed via 
Cartan's magic formula, 
\be
L_V = I_V \delta + \delta I_V.
\label{eqn:magic}
\ee

Our main focus in this work will be diffeomorphism-invariant theories.  
Infinitesimal diffeomorphisms are generated
by vector fields $\xi^a$ on spacetime, and they act on fields via the spacetime
Lie derivative $\lie_\xi \phi$.  
Diffeomorphism invariance implies that $\lie_\xi\phi$ is 
a solution to the linearized field equations, and hence defines a vector field 
on $\sls$, denoted $\h\xi$, through the equation $I_{\h\xi}\delta\phi = \lie_\xi\phi$.
The vector field $\xi^a$ can itself be viewed as a function
on field space, and often it is taken to be a constant, meaning $\delta\xi^a =0$.
However, in many applications it is useful to consider transformations
generated by field-dependent diffeomorphisms, for which $\delta \xi^a\neq 0$.
The Lie bracket $[\h\xi,\h\zeta]_{\fs}$
on field space of the vectors $\h\xi$ associated with field-dependent
$\xi^a$ is given by (see appendix \ref{app:fs}) 
\begin{align}
[\h\xi,\h\zeta]_{\fs} &= -\wh{\brmod{\xi}{\zeta}} \label{eqn:fsbrack} \\
\brmod{\xi}{\zeta}^a &= [\xi,\zeta]^a -I_{\h\xi}\delta \zeta^a + I_{\h\zeta} \delta\xi^a.
\label{eqn:brmod}
\end{align}
This expression employs the modified Lie bracket $\brmod{\cdot}{\cdot}$ introduced
in \cite{Barnich:2009se}, and its relation to the field space Lie bracket was
noted in \cite{Gomes:2018dxs}.  Since the vectors $\h\xi$ are tangent to the solution
space submanifold $\sls$ in $\fs$, the bracket $[\h\xi, \h\zeta]_{\sls}$ is also given 
by (\ref{eqn:fsbrack}).

We will be interested in objects defined on field space that may not transform covariantly
under diffeomorphisms.  Noncovariances arise in objects that depend on a 
background structure such as a nondynamical field.  Being nondynamical means that such a field 
is constant in field space, and hence $L_{\h\xi}$ acts trivially on it.  
In order to track the lack of covariance of a field 
space differential form, it is useful to define the anomaly operator $\Delta_{\h\xi}$, 
first introduced in \cite{Hopfmuller2018}, which acts on field space differential
forms constructed from local fields as\footnote{The operator $I_{\wh{\delta\xi}}$ acts
on the local field variations as $I_{\wh{\delta\xi}}\delta\phi = \lie_{\delta\xi}\phi$. 
See appendix \ref{app:fs} for additional details.}
\beq\label{eqn:noncov-defn}
\Delta_{\h\xi} = L_{\h\xi} - \lie_\xi - I_{\wh{\delta\xi}}.
\eeq
This operator provides a means for replacing field space Lie derivatives $L_{\h\xi}$ with 
spacetime Lie derivatives $\lie_\xi$,  keeping track of the anomalous transformation
of an object when doing so.  A covariant object is one that satisfies $\Delta_{\h\xi} \alpha
=0$, so, for example, since the dynamical fields are covariant, the statement $\Delta_{\h\xi}\phi =0$
is equivalent to the oft-used identity $L_{\h\xi}\phi = \lie_\xi \phi$.  
On the other hand, a nondynamical field $\psi$ satisfies $L_\h\xi\psi = 0$ even though the spacetime
Lie derivative is generically nonzero.  In this case, the anomaly is given by $\Delta_\h\xi \psi =
-\lie_\xi\psi$.
When it is important to emphasize that a certain object is fully
covariant, we will denote it with an overset $c$, 
as in $\cov{\alpha}$; 
hence, for any such quantity, one may always
assume $\Delta_\h\xi \cov\alpha = 0$.

The dynamics of the theory is specified in terms of its Lagrangian $L'$, taken to be a top form
on spacetime, so that the action is given by $ \int_{\stm} L'$ up to boundary terms.  As we will discuss shortly, various quantities 
that we will consider depend on ambiguities in the definition
of the Lagrangian and related quantities, and we employ the notation
that quantities that depend on these ambiguities are 
indicated with a prime, as in $L'$.  Any primed quantity should 
be assumed to be noncovariant in general.
Varying the Lagrangian yields the field equations and symplectic potential $\theta'$ for the theory
according to
\beq \label{eqn:dL}
\delta L' = E\cdot \delta \phi + d\theta'.
\eeq
The solution space $\sls$ which will serve as the pre-phase space for the theory 
consists of all field configurations satisfying the field equations $E=0$.  
Our main focus will be theories  whose field equations are 
diffeomorphism-invariant, meaning $\Delta_{\h\xi} (E\cdot \delta \phi) = 0$.
A condition that guarantees diffeomorphism invariance is that the Lagrangian
be covariant up to an exact term, $\Delta_{\h\xi} L' = d a_\xi'$.  
We will further
restrict attention to theories in which the anomalous  term $a_\xi'$ 
can be written as the anomalous transformation of some other quantity
defined on the boundary, $a_\xi' = \Delta_{\h\xi} b'$.  
This implies that there exists
a choice of Lagrangian that differs from 
$L'$ by an exact term, $\cov L = L'-db'$, and is
fully covariant,
$\Delta_{\h\xi} \cov L = 0$.\footnote{This assumption precludes theories such as topologically
massive gravity 
\cite{Deser:1981wh, Deser:1982vy} whose Lagrangians are not covariant
for any choice of boundary term due to
the presence of Chern-Simons-like terms, but nevertheless yield diffeomorphism invariant
field equations.  The most 
general definition of a diffeomorphism-invariant theory would be one whose 
equations of motion satisfy $\Delta_{\h\xi}(E\cdot \delta \phi) = 0$, which, in light of 
equation  
(\ref{eqn:dL}), implies the anomaly of the Lagrangian need only satisfy
\beq
\Delta_{\h\xi}\delta L' = d\Delta_{\h\xi} \theta'.
\eeq
Given that the formalism is invariant under addition of noncovariant boundary terms,
as discussed in section \ref{sec:JKM}, it seems likely that most of the results described 
in this work can be extended to this more general class of diffeomorphism-invariant theories.
It would be interesting to analyze such generalizations in more detail, for example,
as explored in \cite{Freidel:2021cjp}.
}
Iyer and Wald have shown that whenever there is a covariant Lagrangian,
one can find a symplectic potential $\cov \theta$ that is covariant as well,
$\Delta_{\h\xi}\cov \theta = 0$ \cite{Iyer:1994ys}.    
The covariant symplectic potential can differ from $\theta'$ by 
the addition of an exact term and a total variation, and hence
there must exist quantities $b'$ and $\lambda'$ satisfying the
equations
\begin{subequations}
\label{eqn:blambdadef}
\begin{eqnarray}
\Delta_{\h\xi} L' &=& d\Delta_{\h\xi} b' \label{eqn:b'} \\
\Delta_{\h\xi} \theta' &=& \Delta_{\h\xi}\delta b' + d\Delta_{\h\xi} \lambda',
\label{eqn:noncovtheta}
\end{eqnarray}
\end{subequations}
For a given Lagrangian $L'$ and symplectic potential $\theta'$, 
equations (\ref{eqn:b'}) and (\ref{eqn:noncovtheta}) will be taken
as the definitions of $b'$ and $\lambda'$. Once $b'$ and $\lambda'$ 
satisfying these equations have 
been found, the associated covariant Lagrangian and 
symplectic potential are defined to be
\begin{subequations}
\begin{eqnarray}
\cov L &=& L'-db'  \label{eqn:covL}\\
\cov \theta &=& \theta' - \delta b' - d\lambda' . \label{eqn:covth}
\end{eqnarray}
\end{subequations}

Equations (\ref{eqn:b'}) and (\ref{eqn:noncovtheta}) 
fix $b'$ and $\lambda'$ in terms of 
$L'$ and $\theta'$ up to shifts of the form 
\begin{subequations}
\label{eqn:echidef}
\begin{eqnarray}
b'&\ra& b' + \cov \cb +de \label{eqn:b'ambig}\\
\lambda' &\ra& \lambda' -\delta e + \cov \cgamma + d\chi \label{eqn:lambda'ambig}
\end{eqnarray}
\end{subequations}
with $\cov \cb$ and $\cov \cgamma$ covariant and $e$ and $\chi$ generically noncovariant.
However, we will see below that the localized charges and other
relevant 
quantities do not depend on the freedom to shift by the 
covariant quantities $\cov \cb$, $\cov \cgamma$, nor on
the shift in $\lambda'$ by $d\chi$.  In principle, the 
charges {\it are} sensitive to the shift by $e$ if $\Delta_\h\xi e
\neq 0$, but this can be resolved using a more refined 
treatment of corner terms, as explained in section 
\ref{sec:corners}.

Finally, we mention that the standard ambiguities that 
appear when working with $L'$ and $\theta'$
arise from the fact that any other Lagrangian  that differs
from $L'$ by an exact term, $L' + da'$, yields the same equation
of motion, and hence is an equally valid choice for defining the bulk
dynamics.  For such a shifted Lagrangian, any shifted 
symplectic potential of the form 
\beq
\label{eqn:thprime}
\theta' + \delta a' + d\nu '
\eeq
will satisfy the relation (\ref{eqn:dL}), and hence defines
a valid symplectic potential.  
These freedoms to shift $L'$ and $\theta'$ are often presented 
as ambiguities in the covariant phase space formalism
\cite{Jacobson:1993vj, Iyer:1994ys};
however, it has recently been understood that such ambiguities may be resolved 
by specifying the form of the boundary condition one would impose to 
ensure vanishing symplectic flux through the boundary of the subregion 
\cite{Compere:2008us, Andrade:2015gja, Andrade:2015fna, Harlow:2019yfa, Chandrasekaran:2020wwn}.  
This resolution is explored in detail in section \ref{sec:JKM}, where it
is shown that the charges, fluxes, and subregion action
all involve combinations of the various
objects that are manifestly invariant under these shifts.

\subsection{Symplectic form}
\label{sec:presymplecticform}

Before constructing localized charges associated with a subregion, we must first 
restrict the solution space to the subregion, and equip it with a symplectic structure.
To this end, we let $\sr$ denote the open set in $\stm$ defining the subregion of 
interest, whose boundary includes a timelike or null component
$\ns$.  There may be additional boundaries to the future and past of $\sr$, and,
although these do not play a major role in the construction of charges in the 
present work, these additional boundaries will become important when considering
more detailed resolutions of corner ambiguities, as discussed in 
section \ref{sec:corners}.
We will restrict attention to the space of solutions within the subregion $\sr$,
with no boundary conditions imposed at $\ns$.  We  denote this 
restricted solution space by $\slsu$.

We now consider spatial slices $\Sigma$ in $\sr$ whose boundaries $\partial \Sigma$ lie in $\ns$.
We will define a symplectic form $\Omega$ associated with $\partial \Sigma$
as an integral over $\Sigma$ and $\partial\Sigma$.  The resulting localized phase spaces $(\slsu,\Omega)$ 
will serve as the starting point for constructing localized charges, and  it is important to 
remember that they depend on both the subregion solution space $\slsu$ as well as a choice 
of cut of the boundary.

Two specific examples that illustrate this general framework are as
follows. First,
we take $\sr$ to be a globally hyperbolic, asymptotically flat
spacetime,
$\ns$ to be future null infinity $\scri^+$, and
$\Sigma$ to be an asymptotically null slice which intersects  
$\scri^+$ in some cut $\partial \Sigma$ \cite{Wald:1999wa}.
Second, we take $\sr$ to be a
timelike tube in spacetime, $\ns$ to be the timelike boundary
$\partial \sr$ of the tube,
and $\Sigma$ to be a spatial slice whose
boundary $\partial \Sigma$ lies in ${\cal N}$.  This second example is the
context for the Brown-York quasilocal charge construction \cite{Brown:1992br}.
Note that in both of these examples, the subregion solution space $\slsu$ is not in
one-to-one correspondence with the space of initial data   
on $\Sigma$.
This is a general feature of the framework, since 
$\Sigma$ is generally not a Cauchy surface for the subregion.
In the timelike tube example this arises because we have not imposed
any boundary conditions on $\partial \sr$.

The symplectic form will be constructed as a sum of two terms, one capturing 
the bulk contribution and one involving a boundary contribution.  The bulk term 
is constructed as the integral over a spatial slice $\Sigma$ through $\sr$ 
of the symplectic current,
\beq \label{eqn:omega}
\omega' = \delta\theta'.
\eeq
To determine the boundary contribution, we first consider the pullback 
$\p\theta'$ of the symplectic potential to $\ns$, and decompose it into three 
terms 
\beq \label{eqn:thdecomp}
\p\theta' \heq -\delta \ell' + d\beta' + \beom,
\eeq
where we refer to $\ell'$ as the {\it boundary term}, $\beta'$ as the {\it corner term},
and $\beom$ as the {\it flux term}.  The reason for this terminology relates to the variational
principle for the subregion.
Neglecting contributions from past and future boundaries, 
the action for the subregion $\sr$ is defined to be
\beq\label{eqn:action}
S = \int_{\sr} L' + \int_{\ns}\ell'.
\eeq
Varying this action and applying equations (\ref{eqn:dL}) and 
(\ref{eqn:thdecomp}), we find
\beq
\delta S = \int_{\sr} E\cdot \delta \phi + \int_{\ns} \beom + \int_{\partial \ns}\beta'
\eeq
and hence it is stationary both with the bulk field equations 
hold $E\cdot \delta\phi = 0$ and when the flux through the boundary 
vanishes, $\beom = 0$.  The corner term $\beta'$ 
localizes to the past and future boundaries of $\ns$, and in a complete treatment,
additional corner contributions to the action should be added at the codimension-2
boundaries of $\ns$ and the past and future boundaries,
as described in, e.g., Refs. \cite{Hayward1993, Booth:2001gx, Lehner_2016}.  
Although not crucial to the remaining 
discussion of this paper, these corner contributions to the action can produce
some modifications to the formalism, as described in 
section \ref{sec:corners}.

Without specifying the form of the flux term $\beom$, equation (\ref{eqn:thdecomp})
is ambiguous, since we can always shift it by exact terms and total variations
$\beom\rightarrow \beom +\delta B-d\Lambda$ by making compensating changes
to $\ell'$ and $\beta'$.  These changes  affect the subregion action (\ref{eqn:action}),
as well as the definitions of the charges,
and hence to avoid such ambiguities, 
it is paramount to specify a criterion for selecting
a preferred choice for $\beom$.  In making such a choice, it is important
to realize that the form of $\beom$ determines the boundary condition
one would impose in a variational principle for the subregion by the above discussion.
While different choices are available for these
boundary conditions, we mention that
it is often most useful to choose those in which 
$\beom$ takes a Dirichlet form, meaning only variations of intrinsic quantities
on the surface without derivatives appear in $\beom$. 
 For a timelike surface, this means
\beq \label{eqn:timeD}
\beom = \pi^{ij} \delta h_{ij}
\eeq
where $h_{ij}$ is the induced metric, while for a null surface it means
\cite{Chandrasekaran:2020wwn, Chandrasekaran:2021hxc}
\beq\label{eqn:nullD}
\beom = \pi^{ij}\delta q_{ij} + \pi_i \delta n^i
\eeq
where $q_{ij}$ is the degenerate induced metric, and $n^i$ is the null generator.  
A number of arguments in favor of the Dirichlet form of the flux were presented
in \cite{Chandrasekaran:2020wwn}, such as the relation to junction conditions
across $\ns$ and the semiclassicality of the gravitational path integral 
when gluing subregions.  We will also utilize this condition in section
\ref{sec:algebra} when deriving the algebra satisfied by the localized charges, but 
we argue that other forms of the flux also allow the derivation to go through.  
In writing equations (\ref{eqn:timeD}) and (\ref{eqn:nullD}), we have restricted 
attention to theories such as general relativity that admit a Dirichlet variational
principle (or equivalently, possesses second order equations of motion), and have 
neglected any contributions from matter fields to the symplectic potential.
Note that the conjugate momenta $\pi^{ij}$, $\pi_i$ can involve objects constructed from both
the extrinsic and intrinsic geometry of the surface.  
The Dirichlet requirement fixes the form of $\beom$ up to addition of boundary
and corner terms constructed entirely from  intrinsic quantities, and 
in section \ref{sec:hr} we will discuss how these 
purely intrinsic ambiguities are used in the 
context of holographic renormalization.

We can further interpret how to view $\beom$ by 
taking a variation of equation (\ref{eqn:thdecomp}) and rearranging terms, 
which yields
\beq
\delta\beom = \p\omega' -d\delta\beta'.
\eeq
This shows that  $\beom$ serves as a symplectic potential for the pullback
of the symplectic form $\p\omega'-d\delta\beta'$.  Here, the term $d\delta\beta'$
is precisely of the form of the 
ambiguity in the symplectic potential described in equation (\ref{eqn:thprime}),
and, as described in \cite{Harlow:2019yfa}, by considering an extension of $\beta'$ away from the surface $\ns$, we can
view $\theta'-d\beta'$ as the symplectic potential everywhere in the bulk.
The associated symplectic current is then $\omega' - d\delta \beta'$, 
and integrating this 
over a spatial slice $\Sigma$ 
yields a symplectic form that is the sum of a bulk and boundary term, 
\beq\label{eqn:Omcorner}
\Omega = \int_{\Sigma} \omega' - \int_{\partial\Sigma} \delta\beta'.
\eeq

We remark that we will require the quantities $L'$, $\theta'$, $b'$
and $\lambda'$ to be continuous everywhere on the spacetime subregion
and in particular everywhere on its boundary.  This condition is
necessary for passing from $(d+1)$-dimensional bulk integrals to
$d$-dimensional boundary integrals using Stokes' theorem.  By
contrast, the quantities $\ell'$, $\beta'$ and ${\cal E}$ associated
with the decomposition of the pullback of $\theta'$ to a boundary
component will not be required to be continuous across a corner
joining two boundary components.  This greater generality for these
quantities goes hand in
hand with the use of corner terms in the formalism in
Secs.\ \ref{sec:corners} and \ref{sec:hr} below.

\subsection{Localized charges}
\label{sec:WZcharges}

Having identified a symplectic structure for the subregion $\sr$,
we can proceed to construct gravitational charges associated with 
diffeomorphisms that act near the boundary $\ns$.  
Diffeomorphism invariance of the field equations implies the existence 
of a conserved Noether current $J_\xi'$ associated with each diffeomorphism
generated by a given vector $\xi^a$.  
It follows from equation (\ref{eqn:b'}) that under the action of 
a diffeomorphisms on phase space, the Lagrangian $L'$ transforms as 
\be
I_{\h\xi}\delta L' = \lie_\xi L' +\Delta_{\h\xi}L' 
= di_\xi L' + d\Delta_{\h\xi} b'.
\ee
However, from the definition (\ref{eqn:dL}), the left hand side can be written as
\begin{align}
I_{\h\xi}\delta L' = E\cdot I_{\h\xi}\delta\phi+ dI_{\h\xi}\theta',
\end{align}
and so defining the Noether current to be 
\beq
J_\xi' = I_{\h\xi}\theta' - i_\xi L' - \Delta_{\h\xi} b',
\label{eqn:noethercurrent}
\eeq
we see that $dJ_\xi' = -E\cdot I_{\h\xi}\delta \phi$, which vanishes on shell.
Here, we find a correction to the usual definition of the Noether current involving the 
noncovariance of the boundary term, $\Delta_\h\xi b'$, which was identified previously in 
\cite{Harlow:2019yfa}.
We can relate the Noether current (\ref{eqn:noethercurrent}) to the Noether current
$\mcov J_\xi$ constructed from the covariant
Lagrangian and symplectic potential\footnote{The superscript ``$vc$''
\label{footnote:vc} in this expression stands for ``vector covariant,'' and 
is used to indicate that the only noncovariance in $\mcov J_\xi$ arises from
its dependence on the noncovariant vector field $\xi^a$.  This notation will 
be used to indicate any quantity such as $\mcov J_\xi$ depending linearly 
on $\xi^a$ and its derivatives whose noncovariance is given by
\beq \label{eqn:Delvc}
\Delta_\h\zeta \mcov J_\xi = \mcov J_{\left(\Delta_\h\zeta \xi \right)} = 
-\mcov J_{\left(\brmod{\zeta}{\xi} -I_\h\xi \delta \zeta \right)}.
\eeq} 
\beq \label{eq:Jvc}
\mcov J_\xi = I_\h\xi \cov \theta - i_\xi \cov L
\eeq
using (\ref{eqn:covL}) and (\ref{eqn:covth}), 
 which produces
\begin{align}\label{eqn:Jxi'}
J_{\xi}' = I_{\h\xi}\cov \theta+I_{\h\xi}\delta b' + d I_{\h\xi}\lambda' 
- i_\xi \cov L
-i_\xi db' - I_{\h\xi}\delta b' +\lie_\xi b'
=\mcov J_\xi + d(i_\xi b' + I_{\h\xi}\lambda '),
\end{align}
showing that $\mcov J_\xi$ and $J_\xi'$ differ by an exact term. 
Furthermore, since $\mcov J_\xi$ is identically closed on shell and covariantly
constructed for any vector $\xi^a$, it can be expressed as the
exterior derivative
 of a potential,
 \be \label{eq:Qdefn}
 \mcov J_\xi = d \mcov Q_\xi\,,
 \ee
that is covariantly constructed from $\xi^a$ and the dynamical fields
\cite{W-closed}.  The relation (\ref{eqn:Jxi'}) then
implies that $J_\xi'$ is also expressible as the exterior derivative $
J_\xi' = d Q_\xi'$ of a potential $Q_\xi'$, given by
\beq \label{eq:Qprimexi}
Q_\xi' = \mcov Q_\xi + i_\xi b' +I_{\h\xi}\lambda'.
\eeq

The localized charges $H_\xi$ are now constructed by evaluating the contraction of 
the field space vector field $\h\xi$ into the symplectic form.  Using
the identity (see appendix \ref{sec:phase-space-calc})
\beq\label{eqn:IWident}
-I_{\h\xi}\omega' = d(\delta Q_\xi' -Q_{\delta\xi}' -i_\xi\theta' -\Delta_{\h\xi}\lambda'),
\eeq
and the definition (\ref{eqn:Omcorner}) for the subregion symplectic form,
we find that the contraction of $\h\xi$ into $\Omega$ is give by
\beq \label{eqn:Ham}
 -I_{\h\xi}\Omega = \int_{\partial\Sigma}\left(\delta Q_\xi'
-Q_{\delta\xi}'-\Delta_{\h\xi}\lambda'- i_\xi\theta' + I_{\hat \xi} \delta \beta'\right).
\eeq
Note that because this contraction localizes to a pure boundary integral, any 
diffeomorphism supported purely in the interior of $\Sigma$ is a degeneracy of 
$\Omega$, reflecting that such transformations are pure gauge.  
If $\xi^a$ generated a genuine, global symmetry of the 
subregion phase space, the right hand side
of (\ref{eqn:Ham}) would have to be the total variation $\delta H_\xi$
of a quantity $H_\xi$ that would be identified as
the charge generating the symmetry.  In this case,
equation (\ref{eqn:Ham}) simply becomes the statement of Hamilton's equation,
$-I_\h\xi\Omega = \delta H_\xi$.  However, it is clear from
inspection that the terms $Q'_{\delta\xi} + \Delta_\h\xi \lambda' +i_\xi\theta' + I_{\hat \xi} \delta \beta'$ 
 generically do not take the form of a total variation upon integration
 over $\partial\Sigma$, absent the imposition
of boundary conditions.   
While such boundary conditions arise naturally for 
global charges for closed subsystems, 
in the more general context of an open, localized phase space, such boundary conditions 
unnecessarily constrain the dynamics and eliminate dynamical degrees of freedom 
associated with fluxes of radiation modes.
In this case, we seek to define
a set of localized charges, which satisfy a modification of Hamilton's equation involving a term 
representing the flux of degrees of freedom escaping the 
subregion.

Using the decomposition (\ref{eqn:thdecomp}) of $\theta'$, we find that equation (\ref{eqn:Ham})
can be reorganized into the form (see Appendix \ref{sec:phase-space-calc})
\beq \label{eqn:IxiOm}
-I_{\h\xi}\Omega = \int_{\partial\Sigma}\delta h_{\xi} 
-\int_{\partial\Sigma}\big(i_\xi \beom -\Delta_{\h\xi}(\beta'-\lambda') + h_{\delta\xi}\big),
\eeq
where we have defined the localized charge density $h_\xi$ to be 
\beq \label{eqn:Hxi}
h_\xi = Q_\xi' + i_\xi \ell' - I_{\h\xi}\beta'.
\eeq
This formula takes the same form as the expression derived by Harlow and Wu
\cite{Harlow:2019yfa}, applied now in a context where boundary conditions are not
imposed on the phase space, as  in \cite{Chandrasekaran:2020wwn}.
The first term in (\ref{eqn:IxiOm}) is  a total variation, and we are 
led to identify this with the localized charge associated with 
$\xi^a$,
\beq \label{eqn:Hxiint}
H_\xi = \int_{\partial\Sigma}h_\xi.
\eeq  
The remaining terms in (\ref{eqn:IxiOm})
represent the loss of symplectic flux through the boundary 
$\ns$ under a flow generated by $\xi^a$ 
that moves $\partial\Sigma$ along this boundary.
The modified Hamilton's equation involving nontrivial fluxes for the 
localized charge then takes the form
\beq \label{eqn:Hammod}
\delta H_\xi = -I_{\h\xi}\Omega + \flx_\h\xi 
\eeq
where 
\beq \label{eqn:Axi}
\flx_\h\xi\equiv \int_{\partial\Sigma}\left(
 i_\xi \beom -\Delta_{\h\xi}(\beta'-\lambda') + h_{\delta\xi} \right).
\eeq
 We denote this flux by $\flx_{\h\xi}$ rather than simply $\flx_\xi$ to emphasize that it can depend
nontrivially on the field-dependence of the 
generator of the diffeomorphism.\footnote{We have separated the entire contribution
coming from $\delta\xi^a$ into the flux term, although for cases where $\delta\xi^a$ takes a specific 
form, it may be possible to separate off a total variation from $h_{\delta\xi}$ to include 
as a correction to the charge.  Such field dependence is used in 
\cite{Adami:2020ugu,Ruzziconi:2020wrb, Adami:2021sko,Geiller:2021vpg,Grumiller:2021cwg,Adami:2021nnf}, for example, to 
cancel some terms appearing in the flux, to arrive at integrable generators in the 
absence of gravitational waves.
 }

The charges constructed via equation (\ref{eqn:Hxiint}) obey a nontrivial conservation equation, which can
be obtained by computing the exterior derivative of the charge density $h_\xi$.  This yields 
the identity (derived in Appendix \ref{sec:phase-space-calc}) 
\beq \label{eqn:dhxi}
dh_\xi = I_{\h\xi}\beom -\Delta_{\h\xi}(\ell'+b') - i_\xi(L'+d\ell'),
\eeq
and integrating this between two cuts $S_1$ and $S_2$ 
of the boundary $\ns$ produces the anomalous continuity
equation,
\beq
\label{eqn:anomalouscontinuity}
H_{\xi}(S_2) - H_\xi(S_1) 
= -\int_{\ns_1^2}\left(I_{\h\xi}\beom -\Delta_{\h\xi}(\ell'+b')\right)
\eeq
where the last term in (\ref{eqn:dhxi}) does not contribute since $\xi^a$ is taken to 
be tangent to $\ns$. The minus sign in this equation appears due to the choice 
of orientations of $\ns$ and $\partial\Sigma$, discussed in section 
\ref{sec:notation}.
The first term on the right of this equation is interpreted as the symplectic flux
out of the subregion, while the second term involving $\Delta_{\h\xi}(\ell'+b')$ 
is an anomalous violation of the conservation equation \cite{Chandrasekaran:2020wwn}.

\subsection{Ambiguities and boundary canonical transformations}
\label{sec:JKM}

At this point, it is worth commenting that various objects introduced above, 
such as $\Omega$ (\ref{eqn:Omcorner})  and $h_\xi$ (\ref{eqn:Hxi}), have been defined
without the prime notation.  This is because these objects are in fact insensitive to
the two ambiguities we have mentioned to this point, corresponding to 
shifting $L'\rightarrow L'+d\ja$ and $\theta'\rightarrow \theta' + \delta \ja + d\jl$,
where $\ja$ and $\jl$ are allowed to be noncovariant in general.
We refer to these shifts of $L'$ and $\theta'$ as JKM transformations,
having first been identified in the work of Jacobson, Kang, and Myers 
\cite{Jacobson:1993vj}.
Under such a transformation, we require that the flux $\beom$ remain invariant,
since we are taking the form of the flux to be a physical input defining the dynamics
of the subregion.  
In order to keep $\beom$ invariant even while $\theta'$ changes under the transformation,
we must also shift the quantities $\ell'$ and $\beta'$ appearing in the decomposition
(\ref{eqn:thdecomp}). The JKM transformations of the basic quantities defining the 
phase space are then given by
\begin{subequations}
\label{eqn:JKMtrans}
\begin{align}
L' &\ra L' + d\ja \label{eqn:L'shift}\\
\theta' &\ra \theta' + \delta \ja + d\jl \label{eqn:th'shift}\\
\beom &\ra \beom 
\end{align}
\end{subequations}
which then imply
\begin{subequations}
\begin{align}
 \omega' &\ra \omega' + d\delta\jl\\
 b' &\ra b' + \ja\\
 \lambda' &\ra \lambda' + \jl\\
 \ell' &\ra \ell' -\ja \\
 \beta' &\ra \beta' + \jl \\
 J_\xi' &\ra J_\xi' + d(i_\xi \ja + I_{\h\xi}\jl) \\
 Q_\xi' &\ra Q_\xi' + i_\xi \ja + I_{\h\xi}\jl
\end{align}
\end{subequations}
Given these transformations, it is immediate to check from the definitions
(\ref{eqn:Omcorner}), (\ref{eqn:Hxi}), (\ref{eqn:Axi}), and (\ref{eqn:action}) 
that the 
symplectic form $\Omega$, the charge density $h_\xi$,  the symplectic
flux $\flx_\h\xi$, and the subregion action $S$ are all invariant.

Note in particular that if we start with
a noncovariant Lagrangian and symplectic potential $L'$ and $\theta'$, and 
perform a JKM transformation with $\ja = -b'$, $\jl=-\lambda'$, we obtain
a covariant Lagrangian and symplectic potential, $\cov L$ and 
$\cov \theta$, which are often used in standard treatments of the covariant
phase space.
In this case, the expressions for the charges and symplectic form reduce to 
those constructed in \cite{Chandrasekaran:2020wwn}, which utilized the 
covariant Lagrangian and symplectic potential.  
Invariance of the charges under 
JKM transformations then implies that these expressions will agree with charges
constructed using a Lagrangian and symplectic form that differs by the addition of 
noncovariant boundary and corner terms.  
The main message here is that the only choices that affect the charges are 
the field equations $E\cdot \delta \phi = 0$ and the form of the flux $\beom$. 
The quantity $\beom$ can be viewed as a boundary equation of motion, 
analogous to the bulk expression $E\cdot \delta\phi$, 
and, although we do not impose this field equation for generic localized
subregions, the equation $\beom=0$ would be the boundary condition one would 
have to impose to have a well-defined variational principle for the subregion.  
Note that it is not necessary to absorb all the boundary terms in the 
action principle into total derivative terms in the definition of the Lagrangian
$L'$.  Instead, when writing a variational principle for the subregion, one
should view $L'$ as the bulk part of the action, and $\ell'$ as the boundary
contribution, as in equation (\ref{eqn:action}), 
and together they produce an action that is independent of 
JKM transformations.  The choice of a particular $L'$ is then largely a matter 
of convenience. When discussing consequences of diffeomorphism invariance,
such as the first law of black hole mechanics
\cite{Wald:1993nt, Iyer:1994ys},
it is usually most transparent to work with covariant $\cov L$ and $\cov \theta$.  However, 
it can also be advantageous to work with a noncovariant Lagrangian, such as 
when using the ADM formalism \cite{Arnowitt:1962hi}, or in the context of 
holographic renormalization in order to obtain finite spacetime integrals, as 
expanded upon in section \ref{sec:hr}.  The results of this section indicate 
that charges obtained in either formulation coincide, and the transformations
(\ref{eqn:JKMtrans})
provide a means for translating
between different choices.

Since the resolution of JKM ambiguities employed in this work relies on fixing
a preferred choice
of the flux $\beom$, it is worth commenting on the different choices that are 
available for $\beom$.  The different possible choices
are obtained from transformations that alter the decomposition
(\ref{eqn:thdecomp}) of the presymplectic potential, while leaving $L'$ and
$\theta'$ invariant.
Such transformations induce changes in the flux, corner and and boundary terms
of the form\footnote{By composing with a JKM transformation
(\ref{eqn:JKMtrans}) one can
obtain an alternative form of boundary canonical transformations
in which $L' \to L' + d B$ and $\theta' \to \theta' + \delta B - d \Lambda$ 
while $\ell'$ and $\beta'$
are invariant. However the form (\ref{eq:bct}) is more general
since it naturally accommodates transformation parameters $B$ and
$\Lambda$ that are discontinuous from one boundary component to the
other, while $L'$, $\theta'$ and the JKM parameters $a$ and $\nu$ are
required to be continuous.  See the discussion after Eq.\ (\ref{eqn:Omcorner}) above.}
\begin{subequations}
\label{eq:bct}
  \begin{eqnarray}
    \label{eq:boundarycanonicaltransformation}
    \beom&\ra&\beom +\delta B -d\Lambda, \\
    \ell'&\ra&\ell' + B, \\
    \beta'&\ra&\beta' + \Lambda.
\end{eqnarray}
\end{subequations}
We refer to such transformations as a {\it boundary canonical transformations},
since they affect the division of the canonical pairs that appear in $\beom$ into
coordinates and momenta.\footnote{See \cite{Papadimitriou:2010as}
for a related discussion interpreting such transformations
as a canonical transformation in the context 
of holographic renormalization.}
For example, starting with a Dirichlet 
flux on a timelike boundary $\beom = \pi^{ij}\delta h_{ij}$, the boundary
canonical transformation with $B = -\pi^{ij}h_{ij}$ yields a Neumann form
of the flux, $\beom_N = \beom -\delta(\pi^{ij}h_{ij}) = -h_{ij}\delta \pi^{ij}$.
Quantities that were invariant under the JKM transformations considered above
transform nontrivially under these boundary canonical transformations;
in particular, the action, symplectic form, and charge density change
according to
\begin{subequations}
\begin{align}
S &\ra S+\int_{\ns} B \\
\Omega &\ra \Omega -\int_{\partial\Sigma} \delta \Lambda \\
h_\xi &\ra h_\xi +i_\xi B -I_\h\xi \Lambda.
\end{align}
\end{subequations}
There are a number of situations where boundary canonical transformations
are relevant.  The most important example for the present work
is in considerations of holographic renormalization,
where the naive Dirichlet form for $\beom$ does not admit a finite limit to the 
asymptotic boundary.  In this case, one seeks to find a counterterm
$B = \ell_\text{ct}$ constructed from intrinsic quantities on the boundary
such that the resulting renormalized action is finite as the boundary is 
taken to infinity.  We show in section \ref{genarg} that this ensures that 
the renormalized flux $\beom_\text{ren}$ also has a finite limit, and hence is sufficient
to construct finite asymptotic charges.  Another example 
in which such boundary canonical transformations appear is in AdS/CFT
when considering the alternative quantization of low mass bulk fields
\cite{Klebanov:1999tb}. 

\subsection{Corner improvements} \label{sec:corners}

While specifying the form of the flux $\beom$ resolves the 
standard JKM ambiguities in the covariant phase space formalism, 
there is an additional ambiguity that remains even after fixing 
$\beom$.  This ambiguity occurs because the decomposition
(\ref{eqn:thdecomp}) determines $\ell'$ and $\beta'$ only up to shifts 
of the form
\begin{subequations}
\label{fshift}  
\begin{eqnarray}
\ell' & \ra & \ell' + df  \label{eqn:l'shift}\\
\beta' &\ra& \beta' + \delta f, \label{eqn:beta'shift}
\end{eqnarray}
\end{subequations}
with $f$ generically noncovariant.  Under such a shift, the 
charge density $h_\xi$ is not invariant, instead transforming as 
\beq
h_\xi \ra h_\xi -\Delta_\h\xi f - di_\xi f,
\eeq
and the term $\Delta_\h\xi f$ will affect the value of the integrated 
charge $H_\xi$.  A similar shift occurs in $h_\xi$ under the 
transformations of $b'$ and $\lambda'$ described in equations 
(\ref{eqn:b'ambig}) and (\ref{eqn:lambda'ambig}) by a noncovariant 
quantity $e$, leading to a shift in the charge density
\beq
h_\xi \ra h_\xi -\Delta_\h\xi e - di_\xi e.
\eeq

In order to handle these additional ambiguities, a correction
must be added to the charges that cancels the dependence on these shifts.  
This improvement term in the charges was described in appendix C of 
\cite{Chandrasekaran:2020wwn} when working with covariant $L'$ and $\theta'$
(so that $b'$ and $\lambda'$ are set to zero), and here we will describe the 
generalization of this procedure to generically noncovariant $L'$ and $\theta'$.

\begin{table}
\centering
\footnotesize
\begin{tabular}{|c|c|c||c|c|c|c||c|c|c|c|}
\hline
 & \makecell*{Transformation}  && $a \ \dagger$ & $\nu\ \dagger$ &
$e\ \dagger$ & $\chi\ \dagger$ &
$B$ &
$\Lambda$ &
$F$ &
$\zeta$ 
\\
\hline
 & \makecell*{Name}  && JKM & JKM &
 &  &
BCT &BCT &
CCT &
CCT 
\\
\hline
 & \makecell*{Equation} & & \ref{eqn:JKMtrans} & \ref{eqn:JKMtrans} &
\ref{eqn:echidef} & \ref{eqn:echidef} &
\ref{eq:bct} &
\ref{eq:bct} &
\ref{cct} &
\ref{cct} 
\\
\hhline{| = | = | = # = | = | = | = # = | = | = | =|}
 & \makecell*{Quantity} & Eqn. &  & &  & &  &  &  &  
\\ \hline
\makecell*{$L'\ \dagger$ } & bulk Lagrangian &\ref{eqn:dL} & $da$ &  & & &  & &  &
 \\
\hline
$\makecell*{\theta'} \ \dagger$ & symplectic pot.& \ref{eqn:dL} &$\delta a$  &$d\nu$
& &
& &  & & \\
\hline
\makecell*{$b' \ \dagger$} & $L'$ noncovariance&\ref{eqn:blambdadef}
&$a$ & &$de$ & & &  & &  \\
\hline
\makecell*{$\lambda' \ \dagger$} & $\theta'$ noncovariance & \ref{eqn:blambdadef} & 
& $\nu$& $-\delta e$&$d \chi$ & & &    &  \\
\hline
\makecell*{$\ell'$ } & boundary action & \ref{eqn:thdecomp}&$-a$ &  &&  & $B$&
 & & \\
\hline
\makecell*{$c'$} & corner action &\ref{eqn:bldecomp} & &  &$-e$ & & &   &$-F$
&  \\
\hline
\makecell*{$\gamma'$ }& & \ref{eqn:bldecomp} &   &
& &$-\chi$ & & & &$\zeta$   \\
\hline
\makecell*{$\varepsilon$ }& corner flux &\ref{eqn:bldecomp} &  & &
& &  &$\Lambda$ &$-\delta F$ & $-d\zeta$ \\
\hhline{| = | = | = # = | = | = | = # = | = | = | =|}
\makecell*{${\beom}$ }&boundary flux &\ref{eqn:thdecomp} & &  & & &$\delta B$ &$-d\Lambda$ & &  \\
\hline
\makecell*{$S$} & action & \ref{eqn:action}&
 &   & & & $\int_{\cal N} B$& &$-\int_{\partial N} F$
 &  \\
\hline
\makecell*{$\Omega$} & sympl. form & \ref{eqn:Omcorner}&
&& &  && $-\int_{\partial \Sigma} \delta \Lambda$ &    &  \\
\hline
\makecell*{$\beta'$ }& corner term & \ref{eqn:thdecomp} & & $\nu$ & & & &$\Lambda$  & &  \\
\hline
\makecell*{$J_\xi'$} & Noether current & \ref{eqn:noethercurrent}& $d i_\xi a$ & $d I_{\hat \xi} \nu$
&$ - \Delta_{\hat \xi} de$ & & & & &    \\
\hline
\makecell*{$Q_\xi'$} &Noether charge & \ref{eq:Qprimexi}& $i_\xi a$ &  $I_{\hat \xi} \nu$
&$-\Delta_{\hat \xi} e$ &$d I_{\hat \xi} \chi$ & & & &    \\
 & &  &  &
&$ - d i_\xi e$ & & & & &    \\
\hline
\makecell*{$h_\xi$} &localized charge &\ref{eqn:Hxi} & &  &$-\Delta_{\hat \xi} e $  &$d I_{\hat \xi} \chi$ &$i_\xi B$ &$-I_{\hat \xi} \Lambda$ &
&   \\
 & & & &  &$ -
d i_\xi e$  & & & &
&   \\
\hline
\makecell*{ ${\tilde h}_\xi$ }& improved charge & \ref{eqn:hxiimproved}& &  &$-d i_\xi e$ & $d
I_{\hat \xi} \chi$ &$i_\xi B$ &$-I_{\hat \xi} \Lambda$ & $\Delta_{\hat
  \xi} F$& \\
\hline 
\end{tabular}
\caption{
A summary of the various transformations in the covariant phase space
formalism and how they act on the differential forms.
The first row lists the
eight different independent transformations, the second row their
names, and the third the equation numbers where the transformations
are defined.  The name acronyms are Jacobson-Kang-Myers (JKM) transformation,
boundary canonical transformation (BCT), and corner canonical
transformation (CCT).  The remaining rows list the various quantities
that occur in the formalism, their names, defining equations, and how
they transform under the transformations.  The first eight rows list
the fundamental independent quantities, while the remaining quantities are
derived from the first eight. Quantities indicated with a $\dagger$
are required to be continuous from one component of the boundary to
another, while those without this symbol may be discontinous at these
transitions. These discontinuities are associated with the appearance
of corner terms in the integrated action (\ref{ffs}).  The five
transformations in the columns for  $a,\nu,e, \chi$ and $\zeta$ 
do not change the
integrated action $S$, symplectic form $\Omega$, or integrated (improved)
localized charges, and thus are analogous to gauge freedom in the
formalism.
By contrast, the three transformations in 
the columns for 
$B, \Lambda$ and $F$ do alter these quantities, reflecting the fact that the boundary flux
${\cal E}$ and corner flux $\varepsilon$ must be specified (for
example via a complete action principle) in order to
determine unique localized charges.
}
\label{tab:operations}
\end{table}

\normalsize

The resolution comes from noting that we must not only fix the form of the 
flux $\beom$ on the bounding hypersurface $\ns$, but also must fix a preferred 
corner flux on the codimension-2 surface $\partial\Sigma$ on which 
the charge is being evaluated.  In this case, the quantity $\beta'-\lambda'$
serves as a higher codimension symplectic potential, and hence to 
resolve the ambiguities, we decompose  it in a similar manner as $\theta'$ from 
equation (\ref{eqn:thdecomp}): 
\beq \label{eqn:bldecomp}
\beta' - \lambda' = - \delta c' +d\gamma' + \cflx,
\eeq
where $\cflx$ is the corner flux.  
We will obtain unambiguous 
charges by specifying a preferred form of $\cflx$, which could be determined 
by a Dirichlet variational principle for a subregion of spacetime that 
includes corners, as discussed, for example, in \cite{Lehner_2016}.
In this case, $c'$ is related to the corner terms one adds to the action,
although the full action must include terms coming from both hypersurfaces 
intersecting at the corner.\footnote{In slightly more detail, we consider a region $\sr$
bounded by two hypersurfaces $\ns^+$ and $\ns^-$ intersecting at a codimension-2 corner $\cn$,
oriented such that $\partial \sr \supset \ns^+ - \ns^-$,  $\ns^+\supset \cn$, and $\ns^-\supset \cn$,
where the signs indicate the relative orientations.  The action including contributions from only these 
boundaries is given by
\beq
S = \int_\sr L' + \int_{\ns^+} \ell_+' - \int_{\ns^-}\ell_-' +\int_\cn(c_+'-c_-').
\eeq
Here, $\ell_\pm'$ and $c_\pm'$ arise from independent decompositions on $\ns^\pm$, and these 
quantities, along with $\beta'_\pm$, 
need not be continuous when moving from $\ns^-$ to $\ns^+$ through $\cn$.
On the other hand, $b'$ and $\lambda'$ should be continuous across $\cn$, since they arise from
$L'$ and $\theta'$ which are continuous throughout $\sr$. Invariance under the standard
JKM transformations follows as before, and we can also check invariance under the 
$e$ and $f$ ambiguities described in equations 
(\ref{eqn:b'ambig}), (\ref{eqn:lambda'ambig}), 
(\ref{eqn:l'shift}), and (\ref{eqn:beta'shift}).  For these ambiguities,
the quantity $e$ is required to be continuous through $\cn$, but $f$ can take on separate
values $f_\pm$ at $\cn$.  This then implies that the action is invariant under these 
transformations,
\beq
S\ra S +\int_{\ns^+} df_+ - \int_{\ns^-} df_- + \int_\cn(-e-f_+ +e+f_-) = S.
\eeq
The corner improvement in the present section only considers contributions from
the single hypersurface $\ns^+$ ending on $\cn$, but it would be interesting to extend
this analysis to account for the contributions from $\ns^-$.  
\label{ftn:corneraction} } 
The quantity $\gamma'$ can be viewed as a codimension-2 
symplectic potential, and in principle we could further consider decomposing it in a similar
manner to $\theta'$ and $\beta'-\lambda'$.  Doing so would yield quantities associated contributions
to the action and flux
associated with codimension-3 defects in the shape of the subregion.  Such features would
arise at caustics of a null hypersurface, and also when considering singular diffeomorphisms 
such as superrotations that produce defects on a codimension-2 surface
\cite{Barnich:2009se, Barnich:2011ct, Strominger:2016wns, 
Compere:2016jwb, Adjei:2019tuj, Donnay:2020guq}.  We note, however, in
the absence of such codimension-3 features, the quantity $\gamma'$ drops out of any
expression for the charges, and hence we will not consider it further in this work, although
a careful analysis of this type of term would be an interesting future direction.

The quantity $\Delta_\h\xi(\beta'-\lambda')$ appears in the
identity (\ref{eqn:IxiOm}), and the decomposition (\ref{eqn:bldecomp})
motivates including the $c'$ term in
the localized charge as opposed to the flux.  The improved charge density
is then defined to be 
\begin{align}
\tilde h_\xi &= h_\xi - \Delta_\h\xi c' \label{eqn:hxiimproved} \\ 
&= \mcov{Q}_\xi  + i_\xi(\ell'+b'+dc') -I_\h\xi \cflx + d(i_\xi c'-I_\h\xi\gamma')
\label{eqn:htilde2}
\end{align}
where the expression in the second line follows from applying the definitions (\ref{eqn:Hxi}), (\ref{eq:Qprimexi}),
and (\ref{eqn:bldecomp}).
The improved flux that generalizes equation (\ref{eqn:Axi}) is given by
\beq \label{eqn:tildeflx}
\tilde \flx_\h\xi = \int_{\partial\Sigma}\Big( i_\xi \beom
-\Delta_\h\xi \cflx + \tilde h_{\delta\xi}\Big).
\eeq
Defining the improved  localized charge $\tilde H_\xi$ as the integral
over $\partial\Sigma$ of $\tilde h_\xi$, we find that improved charges 
and fluxes still satisfy the modified Hamilton's equation,
\beq
\delta \tilde H_\xi = -I_\h\xi\Omega + \tilde \flx_\h\xi.
\eeq

Once a preferred form for the corner flux $\cflx$ is chosen, 
the shifts in $\beta'$ and $\lambda'$ described in (\ref{eqn:beta'shift})
and (\ref{eqn:lambda'ambig}) requires that $c'$ transform according to
\beq
c'\ra c' -f-e.
\eeq
It follows immediately that the improved charge density (\ref{eqn:hxiimproved})
shifts only by exact terms under the transformation, and hence the integrated
improved charge $\tilde H_\xi$ is invariant.\footnote{A slightly different proposal
for a corner-improved charge was recently considered in \cite{Odak:2021axr},
which amounts to defining the improved charge density to be $\bar h_\xi = h_\xi+\lie_\xi c'
= \tilde h_\xi +I_\h\xi \delta c'$.  We note that this alternative proposal does not
have the same invariance properties under the ambiguities as does $\tilde h_\xi$,
which serves as an argument in favor of (\ref{eqn:hxiimproved}).}

The corner flux $\varepsilon$ in Eq.\ (\ref{eqn:bldecomp}) can be
shifted by exact terms and total variations, leaving the left hand side $\beta'-\lambda'$ fixed.     
The transformations that achieve this are analogous to the boundary
canonical transformations (\ref{eq:bct}), but arise in the codimension-2
context rather than in codimension 1.  We call these transformations
{\it corner canonical transformations}, given that
they change the form of $\cflx$.
One type of corner canonical transformation is
an adjustment of the decomposition (\ref{eqn:bldecomp}) by 
$\gamma' \to \gamma' + \zeta$, $\varepsilon \to \varepsilon - d \zeta$,
leaving all other quantities fixed.  A second type is 
a transformation (\ref{fshift}) with parameter $f = F$,
followed by a boundary canonical transformation
(\ref{eq:bct}) with parameters $B = -dF$ and $\Lambda = -\delta F$.
Under the combined transformations we have
  \be
  \label{cct}
  c' \to c' - F, \ \ \ \ \ \gamma' \to \gamma' + \zeta,
  \ \ \ \ \ \varepsilon \to \varepsilon - \delta F - d \zeta,
  \ee
while $\ell'$ and $\beta'$ are invariant and ${\tilde h}_\xi \to
{\tilde h}_\xi + \Delta_{\hat \xi} F$.
This combination of transformations is designed to leave $\theta'$ invariant.
We will make use of these corner canonical transformations in our discussion of
holographic renormalization in Sec.\ \ref{sec:hr} below.

We emphasize that in our formalism the charges are uniquely determined by a choice of subregion action principle.  The various
canonical transformations considered here coincide with a 
change in action principle and a corresponding change in the charges.
A small subtlety related to this point occurs 
in regard to the effect of the boundary canonical
transformation $\beom\ra\beom-d\Lambda$.  For consistency,
this transformation must shift the corner flux
according to $\varepsilon\ra\varepsilon + \Lambda$ (see Table \ref{tab:operations}),
and these combined transformations have the 
effect of leaving the variation $\delta S$ of the 
subregion action invariant.  One might be tempted to 
conclude that the corner-improved charges should 
then be invariant under this transformation since 
the subregion action is meant to uniquely determine the 
charges; however, according to 
Table \ref{tab:operations}, this transformation in
fact shifts the charge density $\tilde h_\xi$ nontrivially.
This suggests that one must not only specify the form of 
the subregion action, but also the full form of the 
fluxes $\beom$ and $\varepsilon$, in order to 
obtain unique expressions for the charges.  
In actuality, both $\beom$ and $\varepsilon$ can be 
uniquely extracted from the action provided one 
specifies on which quantities the on-shell action
functionally depends.  In particular, for a Dirichlet
action principle, the on-shell action is a functional
of the boundary induced metric $h_{ij}$ and the corner
induced metric $q_{AB}$.  This uniquely determines the 
momenta $\pi^{ij} = \frac{\delta S}{\delta h_{ij}}$ and 
$\pi^{AB} = \frac{\delta S}{\delta q_{AB}}$, and hence 
the fluxes by the relation $\beom  = \pi^{ij}\delta h_{ij}$
and $\varepsilon = \pi^{AB}\delta q_{AB}$.  Hence,
even though the charges depend on the precise forms 
of the fluxes, we see that these are ultimately
determined in terms of the subregion action principle.

A summary of the various transformations we have defined in this
section is given in
Table \ref{tab:operations}.

\section{Brackets of localized charges}
\label{sec:algebra}

With the definition (\ref{eqn:Hxi}) of the localized charges in hand, we would next like
to compute the algebra they satisfy.
Since these charges arise from the action of diffeomorphisms on a subregion of a 
spacetime manifold, we should expect the charge algebra to be closely related to the 
algebra satisfied by the corresponding vector fields $\xi^a$
on spacetime under the Lie bracket.  
Diffeomorphisms of spacetime induce an action on the solution
space $\sls$ which serves as a phase space of our theory, leading to a related Lie
bracket of the vector fields $\h\xi$ associated with the spacetime vector fields $\xi^a$.
As mentioned in equations (\ref{eqn:fsbrack}) and (\ref{eqn:brmod}), the field
space bracket is simply minus the spacetime Lie bracket for field-independent 
generators, and for field-dependent generators, it is given by minus the modified 
Lie bracket $\brmod{\cdot}{\cdot}$.  

Normally when dealing with Hamiltonian charges
for a symplectic manifold, the Poisson bracket of the charges can be obtained 
by contracting the vector fields generating the symmetry into the symplectic form.  
However, the localized charges do not satisfy Hamilton's equation due to the term 
involving $\flx_\h\xi$
 in (\ref{eqn:Hammod}).  This means that the charges $H_{\xi}$ do not generate
the same flow as the vector $\h\xi$ on the phase space.  
Nevertheless, 
the charges $H_{\xi}$ are functions on phase space, and hence possess a well-defined 
Poisson bracket.   
We will find in this section that this Poisson bracket on the 
subregion phase space reproduces the bracket introduced by Barnich and Troessaert 
in \cite{Barnich:2011mi}, providing a novel derivation of this bracket and justifying 
its use in defining the algebra of localized charges.
Due to the difference between the flows generated by $H_\xi$ and $\h\xi$, we will find
that the charge algebra differs from the algebra of vector fields under the 
modified bracket by the extension terms $K_{\xi,\zeta}$ appearing 
in equations (\ref{eqn:HHbrack}) and (\ref{eqn:Kxizeta}).

To compute the Poisson bracket, it is first helpful to introduce an abstract index notation for 
tensors on field space.  We let $\n A, \n B,\ldots$ denote tensor indices on 
$\slsu$, so that, for example, the symplectic form on phase space can be written
$\Omega_{\n A \n B}$.  We then define $\Omega^{\n A\n B}$ to be an inverse of 
$\Omega_{\n B\n C}$.  The meaning of this statement is somewhat subtle because 
$\Omega_{\n B\n C}$ is degenerate on $\slsu$, and so inverting it requires some
gauge-fixing procedure to define the subspace of $\slsu$ on which we are constructing
the inverse.  This gauge fixing will yield a tensor $\Omega^{\n A \n B}$ satisfying
\be
  \Omega_{\n A \n B} \Omega^{\n B \n C} \Omega_{\n C \n D} = \Omega_{\n A \n D}
  \label{standard}
\ee
Note that the true subregion phase space $\psu$ is obtained from $\slsu$ by 
quotienting out the degenerate directions, and on this quotient space we expect 
$\Omega^{\n A \n B}$ to descend to a well-defined inverse that is independent of the 
gauge-fixing procedure.  We also assume that the vector fields $\h \xi^{\n A}$ 
have been constructed to be tangent to the gauge-fixed submanifold, so that 
$\Omega^{\n A \n B} \Omega_{\n B \n C} \h\xi^{\n C} = \h\xi^{\n A}$.  Often, this requirement introduces field
dependence into the vector $\xi^a$, which is one of the reasons
for considering field-dependent symmetry generators.
 
The Poisson bracket of the localized charges is defined to 
be\footnote{The sign in the definition of the Poisson bracket here is the more common choice, which is 
opposite to the sign employed in \cite{Chandrasekaran:2020wwn}. } 
\beq
\{ H_\xi , H_\zeta\} = \Omega^{\n A\n B}(\delta H_\xi)_{\n A} (\delta H_{\zeta})_{\n B}.
\eeq
Then since the variation of the localized charges satisfies (\ref{eqn:Hammod}), we find
for the Poisson bracket
\begin{align}
\{ H_\xi , H_\zeta \}
&=
\Omega^{\n A\n B}\left(\Omega_{\n A \n C} \h\xi^{\n C} + (\flx_\h\xi)_{\n A} \right)
\left( \Omega_{\n B\n D} \h\zeta^{\n D} + (\flx_\h\zeta)_{\n B}\right)
\\
&= 
-\h \xi^{\n B} \h\zeta^{\n D} \Omega_{\n B \n D}- \h\xi^{\n B}(\flx_\h\zeta)_{\n B}
+\h\zeta^{\n A} (\flx_\h\xi)_{\n A}+ \Omega^{\n A\n B} (\flx_\h\xi)_{\n A}
(\flx_\h\zeta)_{\n B}
\\
&=
\{ H_\xi, H_\zeta\}_\text{BT} +  \Omega^{\n A\n B} (\flx_\h\xi)_{\n A}
(\flx_\h\xi)_{\n B} \label{eqn:PB}
\end{align}
where in the last line we have introduced the Barnich-Troessaert bracket 
\beq \label{eqn:BTbrack}
\{ H_\xi , H_\zeta\}_{\text{BT}} = 
I_{\h\xi}I_{\h\zeta} \Omega - I_{\h\xi} \flx_\h\zeta + I_{\h\zeta} \flx_\h\xi
=-I_{\h\xi}\delta H_\zeta +I_{\h\zeta}\flx_\h\xi.
\eeq

This bracket was proposed in \cite{Barnich:2011mi} as a means for
defining an algebra for localized charges that only satisfy the modified version
of Hamilton's equation, (\ref{eqn:Hammod}), and it was left open the problem of 
interpreting it in terms of a Poisson bracket on a phase space.
From equation
(\ref{eqn:PB}), we see that the BT bracket in fact coincides with the ordinary 
Poisson bracket of the charges $H_\xi$, $H_\zeta$, provided we can argue that the 
final term quadratic in the fluxes $\flx_\h\xi$, $\flx_\h\zeta$ vanishes.  
To see how this occurs, we first assume that the corner term $\beta'$ that appears 
in (\ref{eqn:thdecomp}) is covariant, $\Delta_{\h\xi}\beta' = 0$, that the generators 
are field-independent $\delta \xi^a = 0$, and that the flux has been put into Dirichlet
form, as  in (\ref{eqn:timeD}) for a timelike surface or (\ref{eqn:nullD}) for a null surface.
In this case, the symplectic flux $\flx_\h\xi$ for a timelike surface simply reduces to
\beq \label{eqn:AxiD}
\flx_\h\xi = \int_{\partial\Sigma} i_\xi \pi^{ij} \delta h_{ij} 
\eeq
In this form, we see that the final term in (\ref{eqn:PB}) involves the contraction
of $\Omega^{\n A \n B}$ into two metric variations $\delta h_{ij}(x) \delta h_{kl}(x')$ 
at each pair of points $x, x'$ on the spatial codimension-2 surface $\partial \Sigma$.
Hence, it is determined entirely in terms of the Poisson bracket of the intrinsic metric
on the surface $\{h_{ij}(x), h_{kl}(x')\}$:
\beq
\int dx \int dx' i_\xi\pi^{ij}(x) i_\zeta\pi^{kl}(x') \{ h_{ij}(x), h_{kl}(x')\}.
\eeq
However, this bracket should vanish
on general grounds, since it involves values of the induced metric (without time derivatives)
at causally separated points on $\partial\Sigma$.
Additionally, at coincident points $x=x'$, no delta functions
should appear in the Poisson bracket due to the absence of time
derivatives.  
Because of this, we conclude that the 
final term (\ref{eqn:PB}) vanishes, and hence the Poisson bracket of the localized charges
agrees with the BT bracket, at least in the case that the symplectic flux 
has been reduced to the form
 (\ref{eqn:AxiD}).\footnote{Interestingly, reference \cite{Freidel:2020svx} found that commutativity 
 of the intrinsic metric on a codimension-2 slice of the boundary is violated in the presence
of a nonzero Immirzi parameter when utilizing the first order formulation of general
relativity.  This suggests that the naive BT bracket would be modified in this case, and it 
would be interesting to investigate these corrections in more detail.  }
 The story is entirely analogous for a null surface, and similarly relies on the vanishing
of the brackets between intrinsic quantities at the cut $\partial\Sigma$, $\{q_{ij}, q_{kl}\}
= \{n^i, q_{jk}\} = \{n^i, n^j\} = 0$.
For example, for the components $\delta q_{AB}$ of the induced metric
perturbation on the future horizon of a Schwarzschild black hole,
Ref.\ \cite{Hawking:2016sgy} derives the commutators
\be
\label{schw}
\{  q_{AB}(\p\theta,v),  q_{CD}(\p\theta',v') \} \propto (
q_{AC} q_{BD} + q_{AD} q_{BC} - q_{AB} q_{CD})
\delta^2(\p\theta,\p\theta') \left[ \Theta(v-v') - \frac{1}{2} \right].
\ee
where $(\p\theta,v) = (\theta,\phi,v)$ are coordinates on the
horizon.  This commutator vanishes\footnote{That $\Theta(v-v')-1/2$ should be interpreted as $0$ for $v=v'$ can be seen by integrating Eq.\ (\ref{schw}) against $w^{AB}(\p\theta) w^{CD}(\p\theta')$ for some $w^{AB}$ which yields at $v=v'$ the commutator of an operator with itself.}
at $v = v'$.

We can now ask whether any of the conditions leading to this conclusion can be relaxed.  
We can allow for the additional terms 
$\Delta_{\h\xi}(\beta'-\lambda')$ and $h_{\delta\xi}$ that appear in (\ref{eqn:Axi}),
provided that these also can be put into Dirichlet form.  For the corner term $\beta'-\lambda'$,
 this can be done using the corner improvement procedure described in section
 \ref{sec:corners} by selecting a Dirichlet form for the corner flux 
 $\cflx$ appearing in (\ref{eqn:bldecomp}).  For the $h_{\delta\xi}$ term, 
 this likely leads to some restrictions on the allowed field dependence
 in $\xi^a$.\footnote{Once the field-dependence
 of $\xi^a$ has been fixed, one could decompose $h_{\delta \xi} = -\delta a
 + d\tau + \epsilon$, and include the $a$ contribution in the charge and 
 $\epsilon$ in the flux.  Such a decomposition will lead to additional 
 modifications of the brackets of the charges.  
 This kind of decomposition has been investigated recently
 in \cite{Adami:2020ugu,Ruzziconi:2020wrb, Adami:2021sko,Geiller:2021vpg,Grumiller:2021cwg,Adami:2021nnf}.} 
 Finally, we can even allow for choices of the flux term
 $\beom$ that are not Dirichlet form, for example, we could instead require 
 Neumann form $\beom = -h_{ij}\delta\pi^{ij}$.  As long as the flux $\beom$ 
 is of the form where the equation $\beom = 0$ defines a valid boundary condition
 for the variational principle,\footnote{This can equivalently 
 be phrased as finding a Lagrangian submanifold
 for the boundary phase space involving the symplectic pairs
 $(\pi^{ij},h_{ij})$. } and the condition of vanishing symplectic
 flux $\flx_\h\xi = 0$ imposes no further restrictions, one will still be able to 
 argue that the flux terms in (\ref{eqn:PB}) vanish.  

The interpretation given here of the Barnich-Troessaert bracket as 
an ordinary Poisson bracket on the subregion phase space is a novel
proposal of the present work, and can be compared to previous 
arguments for arriving at this definition of the bracked for localized
charges $H_\xi$.  In \cite{Chandrasekaran:2020wwn}, two of us suggested a
heuristic derivation of the bracket, in which the bracket represented the Poisson
bracket on a larger phase space consisting of the subregion and a complementary phase
space that collects the flux, yielding a closed global phase space.  This interpretation
is similar to the one presented in the present work, but differs in that our 
present proposal shows that no auxiliary system is needed to interpret the 
bracket as a Poisson bracket.  It is likely the two proposals are consistent with
each other, after employing a gluing construction as discussed in
section \ref{sec:gluing}.  Another proposal by Troessaert
\cite{Troessaert:2015nia}
suggested an interpretation in which the boundary fields on which Dirichlet 
conditions would be imposed in the variational principle are interpreted as 
classical sources for the subregion phase space, motivated by similar interpretations
appearing in AdS/CFT.  This interpretation appears to be close in spirit to the 
proposal in the present paper; however, Troessaert's construction involves an explicit
decomposition of the bracket into an ordinary bulk Poisson bracket and corrections 
involving variations of the boundary sources.  This makes comparison to the present 
interpretation difficult, although it would be interesting to further investigate
whether the two proposals are consistent with each other.  Finally, we mention
some recent work by Wieland \cite{Wieland:2021eth}
in which the bracket arose as a Dirac bracket after 
constraining the phase space to remove all radiative modes from the theory.  
By contrast, the bracket in the present work imposes no such constraint, 
and hence disagrees with Wieland's proposal.  However, it may be that the two
proposals agree for a specific choice of transformations and charges that are 
``purely Coulombic,'' as might be expected for charges associated with diffeomorphisms
acting near the boundary.  

Having argued that the BT bracket coincides with the Poisson bracket of the 
localized charges, we can use it to compute the canonical algebra satisfied by 
these functions on the local phase space.  A straightforward calculation (see 
appendix \ref{sec:phase-space-calc}) using 
the bracket definition (\ref{eqn:BTbrack}) and the expressions (\ref{eqn:Hxi})
and (\ref{eqn:Axi}) for the charge density and symplectic flux
 yields the following charge representation theorem:
\begin{align}
\{ H_\xi, H_\zeta\} &= \left(H_{\brmod{\xi}{\zeta} } + K_{\xi,\zeta}\right)
\label{eqn:HHbrack}\\
K_{\xi,\zeta}  &= \int_{\partial\Sigma}\left(i_\xi \Delta_{\h\zeta}(\ell'+b') 
- i_\zeta\Delta_{\h\xi} (\ell'+b')\right). \label{eqn:Kxizeta}
\end{align}
Hence, we see that the bracket reproduces the algebra of the vector fields
$\xi^a$ under the modified bracket (\ref{eqn:brmod}), up to an extension
parameterized by a new set of generators $K_{\xi,\zeta}$.  
The combination $\ell'+b'$ that appears in the formula for the extension
is a JKM-invariant quantity, and reduces to the expression for the extension
in Ref.\ \cite{Chandrasekaran:2020wwn} when utilizing a covariant 
Lagrangian $L'$.
Note that equation (\ref{eqn:mxizeta}) indicates that the extension would involve
an additional contribution $i_\xi i_\zeta(L'+d\ell')$, except for the fact that 
we have assumed that $\xi^a$ and $\zeta^a$ are both parallel to the 
same hypersurface, causing this term to pull back to zero.  In fact, this 
term was first derived in Ref.\ \cite{Speranza2018a} 
assuming boundary conditions
to make such transformations integrable, and was also recently explored 
in \cite{Freidel:2021cjp}.  Assuming charges associated with
the two independent normal directions to $\partial\Sigma$ can be consistently
defined \cite{Speranza2018a, Freidel:2021cjp, Ciambelli:2021vnn}, 
this suggests the full formula 
for the 
extension is given by the sum of (\ref{eqn:Kxizeta}) and the integral 
of $i_\xi i_\zeta(L'+d\ell')$.

The generators $K_{\xi,\zeta}$ are  local functionals of the vector
fields $\xi^a, \zeta^a$ and the geometric quantities defined on the boundary,
and we can therefore compute their brackets with the original localized
charges using a similar set of arguments as above:
\begin{align}
\{H_\xi, K_{\zeta,\psi} \} = \Omega^{\n A \n B} \big(\delta H_\xi\big)_{\n A}
\big(\delta K_{\zeta,\psi}\big)_{\n B}
=-\h\xi^{\n B} \big(\delta K_{\zeta,\psi}\big)_{\n B} + \Omega^{\n A \n B} 
\big(\flx_{\xi}\big)_{\n A} \big(\delta K_{\zeta,\psi} \big)_{\n B}.
\label{eqn:HK}
\end{align}
To simplify this further, we postulate that $K_{\xi,\zeta}$ 
is a functional of only intrinsic variables on the surface, $K_{\xi, \zeta} =
K_{\xi,\zeta}[h_{ij}]$ (including any field-dependence in the vectors 
$\xi^a$, $\zeta^a$).  It can be checked that this condition is satisfied in 
general relativity with a finite null boundary
\cite{Chandrasekaran:2020wwn}, and we expect it to hold more 
generally for theories that admit a Dirichlet variational principle.  
Its variation can therefore be written 
\beq \label{eq:deltaK}
\delta K_{\xi,\zeta} = \int_{\partial\Sigma}( k^{ij}_{\xi,\zeta} \delta h_{ij} +d \sigma_\xi)
\eeq
and the exact piece $d\sigma_\xi$ integrates to zero on $\partial\Sigma$.  
On a null surface, we similarly require that $\delta K_{\xi,\zeta}$ 
involve only the variations $\delta q_{ij}$ and $\delta n^i$. 
As before, using 
the assumption that $\flx_\h\xi$ is in Dirichlet form and  that 
the intrinsic variables commute on $\partial\Sigma$, we see that the second
term in (\ref{eqn:HK}) drops out, and we derive the relation
\beq \label{eqn:HKbrack}
\{H_\xi, K_{\zeta,\psi}\} = -I_{\h\xi}\delta K_{\zeta,\psi}.
\eeq
Finally,  the assumption that $K_{\xi,\zeta}$ is a functional
of only intrinsic quantities leads by the same arguments to the result that 
the $K_{\xi,\zeta}$ generators commute among themselves,
\beq\label{eqn:KKbrack}
\{K_{\xi,\zeta}, K_{\psi, \chi}\} = 0.
\eeq
As before, if we are instead working with a flux $\beom$ that is not of Dirichlet form,
the commutation relations (\ref{eqn:HKbrack}) and (\ref{eqn:KKbrack}) will 
remain valid as long as $K_{\xi,\zeta}$ is a functional only of quantities that would 
be fixed by the variational principle associated with the chosen form of the flux. 

The  relations (\ref{eqn:HHbrack}), (\ref{eqn:HKbrack}), and (\ref{eqn:KKbrack}) fully define the algebra satisfied by the canonical charges $H_\xi$ and 
the extension charges $K_{\xi,\zeta}$.  In the case that $I_\h\xi\delta
K_{\zeta,\psi}$ can be expressed only in terms of the generators 
$K_{\chi, \rho}$, the charges $(H_\xi, K_{\zeta,\psi})$ yield
a representation of an abelian extension of the algebra 
satisfied by the vector fields under the bracket $\brmod{\xi}{\zeta}$.
This condition was confirmed, for example, for a class of vector 
fields satisfying a Witt algebra acting on Killing horizons
in \cite{Chandrasekaran:2020wwn}, where it further was demonstrated 
that only a single independent generator $K_{\xi,\zeta}$ arises, 
yielding a central extension, the Virasoro algebra.  In the most
general case, however, one would expect $I_\h\xi \delta K_{\zeta,\psi}$ 
to be expressed as a sum of $H_\xi$ and $K_{\chi,\rho}$, producing a
more complicated algebra, presumably related to $\text{Diff}(\ns)$ 
or $\text{Diff}(\sr)$, in which
$K_{\xi,\zeta}$ generate an abelian subalgebra.  
It would be interesting to explore these 
more complicated algebras in future work.

The requirement that $K_{\xi,\zeta}$ be a functional of intrinsic data on the boundary
is a nontrivial consistency requirement in order to conclude  
the algebraic relation (\ref{eqn:HKbrack}). To further motivate
it, we remark that 
this requirement can be related to a generalized notion of symmetry for the subregion.\footnote{We
thank Don Marolf for discussions on this point.}
Normally, symmetries are defined as transformations that leave the subregion
action  (\ref{eqn:action}) invariant.  However, we can also consider transformations
that change the action only by a boundary term that is constructed 
entirely from intrinsic data,
\beq
I_{\h\xi}\delta S = \int_{\ns} \anom_\xi.
\eeq
In the Dirichlet variational principle where the intrinsic data is fixed 
by a boundary condition, requiring $\anom_\xi$ to depend only on intrinsic
quantities then says that the action is invariant up to a constant.  In such a 
situation, one can still derive a Noether charge that generates the symmetry, and it 
is conserved up to quantities constructed from
the intrinsic geometry, which commute with all observables. The quantity
$I_{\h\xi}\delta S$
can be reexpressed using the anomaly operator as 
\beq
I_\h\xi\delta S = \int_{\ns}\Delta_{\h\xi}(\ell'+b') + i_\xi(L' + d\ell'),
\eeq
where we find both of the contributions that arise in the formula 
for the extension, as discussed below 
(\ref{eqn:Kxizeta}).\footnote{As a side consequence, this demonstrates that 
extensions appear in the bracket of canonical charges only for transformations
that do not leave the subregion action invariant.}  Restricting
to $\xi^a$ that is tangent to $\ns$, we see that $\anom_\xi$ and 
$\Delta_\h\xi(\ell'+b')$ coincide.  
Hence, $\Delta_\h\xi (\ell'+b')$ must be expressible 
purely in terms of intrinsic quantities to be consistent with this generalized 
symmetry principle.

The charge algebra (\ref{eqn:HHbrack}), (\ref{eqn:Kxizeta}) was derived for localized
charges $H_\xi$ constructed without employing the corner improvement described
in section \ref{sec:corners}.  When working the with improved charges $\tilde H_\xi$ 
defined as integrals of the improved charge density (\ref{eqn:hxiimproved}), the 
Poisson bracket of the charges is again given by the BT bracket (\ref{eqn:BTbrack}) after
replacing the flux terms involving $\flx_\h\xi$ to  the modified fluxes 
$\tilde \flx_\h\xi$, defined 
in (\ref{eqn:tildeflx}).  As before, the charge algebra reproduces the 
modified bracket algebra of the vector fields up to an extension (derived 
in Appendix \ref{sec:phase-space-calc}),
\begin{align}
\{\tilde H_\xi, \tilde H_\zeta\} &= 
\tilde H_{\brmod{\xi}{\zeta}} + \tilde K_{\xi,\zeta} \label{eqn:tilHtilH}\\
\tilde K_{\xi,\zeta} &= \int_{\partial\Sigma}\left(
i_\xi\Delta_\h\zeta(\ell'+b'+dc') - i_\zeta\Delta_\h\xi(\ell'+b'+dc')\right). \label{eqn:tilK}
\end{align}

An important property of the BT bracket is that it reduces to a Dirac bracket 
for the generators $H_\xi$ when boundary conditions are imposed to 
make them integrable, meaning they satisfy Hamilton's equation (\ref{eqn:Ham}) with 
no fluxes.  
For a flux in Dirichlet form, this boundary condition is just that the intrinsic
quantities on the surface are fixed.  In this case, since we also require that 
$\Delta_\h\xi (\ell'+b')$ is constructed purely from intrinsic quantities, the boundary
condition also imposes that $\Delta_{\h\xi}(\ell'+b')$ is constant on the 
constrained phase space, and hence $\delta K_{\xi,\zeta} = 0$.  In this case,
the vector fields generating the symmetry must be chosen to preserve 
the boundary condition, and we find that 
the generators $K_{\xi,\zeta}$ yield a central extension of the vector field
algebra, as required by general arguments 
\cite{Arnold1978, kirillov2004lectures, CFP}
on the properties of group actions on a symplectic manifold.

The more general setup considered in this work does not impose such a boundary
condition, and this allows for abelian extensions or more general 
forms of the algebra.  
The new generators $K_{\xi,\zeta}$ appearing in the extension are functionals
of the intrinsic geometry evaluated on a slice of the boundary $\ns$.  
It is helpful to view the collection of all such intrinsic functionals as 
forming an abelian algebra of boundary observables localized on the cut $\partial\Sigma$.
The charges $H_\xi$ then act on any such functional $f[h_{ij}]$, generating
its evolution along the vector $\xi^a$ just as 
in equation (\ref{eqn:HKbrack}), 
\beq
\{ H_{\xi}, f[h_{ij}] \} = -I_{\h\xi}\delta f[h_{ij}].
\eeq
Hence, even when $\ell$ is covariant so that the extension terms $K_{\xi,\zeta}$ in
(\ref{eqn:HHbrack}) vanish, we can still construct an extended algebra 
by allowing the generators $H_{\xi}$ to act on the intrinsic functionals
$f[h_{ij}]$.  These intrinsic functionals are reminiscent of the edge modes
arising in \cite{Takayanagi:2019tvn} when accounting for the Hayward term
in the gravitational subregion action.  It would be interesting to 
further explore this connection between edge mode degrees of freedom
and the action of the localized charges $H_\xi$ on 
functionals of intrinsic data.

\section{Vacuum general relativity at null infinity}
\label{sec:compendium}

In the formalism developed in the past two sections, we have
implicitly assumed that the boundaries ${\cal N}$ are at finite
locations in spacetime, and that all quantities are finite on those
boundaries.  In Sec.\ \ref{sec:hr} below, we will extend the formalism
to handle the case of asymptotic boundaries, treated in a conformal completion framework to
bring them to finite locations.  In this context, the Lagrangian and symplectic
potential can diverge at the boundaries and must be suitably
renormalized using the techniques of holographic renormalization
\cite{Balasubramanian:1999re,
DeHaro2001, Papadimitriou:2005ii, MM-term, Compere:2018ylh}.
In this section, we take a detour to provide 
a motivating example for our treatment of holographic
renormalizataion of Sec.\ \ref{sec:hr}: an analysis of various
asymptotic symmetry groups for vacuum
general relativity in asymptotically flat spacetimes.

As discussed in the introduction, in recent years a number of
different field configuration spaces have been suggested that modify
the boundary conditions imposed at $\scp$, and that give rise to
extensions of the Bondi-Metzner-Sachs (BMS) group of asymptotic
symmetries.  The BMS group arises when one defines a field
configuration space by fixing some of the diffeomorphism freedom on
the boundary.  However, some of the relevant linearized diffeomorphisms are
not degeneracy directions of the symplectic form, and thus do not
correspond to gauge degrees of freedom.  Hence it is a nontrivial restriction on the theory
to impose these conditions. Lifting some of the boundary conditions
leads to an enlarged symmetry group called the 
generalized BMS group
\cite{Compere:2018ylh,CL,Campiglia:2014yka,Flanagan:2019vbl}.

In this section, we first review
the field configuration space definitions which lead to the BMS group and
generalized BMS groups, using the language of \cite{CFP}. We then
further relax the
boundary conditions at null infinity so as to uncover an even bigger
symmetry group. This procedure allows us to uncover new boundary
degrees of freedom, or edge modes \cite{Donnelly2016a,Speranza2018a,Geiller:2017xad,Geiller:2019bti}.
The enlarged symmetry group
coincides with the symmetry group
on finite null surfaces derived in Ref.\ \cite{CFP}.  Following
Ref.\ \cite{freidel2021weyl} we will call this group the {\it Weyl BMS} group.
Some of the details of the analysis are relegated to Appendix \ref{sec:derivegroups}.

\subsection{Field configuration spaces}
\label{sec:fcs}

We describe asymptotically flat spacetimes using the conformal
completion framework, reviewed in Appendix \ref{sec:conf}.
Some of the relevant fields on spacetime are the physical metric
$\gphys_{ab}$, the conformal factor $\conf$, the unphysical metric
$\gunphys_{ab} = \conf^2 \gphys_{ab}$ and normal $n_a = \nabla_a
\conf$, while some of the fields on $\scp$ are the null generator $n^i$,
inaffinity $\nonaffinity$, and volume forms $\volume_{ijk}$ and $\volumesmall_{ij}$.
We restrict attention throughout to 3+1 dimensions.  Higher dimensional spacetime do posess
supertranslation symmetries \cite{Kapec:2015vwa} and associated charges\footnote{Ref.\ \cite{Hollands:2016oma}
argued for imposing boundary conditions that removes these supertranslation symmetries,
because of the divergence of the associated symplectic flux.  However, as  argued in 
in section \ref{sec:hr} of this paper,
such divergences can generically be addressed using holographic renormalization
and so should not be used as a criterion for determining which boundary conditions to impose.}
\cite{Aggarwal:2018ilg}, and it would be interesting to investigate
analogous extensions of 
those symmetry groups.

We now review a number of field configuration spaces
corresponding to different boundary conditions at $\scp$.
To start, we fix $\stm$ and $\scp$ and consider the set of all
asymptotically flat spacetimes $(\stm, \gunphys_{ab}, \Phi)$.
Since everything should be invariant under the conformal
transformation $(\gunphys_{ab},\conf) \to (e^{-2\sigma} \gunphys_{ab},
e^{-\sigma} \conf)$, it is convenient to fix the conformal freedom by
fixing a choice of conformal
factor $\conf_0$ on a neighborhood ${\cal D}$ of $\scp$,
with $\conf_0 =0$  and $\nabla_a \conf_0 \ne 0$ on $\scp$.
We now define the large configuration phase of all unphysical metrics with that conformal factor:
\be
\Gamma_0 = \left\{ (\stm,\gunphys_{ab},\conf) \right| \left.  \ \ \  \conf =
\conf_0 \  {\rm on} \ {\cal D},  \ \ \ \ {\tilde G}_{ab} = 0\ \right\}.
\label{Gamma0def}
\ee
This is the most general configuration space consistent with the
equations of motion.  All of the spaces we will consider will
correspond to subspaces of $\Gamma_0$ obtained by imposing specific
boundary conditions.

Consider first the BMS field configuration space
\cite{Wald:1999wa,Flanagan:2019vbl}.
We fix a conformal factor $\conf_0$ on a neighborhood
${\cal D}$ of $\scp$, fix a particular unphysical metric
$\gunphys_{0\,ab}$ on $\scp$, and define
\be
\Gamma_{\rm BMS} = \left\{ (\stm, \gunphys_{ab}, \conf) \, \right|
\left. \, \gunphys_{ab \, |\scp} =
\gunphys_{0\,ab \, |\scp},  \ \ \ \conf =
\conf_0\  {\rm on}\  {\cal D}, \ \ \
	 \nabla_{a}\nabla_{b}\conf|_{\scp}=0\, \right\}.
\label{eq:fcs}
  \ee
Here we have used the conformal freedom (\ref{conffreedom}) to fix the
conformal factor, imposed the Bondi condition (\ref{bondicondition})
to set $\nabla_a \nabla_b \Phi$ to zero on $\scp$,
and fixed the unphysical metric on $\scp$.
The original justification for imposing these conditions was that
the entire space $\Gamma_0$ can be obtained by taking the orbit of
$\Gamma_{\rm BMS}$ under diffeomorphisms and conformal transformations \cite{Wald:1999wa}.
However, not all of these diffeomorphisms are gauge in the sense of
corresponding to degeneracy directions of the symplectic form (see Sec.\ \ref{expcalc} below for more details), which 
has led to the recent consideration of enlarged configuration spaces
and weaker boundary conditions \cite{Stro-lectures}.
The enlargement leads to new degrees of freedom on the boundary, 
known as boundary gravitons or edge modes
\cite{Donnelly2016a,Speranza2018a,Geiller:2017xad,Geiller:2019bti}.

We next consider the generalized BMS\footnote{The terminology ``extended BMS group'' was
used in Ref.\ \cite{Flanagan:2019vbl}, but ``generalized BMS group'' seems to be more common.}
configuration space
of Campiglia and Laddha \cite{Compere:2018ylh,CL,Campiglia:2014yka}.
Here we fix a conformal factor $\conf_0$ on a neighborhood
${\cal D}$ of $\scp$, fix a volume form ${\bar \volume}_{0\,ijk}$ and
null generator
${\bar n}^i_0$ on $\scp$ that satisfy the identity (\ref{relations1c})
with $\nonaffinity=0$, and define \cite{Flanagan:2019vbl}
\be 
\Gamma_{\rm GBMS} = \left\{ (\stm,g_{ab},\conf) \right| \left. n^i =
      {\bar n}^i_0, \ \ \ \ \volume_{ijk} = {\bar \volume}_{0\,ijk},  \ \ \ \conf =
  \conf_0\  {\rm on}\  {\cal D}, \ \ \
	 \nabla_{a}\nabla_{b}\conf|_{\scp}=0\right\}.
\label{eq:fcs1}
\ee
Here the normal $n^i$ and volume form $\volume_{ijk}$ are understood
to be computed from $\gunphys_{ab}$ and $\conf$ as described in Appendix
\ref{sec:conf}.  Compared to the BMS configuration space
(\ref{eq:fcs}), we have replaced the unphysical metric evaluated on $\scp$ with the
volume form and the intrinsic normal.

The field configuration space can be further expanded by dropping the
volume form.  This yields the Weyl BMS field configuration space \cite{freidel2021weyl}
  \be
\Gamma_{\rm WBMS} = \left\{ (\stm,g_{ab},\conf) \right| \left. n^i =
      {\bar n}^i_0,  \ \ \ \conf =
  \conf_0\  {\rm on}\  {\cal D}, \ \ \
	 \nabla_{a}\nabla_{b}\conf|_{\scp}=0\right\}.
\label{eq:fcs11}
  \ee  
  Note that the three different configuration spaces we have defined are related by
  \be
  \Gamma_{\rm BMS} \subset \Gamma_{\rm GBMS} \subset \Gamma_{\rm WBMS},
  \ee
when ${\bar n}_0^i$ and ${\bar \eta}_{0\,ijk}$ are those computed from $g_{0\,ab}$ on $\scri$.

\subsection{Symmetry groups}
\label{sec:symgroups}

We now turn to describing the symmetry groups and algebras of the three
field configuration spaces (\ref{eq:fcs}), (\ref{eq:fcs1}), and
(\ref{eq:fcs11}).
The derivations of these groups are given in Appendix \ref{sec:derivegroups},
where we use the universal structure approach of Ashtekar
\cite{2014arXiv1409.1800A} and the techniques of Ref.\ \cite{CFP}.

We start by picking a convenient class of coordinate
systems on $\scp$.
Choose a cross section ${\cal C}$ of $\scp$ and choose coordinates
$\theta^A = (\theta^1, \theta^2)$ on ${\cal C}$.  Extend the
definition of $\theta^A$ to all of $\scp$ by demanding that $\theta^A$ be
constant along integral curves of ${\bar n}_0^i$.
Here for the spaces $\Gamma_{\rm GBMS}$ and $\Gamma_{\rm WBMS}$, ${\bar
  n}_0^i$ is the intrinsic normal that appears explicitly in the
definitions,
while for the BMS case (\ref{eq:fcs}), ${\bar n}_0^i$ is computed from
the metric $\gunphys_{0\,ab}$ on $\scp$ and from $\conf_0$.
Along each integral curve we
define a parameter $u$ by setting $u=0$ on ${\cal C}$ and demanding
that
\be
   {\vec {\bar n}}_0 = \partial / \partial u.
   \label{ucoorddef}
\ee
This construction yields a
coordinate system $y^i = (u,\theta^A)$ on $\scp$.

In this class of coordinate systems, the diffeomorphisms $\varphi :
\scp \to \scp$ have the following form for all 
three groups:
\bes
\label{diffeogen}
\bea
\label{barudef}
    {\hat u} &=& e^{\alpha(\theta^A)}  \left [ u + \gamma(\theta^A)
      \right],\\
    \label{barthetadef}
    {\hat \theta}^A &=& \chi^A(\theta^B),
\eea
\ees
where for a point ${\cal P}$ on $\scp$ we have defined $y^i =
y^i({\cal P})$ and ${\hat y}^i = y^i(\varphi({\cal P}))$.  
Equation (\ref{barthetadef}) defines a mapping $\chi$ from ${\cal C}$
to itself, or equivalently from the space of generators of $\scp$ to
itself.
This set of maps is isomorphic to the set 
${\rm Diff}(S^2)$ of diffeomorphisms of the two-sphere.
Writing the map (\ref{diffeogen}) as $(\alpha,\gamma,\chi)$, the group
composition law is
\be
(\alpha_2, \gamma_2, \chi_2 ) \circ (\alpha_1, \gamma_1, \chi_1 )  =
(\alpha_3, \gamma_3, \chi_3),
\ee
where
\bes
\label{groupcompos}
\bea
\alpha_3 &=& \alpha_1 + \alpha_2 \circ \chi_1, \\
\gamma_3 &=& \gamma_1 + e^{-\alpha_1} \gamma_2 \circ \chi_1, \\
\chi_3 &=& \chi_2 \circ \chi_1.
\eea
\ees
The linearized version of the mapping (\ref{diffeogen}) is
${\hat y}^i = y^i + \xi^i$, where the generator ${\vec \xi}$ is
\be \label{eq:vec-field0}
{\vec \xi} = \left[ \gamma(\theta^A) + \alpha(\theta^A) u \right]
\partial_u + \xi^A(\theta^B) \partial_A.
\ee

Using the notation (\ref{diffeogen}), the structure of the three
different groups can be summarized as follows:

\begin{itemize}

\item For the BMS group,
  as is well known, the function $\gamma$
  parametrizes supertranslations and can be freely specified.  The
  map $\chi$ and the function $\alpha$ are constrained by
  \be
  \chi^* {\bar \inducedmetric}_{AB} = e^{2 \alpha} {\bar \inducedmetric}_{AB}
  \label{alpha444}
  \ee
where $\chi^*$ is the pullback operator and $q_{AB}$ is the spatial metric.
It follows that $\chi$ is a global conformal isometry of the sphere,
of which there is a six parameter
family isomorphic to the Lorentz group, and $\alpha$ is determined by
$\chi$.
The group structure is the semidirect product
\be
   SO(1,3)  \ltimes {\cal S},
   \label{bmsstructure}
\ee
where ${\cal S}$ is the normal subgroup of supertranslations given by
$\alpha=0$, $\chi = $ identity, and the subgroup $\gamma=0$ is
isomorphic to the Lorentz group $SO(1,3)$.
Note that
the semidirect product in Eq.\ (\ref{bmsstructure}) has the property that
the supertranslation $\gamma$ transforms under the conformal
isometries of the two-sphere as a scalar density of weight $1/2$, 
$\gamma \to e^{-\alpha} \chi^* \gamma$, from Eqs.\ (\ref{groupcompos}) and
(\ref{alpha444}) \cite{Barnich:2021dta}.

\item For the generalized BMS group, the only
  change is that the six parameter group of conformal isometries is
  replaced by the infinite dimensionsal group ${\rm Diff}(S^2)$ of
  diffeomorphisms of the two-sphere.
  Thus, the supertranslation function $\gamma$ can be chosen freely as
  before, the diffeomorphism $\chi$ can be freely chosen, and the
  function $\alpha$ is determined as a function of $\chi$ by
    \be
  \chi^* {\bar \volumesmall}_{AB} = e^{2 \alpha} {\bar
    \volumesmall}_{AB},
  \label{alpha333}
  \ee
  where ${\bar \volumesmall}_{ij} = - {\bar \volume}_{ijk} {\bar
    n}^k$.
The group structure is the semidirect product
\be
   {\rm Diff}(S^2)  \ltimes {\cal S}.
   \label{genbmsstructure}
   \ee
The semidirect product here has the property that
the supertranslation $\gamma$ still transforms as a scalar density of
weight $1/2$,
\be
\gamma \to e^{-\alpha} \chi^* \gamma,
\label{scalardensity}
\ee
but now under all diffeomorphisms of the two-sphere instead of just
the conformal isometries.  This follows from 
Eqs.\ (\ref{groupcompos}) and (\ref{alpha333}).

\item Finally, for the Weyl BMS group, the constraint (\ref{alpha333}) that determines $\alpha$ as a
  function of $\chi$ is lifted, and $\alpha$ can now be freely chosen.
  This group is isomorphic to the group of symmetries on any null surface
  at a finite location defined in Ref.\ \cite{CFP}, as can be seen by
  comparing Eqs.\ (\ref{diffeogen}) here with Eqs.\ (4.7) of that paper.
 The group has the following structure.
  We define the subgroup ${\cal T}$ of all supertranslations to
  be given by $\chi = $ identity, which is parameterized by $\alpha$ and $\gamma$.
  This is a normal subgroup and the total group has the structure
\be
   {\rm Diff}(S^2)  \ltimes {\cal T}.
   \ee  
   The supertranslation group ${\cal T}$ contains two subgroups.
   First, there is the subgroup ${\cal S}$ given by $\alpha=0$,
   parameterized by $\gamma$.  These were called {\it affine
     supertranslations} in Ref.\ \cite{CFP} since they correspond to
   displacements in affine parameter.  This is a normal subgroup of
   both ${\cal T}$ and of the full group.  Second, there is the
   non-normal subgroup ${\cal W}$ given by $\gamma=0$, parameterized by $\alpha$.
   These were called {\it Killing supertranslations} in
   Ref.\ \cite{CFP} since they correspond to displacements in Killing
   parameter when there is a Killing vector field. They were also
   called Weyl rescalings in Ref.\ \cite{freidel2021weyl}, as
   mentioned earlier. The supertranslation group has the structure
   \be
      {\cal T} = {\cal W} \ltimes {\cal S},
    \ee
    so the full symmetry group can be written as
\be
   {\rm Diff}(S^2)  \ltimes ( {\cal W} \ltimes
   {\cal S} ).
\label{BMSWeylStructure}
\ee
Here the first semidirect product is such that the supertranslation
functions $\alpha$ and $\gamma$ transform as scalars
under diffeomorphisms of the
two-sphere, not as scalar densities,\footnote{However, if we change from the ${\rm Diff}(S^2)$ subgroup
$(\alpha,\gamma,\chi) = (0,0,\chi)$ to the alternative ${\rm
  Diff}(S^2)$ subgroup given by $(\alpha,\gamma,\chi) = [\alpha(\chi),
0, \chi]$, with $\alpha(\chi)$ the function of $\chi$ given by imposing
Eq.\ (\ref{alpha444}), then in the semidirect product
(\ref{BMSWeylStructure}) the Killing supertranslations ${\cal W}$ transform as
scalars but the affine supertranslations ${\cal S}$ transform as scalar densities 
as in Eq.\ (\ref{scalardensity}), from Eqs.\ (\ref{groupcompos}).  This alternative ${\rm Diff}(S^2)$
subgroup is the one that arises naturally within the generalized BMS subgroup $(\alpha,
\gamma, \chi ) = [\alpha(\chi), \gamma, \chi]$, which is why affine
supertranslations transform as scalar densities in the GBMS and BMS
cases \cite{Barnich:2021dta}.}\footnote{One can also express the Weyl BMS group as the semidirect product 
$({\rm Diff}(S^2)  \ltimes  {\cal W}) \ltimes {\cal S}$. In this case
the action of the second semidirect product endows the affine
supertranslations with a certain weight under the Weyl rescalings
${\cal W}$, as well
as an independent density weight under the diffeomorphisms, which
again can be adjusted at will by choosing the ${\rm Diff}(S^2)$ 
subgroup appropriately.} from Eqs.\ (\ref{groupcompos}).

\end{itemize}

\section{Gravitational charges at asymptotic boundaries: holographic renormalization}
\label{sec:hr}

\subsection{Introduction and Overview}
\label{sec:hr1}

We now return to the context of general theories and general spacetime dimensions.
The general formalism for gravitational boundary symmetries and
charges developed in Sec.\ \ref{sec:charges} above assumed that the boundaries are
at finite locations in spacetime (for example future event horizons),
and that all the quantities defined are finite on those boundaries.
In this section we will extend the formalism to handle cases of
asymptotic boundaries, treated in a conformal completion framework to
bring them to finite locations.\footnote{Although we assume a
conformal completion, our general framework would also be applicable
to situations like odd-dimensional asymptotically flat spacetimes
where the conformal framework does not apply \cite{Godazgar:2012zq}, by making use of a
radial foliation.}
In general the Lagrangian and symplectic
potential can then diverge at the boundaries.  Some previous general treatments of the
covariant phase space framework either did not treat asymptotic
boundaries in general \cite{Harlow:2019yfa}, or used the finiteness of
certain quantities on the boundary as a criterion to determine which
boundary conditions to impose and which field configuration space to
use \cite{Wald:1999wa},which is in general unnecessarily restrictive.

The key idea of the extended formalism is holographic renormalization
\cite{Henningson1998a, Balasubramanian:1999re,DeHaro2001,Papadimitriou:2005ii,
Compere:2008us,CL,Campiglia:2014yka,Compere:2018ylh,Freidel:2019ohg},
which exploits the boundary canonical transformations and JKM transformations 
discussed in Sec.\ \ref{sec:JKM} above
to make the integrated action and symplectic potential $\theta'$
finite on the boundary. 
Once one has a renormalized symplectic potential, the formalism of
Sec.\ \ref{sec:charges} then yields finite gravitational global and
localized charges.

In Sec.\ \ref{sec:charges} we discussed the fact that
it is sometimes necessary to introduce a background structure which violates
covariance in order to obtain gravitational charges.  For example this
occurs with certain boundary conditions on finite null surfaces
\cite{Chandrasekaran:2020wwn}.
Similarly, here it is sometimes the case that the transformations that
are necessary to renormalize the Lagrangian and symplectic form
require the introduction of some background structures. This is the
case, for example, for generalized BMS symmetries (Sec.\ \ref{sec:fcs} above) in vacuum general relativity at
null infinity, where it was shown in Ref.\ \cite{Flanagan:2019vbl}
that a completely covariant renormalization of the symplectic
potential is impossible.  A similar situation arises in asymptotically
AdS spacetimes, where it is necessary to introduce a spacetime
foliation near the boundary as a background structure when renormalizing the 
action.  Dependence of the renormalized quantitities on this foliation 
signals the appearance of Weyl anomalies in the dual CFT description 
\cite{Henningson1998a, Balasubramanian:1999re,
Skenderis:2000in, Papadimitriou:2005ii, 
Compere:2008us}.

The various possible background structures that we will
consider are (i) a foliation of spacetime near the boundary; (ii) A
choice of conformal factor $\conf$ near the boundary; and (iii) A
choice of vector field ${\vec v}$ near the boundary which satisfies
$v^a \nabla_a \conf =1$, which we will call a {\it rigging vector
  field}.  We expect that in many situations only the foliation will
be necessary to effect holographic renormalization.  However,
later in this section will make use of the more restrictive assumption
of the existence of a rigging vector field to argue that holographic
renormalization can always be successfully carried out.

We start by defining our notation and conventions for the covariant
phase space framework with asymptotic boundaries.
As discussed in Sec.\ \ref{sec:presymplecticform}, our phase subspaces are defined by a spacetime subregion $\sr$,
and a spatial slice $\Sigma$ whose boundary $\partial \Sigma$ lies in
$\partial \sr$.  We define ${\cal D} = \sr \cap I^+(\Sigma)$, the
intersection of the subregion with the chronological future
of the spatial slice, i.e.\
the set of all points 
to the future of $\Sigma$ in $\sr$.  The region ${\cal D}$ will be the setting for the
action principle for the phase subspace.  We will denote by ${\cal
  N}_j$ the various components of the boundary of ${\cal D}$, one of
whom will be the portion of the boundary ${\cal N}$ discussed in
Sec.\ \ref{sec:presymplecticform} in the chronological future of $\Sigma$, 
and one of whom will be the spatial
slice $\Sigma$.  
For asymptotic boundaries we use the conformal completion framework
and work with conformally rescaled fields which are finite on the
boundary.
The conformal factor $\Phi$ is chosen to be vanishing on
asymptotic boundaries and to have a nonzero gradient there, and to be positive on
${\cal D}$.

In order to discuss renormalized actions, we define a 
cutoff manifold ${\cal D}_\upsilon$
to be the set of points in
${\cal D}$ with $\Phi > \upsilon$, which excludes a neighborhood of
the asymptotic boundaries.  We will assume that the boundary of the truncated
manifold can be decomposed into a number of components in the same way
as the full manifold: 
\be
   \partial{\cal D}_\upsilon = \bigcup_j {\cal N}_{j,\upsilon}.
\ee
The standard example we will have in mind is depicted in Figure \ref{fig:holog}.  
Here, the region $\mathcal{D}_\ups$ has a single timelike boundary $\mathcal{N}_\ups$
that limits to a segment of $\scp$, and $\Sigma_\ups$ comprises the past boundary
of $\mathcal{D}_\ups$.  We also include a future boundary $\Sigma_\ups^f$ 
so that the cutoff region is bounded in spacetime, making all cutoff integrals
manifestly finite.


\begin{figure}
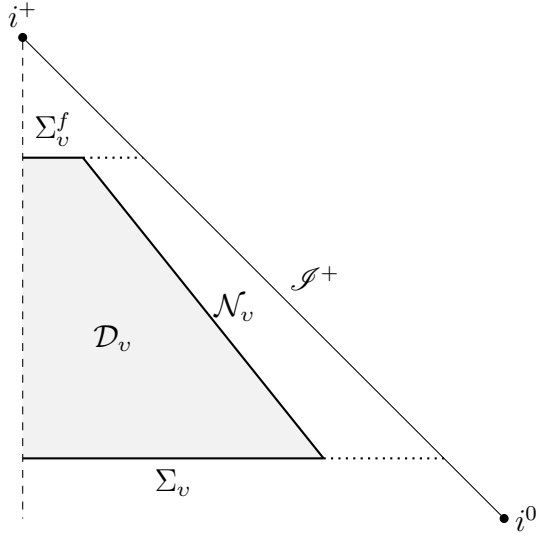

\centering
\tikz [scale = 0.8,
        pt/.style = {draw, circle, fill=black, inner sep =0 mm, minimum size = 3pt}
        ]
{    
    \coordinate  (p1f) at (0,6);
    \coordinate  (p1p) at (0,1);
    \coordinate  (p2f) at (1,6);
    \coordinate  (p2p) at (5,1);
    \coordinate  (scrif) at (2,6);
    \coordinate  (scrip) at (7,1);
    \coordinate  (origin) at (0,0);
    
    \fill [gray!10] (p1f) -- (p2f) -- (p2p) -- (p1p) -- cycle;
    
    \node [pt] (ip) at (0,8) {};
    \node [pt] (i0) at (8,0) {};
    \node (u) at (1.5,3) {$\mathcal{D}_\ups$};
    
    \draw (ip) node[above]{$i^+$} -- (i0) node[right] {$i^0$} node[midway,right=2mm]{$\scp$};
    \draw [thick] (p2f) --(p2p) node[midway, right] {$\mathcal{N}_\ups$};
    \draw [thick] (p1f) -- (p2f) node[midway, above] {$\Sigma^f_\ups$};
    \draw [thick] (p1p) -- (p2p) node[midway, below] {$\Sigma_\ups$};
    
    \draw [dashed] (ip) -- (origin);
    \draw [dotted, thick] (p2f) -- (scrif);
    \draw [dotted, thick] (p2p) -- (scrip);

}
\caption{Standard setup for holographic renormalization in asymptotically
flat spacetimes.  The subregion under consideration is associated with a segment of 
$\scp$ to the future of a spatial surface $\Sigma$.  The cutoff subregion $\mathcal{D}_\ups$
is 
depicted in gray, and its boundary components $\mathcal{N}_{j,\ups}$ consist
of $\Sigma_\ups$, $\mathcal{N}_\ups$, and $\Sigma_\ups^f$.\label{fig:holog} }
\end{figure}

On each boundary we define a boundary action $\ell'_j$, and on the
boundaries $\partial {\cal N}_j$ we define corner actions $c'_j$.  Then
the total action is defined to be
\be
S_\upsilon = \int_{{\cal D}_\upsilon} L' + \sum_j \int_{{\cal
    N}_{j,\upsilon}} \ell'_j + \sum_j \int_{\partial {\cal
    N}_{j,\upsilon}} c'_j.
\label{ffs}
\ee
A successful renormalization of the action consists of 
finding definitions of $\ell'_j$ and $c'_j$ so that
$S_\upsilon$ has a finite limit as $\upsilon \to 0$.
We will generally drop the subscript $j$ in the remainder of this section.

In the action (\ref{ffs}) we allow the boundary term $\ell'_j$ to have
different forms on different boundaries ${\cal N}_j$, and be in effect
discontinuous across the interfaces $\partial {\cal N}_j$ between two
different boundaries.  This is the reason for including the corner
terms, which otherwise would give a vanishing contribution if
$\ell'_j + d c'_j$ were a continuous function on $\partial
{\cal D}$.  An additional reason for separating out corner terms is
as follows. One might imagine eliminating such terms by replacing
$\ell'_j \to \ell'_j + d c'_j$.  However such a replacement may violate
corner
covariance; it can arise that there exist  definitions of
$c'_j$ that are covariant with respect to 
corner-preserving diffeomorphisms, but no  extensions of these definitions
to objects that are fully covariant with respect to diffeomorphisms of the entire boundary.

The general scheme for holographic renormalization can be described in
terms of a number of stages, starting with conventional stages to
renormalize the action, and then subsequent stages to renormalize the
symplectic potential and to adjust corner terms.  We now give an overview of the various
stages.  Although in practice the charges and other JKM-invariant quantities
would only be computed at the end of the process, we list in the overview
which of the these quantities would change at each stage, in order
to clarify the logical structure of the process.
In later subsections we will show that the steps described
here can be carried out successfully for general theories using a
rigging vector field and a boundary vector field as background structures (although we expect that
a bulk foliation will in general be sufficient).

\begin{enumerate}

\item We start with a divergent Lagrangian $L'$ and symplectic
  potential $\theta'$.  We choose the decomposition
  (\ref{eqn:thdecomp}) to make $\ell'=\beta'=0$.  We assume that the
  initial quantities are covariant, so that $b' = \lambda'=0$, and it
  follows from Eq.\ (\ref{eqn:bldecomp}) that $c' =\gamma' = \varepsilon=0$.
The action $S$, symplectic form $\Omega$, gravitational charges
${\tilde H}_\xi$ and flux ${\cal E}$ are all divergent.

\item We perform a boundary canonical
  transformation (\ref{eq:bct}) with $B$ and $\Lambda$ chosen 
  to put the flux ${\cal E}$ into Dirichlet form.
  The modifies ${\cal E}$, $\ell'$, $\beta'$, $S$ and ${\tilde H}_\xi$, and in particular involves
  adding a boundary term to $\ell'$.  For example, in vacuum general relativity with a
  timelike boundary the corresponding 
  boundary term $B$ is the Gibbons-Hawking-York term.

\item We perform a second boundary canonical transformation
  (\ref{eq:bct}) to make $L' + d\ell'$ finite and also to make the flux ${\cal
  E}$ finite, while preserving the Dirichlet form of the flux.
  The parameter $B$ of the transformation is 
  a boundary counterterm that is added to $\ell'$ to make $L' + d \ell'$
  finite on the boundary, which is a functional of the intrinsic data
  on the boundary \cite{Henningson1998a, Balasubramanian:1999re, MM-term}. 
  The parameter $\Lambda$ is computed from $B$ in
  the same way that the symplectic potential $\theta'$ is computed
  from the bulk Lagrangian, by varying with respect to the intrinsic
  data, see for example \cite{Fiorucci:2020xto}.  Schematically
  $$
  \delta B = \frac{\delta B}{\delta ({\rm intrinsic})} \delta ({\rm
    intrinsic}) + d \Lambda,
  $$
which guarantees that the modification ${\cal E} \to {\cal E} + \delta
B - d\Lambda$ to the flux preserves Dirichlet form.
The quantities which change in this step are ${\cal E}$, $\ell'$, 
  $\beta'$, $\Omega$, $S$ and ${\tilde H}_\xi$.
  In Sec.\ \ref{exparg} below we show explicitly that it is always possible
  to find a boundary canonical transformation $(B,\Lambda)$ that makes
  both $L' + d \ell'$ and ${\cal E}$ finite.  (It should be possible to
  always specialize the transformation to preserve Dirichlet form
  although we do not show this here.)

\item We next repeat these steps in one lower dimension to adjust
  corner terms.
  We perform a corner
  canonical transformation (\ref{cct}) parameterized by forms $F$ and $\zeta$ in order to put the corner flux
  $\varepsilon$ into Dirichlet form.  This involves identifying
  appropriate intrinsic degrees of freedom on the corners $\partial
  {\cal N}$.  Additional quantities that change are $S$ and ${\tilde H}_\xi$.

\item We then perform a second corner canonical transformation
  (\ref{cct}) to make
  the integrated action $S$ and corner flux $\varepsilon$ finite,
  while maintaining the Dirichlet form of the corner flux.
  The parameter $-F$ of the transformation is 
  a corner counterterm that is added to $c'$ to make the integrated
  action finite, which is a functional of the intrinsic data
  on the corner. 
  The parameter $\zeta$ is computed from $F$ as described in step 3 above, by
  varying with respect to the intrinsic data, to ensure that the
  modification (\ref{cct}) to the corner flux $\varepsilon$ preserves Dirichlet
  form.
  The gravitational charges ${\tilde H}_\xi$ as well as the
symplectic form $\Omega$ are now finite, if the corner
terms can be chosen so that $\Delta_{\hat \xi} (\ell' + dc')$ is finite
[from Sec.\ \ref{sec:corners} the charges
    ${\tilde H}_\xi$ will be finite if ${\cal E}$, $L' + d \ell'$ and
    $\Delta_{\hat \xi}(b' + \ell' + d c')$ are all finite, and $b'=0$ here].
In Sec.\ \ref{exparg1} below we show explicitly that this is the case: it is always possible
  to find a corner canonical transformation $(f,\zeta)$ of this kind that makes
  the integrated action $S$, corner flux $\varepsilon$ and $\ell' + d c'$ finite.  The general proof requires the introduction of additional background structure on the boundary.

\item We use a JKM transformation (\ref{eqn:JKMtrans}) with $a = \ell'$
  and $\nu = -\beta'$.  This step is not necessary to
  obtain finite charges, but it is convenient in order to make other
  quantities in the formalism finite.
  It makes finite the Lagrangian $L'$
  and symplectic potential $\theta'$, and also modifies $b'$, $\lambda'$
  and $\ell'$, but does not affect the charges, action or symplectic form.

\item Finally we use a transformation of the form (\ref{eqn:echidef})
  with $e = c'$, $\chi = \gamma'$ and $\cov \cb = \cov
  \cgamma=0$. This step is also optional, since it does not change the
  charges, but it is convenient as it makes $c'$ and $\gamma'$ vanish and
  $\lambda'$ finite.  The other quantity that is modified is $b'$, 
  which can still be divergent. However $\Delta_{\hat
    \xi} b' = \Delta_{\hat \xi}(b' + \ell' + dc')$ will be finite
  (since this quantity was finite at step 5 and is not modified by steps 6 or 7), which is sufficient for
  finiteness of the charges.

\end{enumerate}

The algorithm described here is summarized in Table \ref{tab:ren},
which shows which quantities change at each step, and when quantities
diverge or become finite.

Although we show that this procedure can always be carried
out in order to obtain finite renormalized charges,
there are a number of subtleties that arise 
in the asymptotically flat case that are not present in 
asymptotically dS or AdS spaces.\footnote{We thank Kostas
Skenderis and Ioannis Papadimitriou for discussions on this point.}
These subtleties relate to 
the form of the counterterm for the action $\ell_{\text{ct}}$ 
and their dependence on the free data associated with the cutoff
surface $\ns$.  In asymptotically (A)dS spaces, this free
data can be chosen to be the induced metric on the cutoff surface,
which is fully unconstrained by the equations of motion in 
the limit that the cutoff is taken to the boundary.  The counterterms
needed to renormalize the action are covariantly constructed 
from the boundary metric, and hence they are given by local
expressions in terms of the free data.  Locality and covariance
of the counterterms are important in the holographic correspondence,
as they allow the on-shell action to be interpreted as the 
generating functional of correlation functions in a local CFT
dual
\cite{Witten:1998qj, Gubser:1998bc, Bianchi:2001kw}.  

For asymptotically flat spacetimes, one finds that the 
induced metric on the cutoff surface is not freely
specifiable, but instead satisfies a number of differential
constraints in the limit that the cutoff surface is taken
to infinity \cite{deHaro:2000wj}.
Even though counterterms that are local and covariant
with respect to the induced metric on the cutoff 
surface can be constructed \cite{MM-term}, since this induced 
metric is not freely specifiable, the variation of the on-shell
action with respect to this induced metric is necessarily subject
to constraints.  This generically leads to nonlocal dependence of the 
counterterms on the free data on the cutoff surface, and further
complicates a simple holographic interpretation for 
asymptotically flat spaces
\cite{deHaro:2000wj, Skenderis:2002wp, Papadimitriou:2010as}.  
While this does not affect the 
main results of this work on obtaining finite gravitational charges,
addressing these subtleties is an important question for 
further developing the holographic correspondence in flat space.

\begin{table}
\centering
\footnotesize
\begin{tabular}{|c|c|c||c|c|c|c|c|c|c|}
\hline
 & \makecell*{Step}  && 1 & 2 &
3 & 4 &
5 &
6 &
7 \\
\hline
 & \makecell*{Transformations}  && &$B,\Lambda$ & $B,\Lambda$ &
$F,\zeta$ &$F,\zeta$   &$a,\nu$ &
$e,\chi$ \\
\hline
 &\makecell*{Quantity}& Eqn. &  & &  & &  &  &  
\\
\hhline{| = | = |= # =|=|=|=|=|=|=|}
\makecell*{$S$} & action & \ref{eqn:action}& $\infty$
 & $\to \infty $  &$\to \infty $ & $\to \infty $ &$\to$ f & & 
   \\
\hline
\makecell*{$\Omega$} & sympl. form & \ref{eqn:Omcorner}& $\infty$
& $\to \infty$ &$\to $ f &  &&  &      \\
\hline
\makecell*{$\int_{\partial \Sigma}{\tilde h}_\xi$ }& integ. charge &
\ref{eqn:hxiimproved}&$\infty$ &$\to \infty$  &$\to \infty$ & $\to
\infty$ &$\to $ f &  &  \\
\hhline{| = | = |= # =|=|=|=|=|=|=|}
\makecell*{$L'$} & bulk Lagrangian &\ref{eqn:dL} & $\infty$ &$\infty$  &$\infty$
&$\infty$ &$\infty $  &$\to $ f &  
 \\
\hline
\makecell*{$\theta'$} & symplectic pot.& \ref{eqn:dL} &$\infty$  &$\infty$
& $\infty$ & $\infty$
&$\infty$  &$\to$ f  &  \\
\hline
\makecell*{${\cal E}$} &boundary flux &\ref{eqn:thdecomp} &$\infty$ & $\to \infty$
&$\to $ f & & & &   \\
\hhline{| = | = |= # =|=|=|=|=|=|=|}
\makecell*{$b'$} & pot., noncovar.&\ref{eqn:blambdadef}
&$0$ &$0$ &$0$ &$0$ &$0$ &$\to \infty$  & $\to \infty$  \\
\hline
\makecell*{$\lambda'$} & pot., noncovar. & \ref{eqn:blambdadef} & $0$
& $0$& $0$ &$0$ &$0$ & $\to \infty$ &$\to$ f       \\
\hline
\makecell*{$\ell'$} & boundary action & \ref{eqn:thdecomp}&$0$ &$\to \infty$  &
$\to \infty$& $\infty$ & $\infty$&$\to 0$
&  \\
\hline
\makecell*{$\beta'$} & corner term & \ref{eqn:thdecomp} &$0$ & $\to \infty$ &$\to
\infty$ &$\infty$ &$\infty$ &$\to 0$  &   \\
\hhline{| = | = |= # =|=|=|=|=|=|=|}
\makecell*{$c'$} & corner action &\ref{eqn:bldecomp} &$0$ & $0$ &$0$ &$\to \infty$
&$\to \infty$ & $\infty$  &$\to 0$
  \\
\hline
\makecell*{$\gamma'$} & & \ref{eqn:bldecomp} & $0$  &$0$
&$0$ &$\to \infty$ &$\to \infty$ &$\infty$ &$\to 0$    \\
\hline
\makecell*{$\varepsilon$} & corner flux &\ref{eqn:bldecomp} & $0$ &$\to \infty$ &
$\to \infty$&$\to \infty$ & $\to$ f & &  \\
\hline
\end{tabular}
\caption{
A summary of the steps in the general holographic renormalization procedure
of this paper.
The first row lists the
numbered steps described in Sec.\ \ref{sec:hr1}, and the second row the
transformations used in getting to each step from the previous step
(detailed in Table \ref{tab:operations}).
The remaining rows list the various quantities
that occur in the formalism, their names, defining equations, and how
they transform under the steps. The symbol $\to$ means that the
corresponding quantity is altered by the transformation of that step.
The symbols $\infty$, f, and $0$ mean that the quantity is diverging,
finite, and vanishing, respectively, at the end of the
step.
}
\label{tab:ren}
\end{table}

The remainder of this section is organized as follows.
In Section \ref{genarg} we provide a very general argument which shows
that the existence of a renormalized total action (including bulk,
boundary and corner terms) is sufficient to show that it is possible to
renormalize the symplectic potential.
The renormalized symplectic potential and charges will in general depend on the choice of
foliation used to renormalize the action, as discussed above.
In Section \ref{exparg} we present a similar and complementary
result. Given any covariant theory in which the Lagrangian
and symplectic potential diverge near the boundary, we
demonstrate the existence of a renormalized Lagrangian and
a renormalized symplectic potential,  by explicitly
computing the counterterms that one needs to subtract off
in order to obtain finite quantities.  These counterterms again depend
on the choice of conformal factor, and in addition depend on a choice
of rigging vector field.

As an application of our holographic renormalization formalism, in Sec.\ \ref{expcalc} below
we specialize to vacuum general relativity in asymptotically flat
3+1 dimensional spacetimes, specialize to the generalized BMS field configuration
space, and compute the renormalized symplectic potential and the associated localized charges at future
null infinity.

\subsection{Existence of renormalized symplectic potential: general argument assuming a finite action}
\label{genarg}

In this section we show that a renormalized action functional is
sufficient to provide a renormalized symplectic potential. From this one can obtain
a complete set of IR finite observables which act on the boundary
phase space.

We emphasize from the outset that non-covariances and
background structures play a crucial role in this renormalization
procedure, which thus avoids the no-go theorem of
\cite{Flanagan:2019vbl}. In particular, the general argument relies on
the introduction of a background foliation near the boundary, on
which the renormalized action depends [cf.\ Eq.\ (\ref{ffs}) above]. Thus, any renormalization of
the boundary action can itself involve non-covariant counter-terms.
This is entirely analogous to the way holographic renormalization
works in AdS/CFT, where a radial foliation is
introduced near the AdS boundary, and different choices of foliations
lead to different values of the boundary action
\cite{Henningson1998a, Balasubramanian:1999re,
Skenderis:2000in, DeHaro2001, Papadimitriou:2005ii, 
Compere:2008us}. On the boundary, this is simply the statement
that UV regulators 
break conformal invariance. 
However, we do demand that the renormalization procedure respects
boundary covariance, a point which 
was also emphasized in \cite{Harlow:2019yfa}. Thus, we do not
introduce any background structures beyond the 
foliation.

We start by taking a variation of the integrated action (\ref{ffs}) and
using the equation of motion (\ref{eqn:dL}) to eliminate the bulk contribution.
Applying the decomposition
(\ref{eqn:thdecomp}) of the symplectic potential results in\footnote{Note that 
some of the boundaries $\partial \ns_{j,\ups}$ may have multiple components, in which
case the quantities $c_j'$ and $\cflx_j$ are understood to be independently specifiable 
on each component.}
\be
\delta S_\upsilon = \sum_j \int_{{\cal N}_{j,\upsilon}} {\cal
  E}_j +  \sum_j \int_{\partial {\cal N}_{j,\upsilon}} (\beta'_j +
\delta c_j).
\ee
We next use the decomposition (\ref{eqn:bldecomp}) of $\beta_j'$ to give
\be
\delta S_\upsilon = \sum_j \int_{{\cal N}_{j,\upsilon}} {\cal
  E}_j +  \sum_j \int_{\partial {\cal N}_{j,\upsilon}} (\lambda' +
\varepsilon_j).
\ee
Finally, we note that $\lambda'$ is assumed to be continuous on
$\partial {\cal D_\ups}$, implying continuity at the interfaces between
different boundaries ${\cal N}_j$.  Since the contributions from $\lambda'$ 
always occur in pairs with opposite signs coming from the two $\mathcal{N}_{j,\ups}$ 
intersecting at each corner, the overall contribution
from $\lambda'$ vanishes, and so
\be
\delta S_\upsilon = \sum_j \int_{{\cal N}_{j,\upsilon}} {\cal
  E}_j +  \sum_j \int_{\partial {\cal N}_{j,\upsilon}} 
\varepsilon_j.
\label{xx}
\ee
By assumption,  the left hand side has a finite limit
as $\ups\ra 0$, implying that the sum of all the boundary and corner
fluxes on the right must also be finite in the limit.  Hence, if any individual contribution
in these sums diverges, it must be canceled by a divergence appearing in a different term.  
In this case, one expects to be able to choose each of the corner fluxes
$\cflx_j$ such that $\beom_j+d\cflx_j$ has a finite limit.  

To see how this is borne out in more detail, we can focus on the 
standard example given in Figure \ref{fig:holog} in which the codimension-1
boundaries consist of $\Sigma_\ups$, $\Sigma_\ups^f$ and $\ns_\ups$, and the corners
are $\partial\Sigma_\ups$ and $\partial\Sigma_\ups^f$.  
The spatial surfaces $\Sigma_\ups$ and $\Sigma_\ups^f$ intersect the 
$\Phi$ foliation transversally, and so any divergence coming from an
integral of the respective fluxes over these surfaces must be localized in the 
$\Phi\ra 0$ regions of these surfaces, which are just their boundaries. Since
the boundaries of $\Sigma_\ups$ and $\Sigma_\ups^f$ coincide with the boundaries
of $\ns_\ups$, it follows that any remaining divergence in the flux on $\ns_\ups$ must 
cancel against divergences at its boundary.  Hence, it is possible to shift 
the flux $\beom_{\ns}$ by an exact term to cancel its divergence.  This allows us to 
conclude that we can arrange for the $\cflx_j$ to be chosen such that 
$\beom_j + d\cflx_j$ is finite on each boundary.

Finally, from the decomposition (\ref{eqn:bldecomp}), it follows that
the symplectic potential $\theta'$ is finite up to exact terms and
total variations.  Hence we can  find a JKM transformation (\ref{eqn:JKMtrans})
that makes the symplectic potential at any given asymptotic boundary finite.
This renormalized symplectic potential will in general depend on the choice of
foliation.

\subsection{Explicit renormalization using background structures}
\label{exparg}

In this section we provide an explicit algorithm for holographic renormalization in
general contexts, based on allowing the counterterms to depend
on a rigging vector field in addition to a foliation.  The intent is
to provide an existence proof for background-dependent counterterms.
However we expect that in applications it will be generally possible
to find counterterms that depend only on a foliation.  Note that 
the dependence of the counterterms on the additional background
structure provided by the rigging vector is possibly related to the nonlocality of these terms
relative to the free data on the asymptotic boundary, as 
discussed in \cite{deHaro:2000wj, Skenderis:2002wp}.  

As discussed in Sec.\ \ref{sec:hr1} above, the setup is that we have a region
${\cal D}$ in a $d+1$ dimensional spacetime, and a portion ${\cal N}$ of the boundary of
${\cal D}$, where we are using conformal compactification to treat
the asymptotic boundary ${\cal N}$ as a finite boundary.
We assume that the Lagrangian $L'$ and 
symplectic potential $\theta'$ are smooth in the interior of ${\cal D}$
but can diverge on ${\cal N}$, and assume a smooth conformal
factor $\conf$ with $\conf=0$ on ${\cal N}$.  We fix a rigging vector
field ${\vec v}$
which is defined on a neighborhood of ${\cal N}$, is nowhere
vanishing, and satisfies $v^a \nabla_a \conf = 1$
\cite{Mars1993,Schouten1954}. Note that $\conf$ is determined in terms
of ${\vec v}$ from
this condition and the condition $\conf=0$ on ${\cal N}$.

Consider now a boundary canonical transformation (\ref{eq:bct})
parameterized by $B$ and $\Lambda$.  If we combine this with a JKM
transformation (\ref{eqn:JKMtrans}) with $ a = B$ and $\nu = -
\Lambda$, the Lagrangian and symplectic potential transform as
\begin{subequations}
  \label{JKM11}
  \begin{eqnarray}
\label{Lren}
L' \to L'_{\rm ren} &=& L' + d B, \\
\label{thetaren}
\theta' \to \theta'_{\rm ren} &=& \theta' + \delta B - d\Lambda,
  \end{eqnarray}
\end{subequations}
The two main results of this section are:
\begin{enumerate}

\item There exists a 
  transformation of this kind
for which the renormalized Lagrangian and symplectic potential (\ref{JKM11})
(and not just
their pullbacks to surfaces of constant $\conf$) have finite limits to
the boundary ${\cal N}$. We will construct this transformation explicitly.
This is sufficient to make the charges ${\tilde H}_\xi$ finite,
assuming the property of corner terms described in step 3 of Sec.\ \ref{sec:hr1}.

\item The anomaly
  \be
  \Delta_{\hat \xi} \p\theta'_{\rm ren}
  \ee
in the pullback to the boundary of the renormalized symplectic
potential is the sum of an exact term and a total variation, both of which
are finite on the boundary.

\end{enumerate}
The derivation of these results in given in Appendix \ref{app:hr}.
Here we summarize some of the details.

\subsubsection{Renormalized symplectic potential and Lagrangian}
\label{sec:rspL}

We start by defining some notations.  We define for $\upsilon > 0$ a map
\be
\pi_\upsilon: {\cal N} \to {\cal D}
\ee
which moves any point $\upsilon$ units along the
integral curve of ${\vec v}$ passing through that point.  The image of
${\cal N}$ under this mapping is the surface $\conf = \upsilon$
which we will denote by ${\cal N}_\upsilon$.  Any differential form $\Lambda$
on a neighborhood of ${\cal N}$ is uniquely determined by specifying
(i) $i_v \Lambda$, and (ii) the pullbacks $\pi_{\upsilon}^* \Lambda$ for the values of
$\upsilon$ that cover the neighborhood, via
$\Lambda = d\Phi\wedge i_v \Lambda + \Lambda_h$ with $i_v \Lambda_h =
0$ and $\pi_\upsilon^*\Lambda_h = \pi_\upsilon^*\Lambda$.

The quantities $B$ and
$\Lambda$ that define the boundary canonical transformation (\ref{JKM11}) are given
by
\bes
\label{bprimelambdaprime}
\bea
i_v B &=& 0, \\
\pi_{\upsilon}^* B &=&  \int_{\upsilon}^{\upsilon_0} d{\bar
  \upsilon} \, \pi_{{\bar \upsilon}}^* i_v L', \\
i_v \Lambda &=& 0, \\
\pi_{\upsilon}^* \Lambda &=&  -\int_{\upsilon}^{\upsilon_0} d{\bar
  \upsilon} \, \pi_{{\bar \upsilon}}^* i_v \theta',
\eea
\ees
where we have chosen a fixed $\upsilon_0 > 0$.
The renormalized Lagrangian and 
symplectic potential are given by
\bes
\label{thetaprimeans}
\bea
i_v L'_{\rm ren} &=& 0, \\
i_v \theta'_{\rm ren} &=& 0, \\
\pi_{\upsilon}^* \theta'_{\rm ren} &=&  \pi_{\upsilon_0}^* \theta'.
\eea
\ees

The expressions (\ref{bprimelambdaprime}) and (\ref{thetaprimeans})
become more transparent when expressed in a suitable class of 
coordinate systems.
We choose a coordinate system $(x^0,x^1, \ldots ,x^d)$ for which $x^0 = \conf$
and ${\vec v} = \partial / \partial x^0$.
We define for convenience the basis $d$-forms, $(d-1)$-forms and $(d-2)$-forms 
\bes
\bea
\varpi &=& dx^1 \wedge \ldots \wedge dx^d, \\
\varpi_i &=& - i_{\partial_i} \varpi = (-1)^i dx^1 \wedge \ldots \wedge {\widehat {dx^i}} \wedge \ldots \wedge dx^d, \\
\varpi_{ij} &=& i_{\partial_i} \varpi_j,
\eea
\ees
where the hat on a basis one-form means that that one-form is omitted
in the wedge product,
and $i$ and $j$ run over $1 \ldots d$.
We expand the symplectic potential $\theta'$, Lagrangian $L'$, and boundary canonical
transformation forms $B$ and $\Lambda$ as
\bes
\label{fieldsexpand}
\bea
\label{Lexpand}
L'  &=& {\cal L}  \, dx^0 \wedge \varpi,\\
\label{thetaexpand}
\theta' &=& \theta^{\prime\,0} \varpi +  \theta^{\prime\,i} \, dx^0 \wedge \varpi_i,\\
\label{bprimeexpand}
B &=& B^{0} \varpi +  B^{i} \, dx^0 \wedge
\varpi_i,\\
\Lambda &=& \Lambda^{i} \varpi_i +  \Lambda^{ij} \, dx^0 \wedge
\varpi_{ij},
\eea
\ees
together with a similar notation for $L'_{\rm ren}$ and $\theta'_{\rm ren}$.
Using these notations Eqs. (\ref{bprimelambdaprime}) reduce to
\bes
\label{translated}
\bea
B^{i} &=& 0,\\
B^{0}(\upsilon,x^j) &=& \int_{\upsilon}^{\upsilon_0}
d {\bar \upsilon} \, {\cal L}({\bar \upsilon},x^j),\\
\Lambda^{i}(\upsilon,x^j) &=& -\int_{\upsilon}^{\upsilon_0}
d {\bar \upsilon} \, \theta^{\prime\,i}({\bar \upsilon},x^j),\\
\Lambda^{ij}(\upsilon,x^j) &=& 0,
\eea
\ees
where we have written
$(x^0,x^1, \ldots x^d) = (\upsilon,x^1, \ldots x^d) = (\upsilon,x^j)$.  Similarly
Eqs.\ (\ref{thetaprimeans}) reduce to
\bes
\label{thetaLfinite}
\bea
\label{thetaifinite}
{\cal L}_{\rm ren} &=& 0,\\
\theta^{\prime\,i}_{\rm ren} &=& 0, \\
\label{theta0finite}
\theta^{\prime\,0}_{\rm ren}(\upsilon,x^j) &=& \theta^{\prime\,0}_{\rm ren}(0,x^j) =
\theta^{\prime\,0}(\upsilon_0,x^j).
\eea
\ees

We now turn to some examples of applications of this formalism.
In many cases we can split the Lagrangian and symplectic potential
into diverging and finite pieces, $L' = L'_{\rm div} + L'_{\rm finite}$,
$\theta' = \theta'_{\rm div} + \theta'_{\rm finite}$, such that the
diverging pieces obey the identity (\ref{eqn:dL}) on shell,
$\delta L'_{\rm div} = d \theta'_{\rm div}$.  It is convenient then to compute the
boundary canonical transformation using just the diverging pieces, which is
sufficient to make $L'_{\rm ren}$ and $\theta'_{\rm ren}$ finite.  In this case the result
(\ref{thetaLfinite}) for the renormalized Lagrangian and symplectic potential
on the boundary becomes
\bes
\bea
{\cal L}_{\rm ren}(0,x^j) &=& {\cal L}_{\rm finite}(0,x^j), \\
\theta^{\prime\,i}_{\rm ren}(0,x^j) &=& \theta^{\prime\,i}_{\rm finite}(0,x^j), \\
\label{theta0finite1}
\theta^{\prime\,0}_{\rm ren}(0,x^j) &=&
\theta_{\rm finite}^{\prime\,0}(0,x^j)+ \theta_{\rm div}^{\prime\,0}(\upsilon_0,x^j).
\eea
\ees
If we can further choose $\upsilon_0$ to make the second term in Eq.\ (\ref{theta0finite1})
vanish, then we obtain 
\be
\label{theta0finite2}
\theta^{\prime\,0}_{\rm ren}(0,x^j) =
\theta_{\rm finite}^{\prime\,0}(0,x^j),
\ee
so the pullback of the renormalized symplectic potential is just the
pullback of its finite piece.

A class of examples which includes vacuum general relativity at
null infinity (see Sec.\ \ref{expcalc} below) is when $L'_{\rm div}$ and
$\theta'_{\rm div}$ are finite polynomials in $\conf^{-1}$.
In this case the second term in Eq.\ (\ref{theta0finite1}) can be made to vanish by
choosing $\upsilon_0=\infty$.  The choice $\upsilon_0 = x^0 =
\infty$ seems at first glance to be problematic, since
the coordinates need only be defined for a finite range of values of $x^0$.
However, we can regard
the choice of value of $\upsilon_0$ as a specification of a
prescription for finding antiderivatives of the specific functions encountered
in Eqs.\ (\ref{translated}), so the inconsistency can be finessed.

A more general class of examples is obtained by allowing log terms,
which arise for example in asymptotically AdS spacetimes. 
Suppose that there exists some integer $t$ for
which $\conf^t L'$ and $\conf^t \theta'$ are smooth functions of $\conf$ and $\conf \log \conf$.
Then we can expand $L'$ and $\theta'$ as
\bes
\label{ex}
\bea
\label{e1}
\theta^{\prime\,0} &=& 
\sum_{p=-t}^{-1} \sum_{q=0}^\infty \theta^{0\,(p,q)} \conf^p ( \conf \log \conf)^q +
\theta^{\prime\,0}_{\rm
  finite}, \\
\label{e2}
\theta^{\prime\,i} &=& 
\sum_{p=-t}^{-1} \sum_{q=0}^\infty \theta^{i\,(p,q)} \conf^p ( \conf \log \conf)^q +
\theta^{\prime\,i}_{\rm
  finite}, \\
\label{e3}
{\cal L} &=&
\sum_{p=-t}^{-1} \sum_{q=0}^\infty {\cal L}^{(p,q)} \conf^p ( \conf \log \conf)^q +
{\cal L}_{\rm
  finite}, 
\eea
\ees
where the coefficients depend only on $x^i$ and not on $x^0 = \conf$.
The integrals in Eqs.\ (\ref{translated}) can conveniently be
evaluated by assuming similar expansions for the integrals and
equating coefficients of $\conf^p ( \conf \log \conf)^q$, which yields
recursion relations that can be solved.
This yields
\begin{eqnarray}
\Lambda^{i} &=&   \sum_{\substack{k=-t+1 \\ k\ne 0}}^\infty
\sum_{p=-t+1}^{{\rm min}(k,-1)} \sum_{j=-t}^{p-1} \frac{1}{k} \left\{\prod_{l=j+1}^{p-1}
\left(\frac{l}{k}-1\right) \right\}
  {\theta}^{i\,(j,k-j-1)}
\conf^p (\conf \log \conf)^{k-p}  \nonumber \\
&&- \sum_{p=-t}^{-1} \frac{1}{p}  {\theta}^{i\,(p,-p-1)}
\conf^p (\conf \log \conf)^{-p}
\label{ans1}
\end{eqnarray}
and 
\begin{eqnarray}
B^{0} &=&
- \sum_{\substack{k=-t+1 \\ k\ne 0}}^\infty
\sum_{p=-t+1}^{{\rm min}(k,-1)} \sum_{j=-t}^{p-1} \frac{1}{k} \left\{\prod_{l=j+1}^{p-1}
\left(\frac{l}{k}-1\right) \right\} {\cal L}^{(j,k-j-1)}
\conf^p (\conf \log \conf)^{k-p} 
\nonumber \\
&& + \sum_{p=-t}^{-1} \frac{1}{p}
{\cal L}^{(p,-p-1)} \conf^p (\conf
\log \conf)^{-p}.
\label{ans2}
\end{eqnarray}
Here we have effectively chosen $\upsilon_0=\infty$ for terms in
the integrands with $p+q<-1$, $\upsilon_0=1$ for $p+q=-1$, and
$\upsilon_0 =0$ for $p+q>-1$.
These choices again make the second term in Eq.\ (\ref{theta0finite1})
effectively vanish\footnote{Note that the condition $\delta L'_{\rm div} =
d \theta'_{\rm div}$ is not satisfied in this case due to mixing with
the finite terms, but one can
directly check that Eq.\ (\ref{theta0finite2}) is still valid.}, so
we again recover the result (\ref{theta0finite2}).

\subsubsection{Anomalies of renormalized symplectic form and Lagrangian}

Although the renormalized Lagrangian $L'_{\rm ren}$ and renormalized symplectic potential $\theta'_{\rm ren}$ are finite,
they are no longer covariant, assuming one with starts with covariant quantities.  In this section we will discuss
the corresponding anomalies, which are computed explicitly in
Appendix \ref{app:hr}.  This will yield
the second result described above, that 
the dependence of the pullback of the renormalized symplectic
potential on the background structures arises only through
corner and boundary terms.

We will show that the anomalies in $B$ and $\Lambda$ are of the form
\bes
    \label{conj}
    \bea
    \label{conj1}
    \Delta_{\hat \xi} B &=&     (\Delta_{\hat \xi} B)_{\rm finite} + d
    \kappa_\xi, \\
    \label{conj2}
    \Delta_{\hat \xi} \Lambda &=&     (\Delta_{\hat \xi} \Lambda)_{\rm finite} + \delta
    \kappa_\xi - d \mu_\xi- I_{\hat {\delta \xi}} \Lambda,
\eea
\ees
where the indicated quantities are finite and $\kappa_\xi$ and
$\mu_\xi$ are quantities that in general can diverge on the boundary.
Inserting these expressions into Eqs.\ (\ref{JKM11})
for the renormalized Lagrangian and symplectic potential, and acting
with the anomaly operator yields for field-independent symmetries
\bes
\bea
\Delta_{\hat \xi} L'_{\rm ren} &=& d ( \Delta_{\hat \xi} B)_{\rm finite}, \\
\Delta_{\hat \xi} \theta'_{\rm ren} &=& \delta ( \Delta_{\hat \xi} B)_{\rm
  finite} - d (\Delta_{\hat \xi} \Lambda)_{\rm finite}.
\eea
\ees
If we now take a pullback to the boundary, and use the fact that the
pullback operator commutes with the anomaly operator and the spacetime
and phase space exterior derivatives $d$ and $\delta$, we obtain
\be
\Delta_{\hat \xi}\p\theta'_{\rm ren} = \delta \, \pi_{0}^{\ *} ( \Delta_{\hat \xi} B)_{\rm
  finite} - d \, \pi_{0}^{\ *} (\Delta_{\hat \xi} \Lambda)_{\rm finite}.
\label{rightform}
\ee
Here on the right hand side we have denoted the pullback by
$\pi_{0\,*}$ instead of using our usual boldface notation.

The explicit expressions for the quantities appearing in
the anomalies (\ref{conj}) are as follows.  The finite pieces are given by
\bes
\bea
i_v ( \Delta_{\hat \xi} B)_{\rm finite} &=& 0, \\
\pi_{\upsilon}^* ( \Delta_{\hat \xi} B)_{\rm finite} &=&
\pi_{\upsilon_0}^* i_\xi L', \\
i_v ( \Delta_{\hat \xi} \Lambda)_{\rm finite} &=& 0, \\
\pi_{\upsilon}^* ( \Delta_{\hat \xi} \Lambda)_{\rm finite} &=&
-\pi_{\upsilon_0}^* \left[ i_v \theta' \lie_\xi \Phi \right].
\eea
\ees
In coordinate notation these relations are
\bes
\bea
\label{wq1}
( \Delta_{\hat \xi} B)^i_{\rm finite} &=& 0, \\
\label{wq2}
( \Delta_{\hat \xi} B)^0_{\rm finite}(\upsilon,x^j) &=&
(\xi^0 {\cal L})(\upsilon_0,x^j), \\
\label{wq3}
( \Delta_{\hat \xi} \Lambda)^{ij}_{\rm finite} &=& 0, \\
\label{wq4}
( \Delta_{\hat \xi} \Lambda)^i_{\rm finite}(\upsilon,x^j) &=&
-(\xi^0 \theta^{\prime\,i})(\upsilon_0,x^j), 
\eea
\ees
where we have decomposed the symmetry generator as ${\vec \xi} = \xi^0
\partial_0 + \xi^i \partial_i$.
The quantities $\kappa_\xi$ and $\mu_\xi$ are given by
\bes
\bea
i_v \kappa_\xi &=& 0,\\
\pi_{\upsilon}^* \kappa_\xi &=&
\int_\upsilon^{\upsilon_0} d {\bar \upsilon} \, \pi_{{\bar
    \upsilon}}^* i_\xi i_v L'
- i_{\pi_{\upsilon}^* \xi} \int_\upsilon^{\upsilon_0} d {\bar \upsilon} \, \pi_{{\bar
    \upsilon}}^* i_v L',\\
i_v \mu_\xi &=& 0,\\
\pi_{\upsilon}^* \mu_\xi &=&
\int_\upsilon^{\upsilon_0} d {\bar \upsilon} \, \pi_{{\bar
    \upsilon}}^* i_\xi i_v \theta'
- i_{\pi_{\upsilon}^* \xi} \int_\upsilon^{\upsilon_0} d {\bar \upsilon} \, \pi_{{\bar
    \upsilon}}^* i_v \theta'.
\eea
\ees
In coordinate notation these relations are
\bes
\label{kappamuans}
\bea
\label{kappaxic}
\kappa_\xi &=& \left( \xi^i \int_{\upsilon}^{\upsilon_0} {\cal
  L} \, d {\bar
  \upsilon}  
- \int_{\upsilon}^{\upsilon_0} 
\xi^i {\cal L} \, d {\bar \upsilon}
\right) \varpi_i,
\\
\mu_\xi &=& - \frac{1}{2} \left[ \xi^i
  \int_{\upsilon}^{\upsilon_0} \theta^{\prime\,j} \, d {\bar \upsilon}-
  \xi^j \int_{\upsilon}^{\upsilon_0} \theta^{\prime\,i} \, d {\bar \upsilon}
  - \int_{\upsilon}^{\upsilon_0} ( \xi^i \theta^{\prime\,j} - \xi^j
  \theta^{\prime\,i}) \, d {\bar \upsilon}
  \right] \varpi_{ij},
\eea
\ees
where the integrands are evaluated at $\upsilon= {\bar \varepsilon}$.

\subsection{Renormalization of corner terms}
\label{exparg1}

In this section we show that it is always possible to find
a corner canonical transformation (\ref{cct}) that makes the 
  the integrated action $S$, corner flux $\varepsilon$ as well as
  $\ell' + d c'$ finite.

The basic idea is the trivial integral identity, for any function
${\cal L}(u,v)$ of two variables $u,v$:
\begin{eqnarray}
\int_{u_0}^\infty d {\bar u} \int_{v_0}^\infty d{\bar v} {\cal
  L}({\bar u},{\bar v}) &=& \int_{u}^\infty d {\bar u} \int_{v}^\infty d{\bar v} {\cal
  L}({\bar u},{\bar v}) - \int_u^{u_0} d {\bar u} \int_v^\infty d{\bar
  v} {\cal L}({\bar u},{\bar v})
\nonumber \\ &&
- \int_u^\infty d {\bar u} \int_v^{v_0} d{\bar
  v} {\cal L}({\bar u},{\bar v})
+ \int_u^{u_0} d{\bar u} \int_v^{v_0} d{\bar v} {\cal L}({\bar
  u},{\bar v}).
\end{eqnarray}
Here the first term on the right hand side will be the bulk action,
which can diverge as $v \to 0$ or $u \to 0$ (we assume that the
integrals are finite at large $u$ and $v$). The second and third terms
on the right hand side are boundary terms that are added at the
boundaries $u  = $ constant and $v = $ constant
(Eq.\ (\ref{translated}) above).  Finally the fourth term is the corner
term that when added makes the total integral manifestly finite in the
limit $u\to 0$, $v \to 0$, as can be seen from the left hand side.

We now translate this idea into a covariant notation and add the
additional dimensions which were suppressed in the above argument.
Consider two boundaries ${\cal N}$ and ${\tilde {\cal N}}$ which
intersect in a $(d-1)$-surface ${\cal C}$.  In Sec.\ \ref{exparg} above
above we introduced a 
vector field ${\vec v}$ and a diffeomorphism $\pi_v$ which moves
points $v$ units along integral curves of ${\vec v}$. We also
introduced a coordinate $v$ which vanishes on ${\cal N}$ and for which
$v^a \nabla_a v = 1$.

We now introduce the additional background structure of a nowhere
vanishing vector field ${\vec u}$ on the boundary ${\cal N}$.  We can
then introduce a coordinate
$u$ on ${\cal N}$ by $u=0$ on ${\cal C}$ and $u^a \nabla_a u=1$ on
${\cal N}$.  We can then extend the definitions of ${\vec u}$ and $u$
off ${\cal N}$ by Lie transporting with respect to ${\vec v}$.
Finally we define the diffeomorphism ${\tilde \pi}_u$ to be the map
that moves points $u$ units along integral curves of ${\vec u}$.

Using this notation, the boundary term $B$ associated with the
boundary ${\cal N}$ is given by Eqs.\ (\ref{bprimelambdaprime}) and (\ref{translated}) above.
There is an analogous boundary term ${\tilde B}$ for the boundary
${\tilde {\cal N}}$, given by
\bes
\label{bprimelambdaprime1}
\bea
i_u {\tilde B} &=& 0, \\
{\tilde \pi}_{u}^* {\tilde B} &=&  \int_{u}^{u_0} d{\bar
  u} \, {\tilde \pi}_{{\bar u}}^* i_u L'.
\eea
\ees
Finally the corner term $c$ on ${\cal C}$ required to make the total
action finite is
given by
\bes
\label{bprimelambdaprime2}
\bea
i_u c &=& 0, \\
i_v c &=& 0, \\
{\tilde \pi}_{u}^* \pi_{v}^*
c &=&
\int_{u}^{u_0} d{\bar  u}
\int_{v}^{v_0} d{\bar  v}
\, {\tilde \pi}_{{\bar u}}^* \, \pi_{{\bar v}}^*
i_u i_v L'.
\eea
\ees

\section{Vacuum general relativity at future null infinity: gravitational charges}
\label{expcalc}

In this section, to explicitly demonstrate the holographic
renormalization procedure described in section~\ref{exparg}, 
we specialize to the case of future null infinity in vacuum general relativity in four dimensional
asymptotically flat spacetimes, and
 to the generalized BMS (GBMS) field configuration space \eqref{eq:fcs1}. We then compute the
renormalized symplectic potential and the localized charge using the results laid out in 
section~\ref{exparg} and section~\ref{sec:charges}, with the conformal factor and rigging vector field
taken to be those associated with Bondi-Sachs coordinates. 
The results obtained in this section 
agree with expressions for the charges obtained previously in,
e.g., \cite{Compere:2018ylh, FN}, explicitly demonstrating the utility of the 
holographic renormalization algorithm described in section \ref{exparg}. 

We identify the coordinate system  discussed in section~\ref{exparg} 
with 
the Bondi-Sachs coordinates according to
$(x^0,x^1,x^2,x^3)=(\conf,u,\theta,\phi)\equiv(\conf,u,x^{A})$. 
The physical metric in these coordinates corresponds to the line element \cite{Compere:2018ylh, FN} \footnote{We use the symbol ``$F$'' here in place of the symbol $\beta$ that is used in \cite{Compere:2018ylh,FN} to avoid confusion with our symbol for the corner term in decomposition of the symplectic potential [as in \eqref{eqn:thdecomp}].} 
\be \label{eq:phys-line-element}
d \tilde{s}^{2} =-U e^{2 F} d u^{2}+2 e^{2 F} \conf^{-2} d u\, d
\conf+\conf^{-2} h_{A B}\left(d x^{A}-U^{A} d u\right)\left(d
x^{B}-U^{B} d u\right).
\ee
Here we have used $\conf = 1/r$ in place of the radial
coordinate $r$ used in \cite{Compere:2018ylh, FN} (see, e.g, the discussion in \cite{GPS} for more
details on the construction of the Bondi-Sachs coordinate system in the physical spacetime). In these coordinates, $\conf=0$ corresponds to $\scp$, and the null generator of $\scp$ is given by $n^{i} \hateq (\partial_{u})^{i}$. Moreover, in these coordinates, we have a foliation of $\scp$ by cross-sections of constant $u$. The 1-form on $\scp$ normal to this foliation is given by $l_{i}:=-\nabla_{i} u$.  The functions appearing in \eqref{eq:phys-line-element} are smooth in their dependence on $(\conf, u, x^{A})$ and their expansions in powers of $\conf$, after imposing the Einstein equations, are 
given by \cite{Compere:2018ylh}
\begin{subequations} 
\begin{align}\label{eq:bondi-metric-exp}
U&=\frac{1}{2} \mathcal{R}-2 \conf M+O\left(\conf^{2}\right)\,, \\
F &=-\frac{1}{32} \conf^{2} C^{A B} C_{A B}+O\left(\conf^{3}\right)\,, \\
U^{A} &=-\frac{1}{2} \conf^{2} \mathscr{D}_{B} C^{A B}+2 \conf^{3} L^{A}+O\left(\conf^{4}\right)\,, \\
h_{A B} &=q_{A B}+\conf C_{A B}+\frac{1}{4} \conf^{2} q_{A B} C^{C D} C_{C D}+O\left(\conf^{3}\right) \label{eq:sphere-met-expansion}\,,
\end{align}
\end{subequations}
where $\mathcal{R}$ denotes the  Ricci scalar of the leading order
sphere metric $q_{AB}$, and $\sqrt{q}$ denotes the square root of its
determinant. In addition, $\ms{D}_{A}$ is the derivative operator
compatible with $q_{AB}$. Moreover, $C_{AB}$ is ($-2$ times) the shear associated with the auxiliary normal $l^{a}:=-\conf^{2} \tilde{g}^{ab}\nabla_{a} u$  and satisfies $q^{AB} C_{AB} =0$ as well as $\delta (q^{AB}
C_{AB})$ for all perturbations. Furthermore, $M$ denotes the Bondi
mass aspect and $L^{A}$ is related to the angular momentum
aspect.\footnote{The angular momentum aspect $N_{A}$, defined in
\cite{FN}, is related to $L_{A}$ by $N_{A} = -3 L_{A} + \frac{3}{32}
\ms{D}_{A} (C_{BC} C^{BC}) + \frac{3}{4} C_{A}{}^{B} \ms{D}^{C}
C_{BC}$.} Note also that it follows from \eqref{eq:fcs1} that $\delta n^{i}= \delta \sqrt{q} =0$. Finally, capital Roman indices are raised and lowered using
$q_{AB}$ throughout this section.\\

Next, we compute the symplectic potential $\theta'$, which we take to be the spacetime covariant one given in Eq.\ (39) of Ref.\ \cite{Wald:1999wa} and so here  $\theta' = \cov\theta$. We will also take $L'$ to be the covariant Einstein-Hilbert Lagrangian, so $L' = \cov L$.  We make use of the Einstein equations for the background metric, the linearized Einstein equations for the perturbations and the Bondi condition (\ref{bondicondition}). We find that the symplectic potential diverges as $\conf^{-2}$ and that there are no logarithmic divergences. In particular, the divergent pieces  are given in the notation of (\ref{ex}) by \footnote{We omit writing the explicit expressions for $\theta_{A}^{(-2,0)}$ and $\theta_{A}^{(-1,0)}$ since they will not be needed for the explicit charge calculation later in this section.} (setting $16\pi G = 1$) 

\begin{align} \label{eq:thetadiv}
    \theta_{0}^{(-2,0)}&=0\,, \, \,\,\,\,\,\,\,\quad \quad \theta_{0}^{(-1,0)}=\sqrt{q} \big(-\delta \mathcal{R}- \frac{1}{2} N_{AB} \delta q^{AB}\big)\,, \notag \\
    \theta_{1}^{(-2,0)} &= - \frac{\sqrt{q}}{2} C_{AB} \delta q^{AB}\,, \quad  \theta_{1}^{(-1,0)}=0\,.
\end{align}
 Note that here  $\delta q^{AB}$ is the variation of the inverse metric, $q^{AB}$, and not
$q^{AC} q^{BD}\delta q_{CD}$. In addition, we have
\begin{equation} \label{eq:finpiece}
    \theta_{0}^{(0,0)} = \sqrt{q} \big[2 \delta M\, +  2 \partial_{u} \delta Z\,  + \delta (\ms{D}_{A} \ms{U}^{A})\, +\, \frac{1}{2} N_{AB} \delta C^{AB} -\frac{1}{4} \mathcal{R} C_{AB} \delta q^{AB} - \ms{D}_{A} \ms{U}_{B} \delta q^{AB}\big]\,.
\end{equation}
Note also that here
\begin{equation}\label{eqn:NABdef}
  N_{AB}:=\partial_{u} C_{AB}\,, \,  Z:=-\frac{1}{32} C_{AB} C^{AB}\,, \, \ms{U}^{A}:=-\frac{1}{2} \ms{D}_{B} C^{AB} \,.
\end{equation}
 We now consider the renormalization of the symplectic potential. Using the expression \eqref{ans1} specialized to $t=2$, $d=3$ with the coefficients of the logarithmic terms taken to vanish, we have that
 \begin{align} \label{eq:lambdapvacGR}
\Lambda &=   -\left[ \log \conf  \, \theta^{(-1,0)}_1 -  \conf^{-1} \theta_1^{(-2,0)} \right]
  dx^2 \wedge dx^3 -  \left[ \log \conf  \, \theta^{(-1,0)}_3 -  \conf^{-1} \theta_3^{(-2,0)} \right]
  dx^1 \wedge dx^2  \notag \\
 & - \left[ \log \conf  \, \theta^{(-1,0)}_2 -  \conf^{-1} \theta_2^{(-2,0)} \right]
  dx^3 \wedge dx^1.
\end{align}
Note that since we have taken $L'= \cov{L}$ and $\theta'= \cov{\theta}$, in carrying out the renormalization, we pick $\Lambda=-\lambda'$ and $B=b'$ [see \eqref{eqn:covL}, \eqref{eqn:covth} and compare with \eqref{Lren}, \eqref{thetaren}]. Note also that in vacuum general relativity with zero cosmological constant, the Lagrangian vanishes on shell. For that reason, it does not need to be renormalized, and so we take $b'= B=0$. Using the explicit expressions for the (unrenormalized) symplectic potential along with the linearized Einstein equations to compute $\theta'_{\text{ren}} = \theta' - d\Lambda$ [\eqref{thetaren} with $B=0$], we find that the effect of adding $d\Lambda$ is to cancel the diverging pieces in each component of $\theta'$ while leaving the finite pieces unchanged. Moreover, the pullback of the renormalized symplectic potential to $\scp$ is given by (\ref{eq:finpiece}).\\

Note that our expression for the symplectic potential and the subsequent renormalization procedure, when implemented in Bondi coordinates, coincide with those in \cite{Compere:2018ylh}. Note also that even though we have demonstrated our renormalization procedure for  conditions that correspond to the generalized BMS configuration space, subject to the Bondi condition, the procedure itself is completely general and can be applied to any of the extensions of the BMS algebra discussed in section~\ref{sec:compendium}, with or without the Bondi condition. It is guaranteed to work in any of these cases using the general algorithm described in section \ref{exparg}.\\

Having obtained an expression for the pullback of the (renormalized) symplectic potential, we now seek to obtain a decomposition of it into a boundary term, a corner term, and a flux term in keeping with (\ref{eqn:thdecomp}), that is, a decomposition of the form

\be \label{eq:theta-decom-2}
\underline{\theta}'_{\text{ren}} \hateq -\delta \ell' + d\beta' + \beom\,.
\ee
Comparing this with \eqref{eq:finpiece}, suggests the following choice for the flux term $\beom$
\be \label{eq:flux-in-Dirichlet}
\beom =  -\eta\bigg[\frac{1}{2} N_{AB} \delta C^{AB} -\frac{1}{4} \mathcal{R} C_{AB} \delta q^{AB} -\ms{D}_{A} \ms{U}_{B} \delta q^{AB}\bigg]\,,
\ee
where $\eta$ is the volume element on $\scp$ given by $\eta = -\sqrt{q} du \wedge d\theta \wedge d\phi$. Moreover, we can read off that
\be \label{eq:bplp}
\ell'=\eta \bigg[ 2 M + 2 \partial_{u} Z +\ms{D}_{A} \ms{U}^{A}\bigg]\,, \quad  \beta' = 0\,,
\ee
where we have used the fact that $\partial_{u} \delta Z = \lie_{n} \delta Z = \delta \partial_{u} Z$ since $\delta n^{i} =0$.

Note that while $C_{AB}$ is not an intrinsic quantity on $\scri^{+}$, its variation still occurs in our expression for the flux in \eqref{eq:flux-in-Dirichlet}.
This appears to be at odds with the Dirichlet form of the flux advocated for in this 
work, since $C_{AB}$ is related to the extrinsic geometry of $\scp$ with respect to the 
auxiliary null direction $l^a$. In asymptotically dS or AdS spaces, the equations
of motion allow one to solve for $C_{AB}$ in terms of the leading metric $q_{AB}$ at
$\scp$
\cite{Poole:2018koa, Compere:2019bua, Barnich2010}, and hence flux terms involving $\delta C_{AB}$ 
are still consistent with Dirichlet form.  This is no longer the case in
asymptotically flat spacetimes, in which $C_{AB}$ represents free data on 
$\scp$.  Nevertheless, from \eqref{eq:sphere-met-expansion} we see that 
$C_{AB}$ is a subleading component of the spherical part of the metric, $h_{AB}$, 
which is an intrinsic quantity on each $\Phi=\text{const.}$ surface, which limit
to $\scp$.  It is therefore not entirely surprising that $C_{AB}$ appears as a configuration
variable in the expression for the flux.  
Furthermore, the news tensor $N_{AB}$ that appears conjugate 
to $C^{AB}$ in the expression of the flux is given by the 
$u$-derivative of $C_{AB}$ according to 
(\ref{eqn:NABdef}), as one would expect of a momentum
variable, lending additional support to interpreting 
(\ref{eq:flux-in-Dirichlet}) as the appropriate analog of the 
Dirichlet form of the flux.
An interesting question for future work would be to understand better the principle 
for selecting a preferred form of the flux for asymptotically null surfaces, rather
than postulating the form as is done in this section.

We now proceed to calculate 
\be \label{eq:integratedcharge}
H_\xi = \lim_{S'\to S} \int_{S'} h_\xi\,,
\ee
where, $S'$ here denotes $u=const$ cross-sections  of a one-parameter family of $\conf=const$ surfaces that limit to $\scp$ in the unphysical spacetime. As denoted above, to define the charge, $H_{\xi}$, on a cross-section, $S$, of $\scp$, we will perform this integral and then take the limit  $S'\to S$. We calculate (\ref{eqn:Hxi}),
where $\xi^{a}$ for a generalized BMS vector field is given by \cite{Compere:2018aar}

\begin{align} \label{eq:vec-field}
\xi^{a} &= f \partial_{u} + \bigg[Y^{A}- \conf \ms{D}^{A} f + \frac{1}{2} \conf^{2} C^{AB}\ms{D}_{B}f + O(\conf^{3})\bigg]  \partial_{A} \notag\\
&+ \conf^{2}\bigg[\frac{1}{2} \conf^{-1} \ms{D}_{A} Y^{A} -\frac{1}{2} \ms{D}^{2} f - \frac{1}{2} \conf \ms{U}^{A} \ms{D}_{A} f + \frac{1}{4} \conf \ms{D}_{A} (\ms{D}_{B} f C^{AB}) + O(\Phi^{2})\bigg]\partial_{\conf}\,.
\end{align}
Here $f(u,x^{A})= \gamma(x^{A}) + \frac{1}{2} u \ms{D}_{B} Y^{B}(x^{A})$ where $\gamma(x^{A})$ is the supertranslation function and $Y^{A}$ is the generator of arbitrary smooth diffeomorphisms on $S^{2}$. Moreover, $Q_{\xi}'$ is given by \eqref{eq:Qprimexi} and $\mcov Q_{\xi}$ in vacuum general relativity is given by (where recall that we have set $\frac{1}{16 \pi G}=1$) 
\be
\mcov {Q_{\xi}}_{ab} = -\tilde{\epsilon}_{cdab} \tilde{\nabla}^{c} \xi^{d}\,.
\ee
Using \eqref{eq:thetadiv}, \eqref{eq:lambdapvacGR}, \eqref{eq:vec-field}) and the fact that 
\be
I_{\hat{\xi}} \delta q^{AB}=  - \ms{D}^{A} Y^{B} - \ms{D}^{B} Y^{A} + q^{AB} \ms{D}_{C} Y^{C}\,,
\ee
we find that
\begin{align} \label{eq:contr-lambdap}
I_{\hat{\xi}} \lambda'= I_{\hat{\xi}} \Lambda= -\conf^{-1} \, \mu\, C_{AB} \ms{D}^{A} Y^{B} + \cdots \,,
\end{align}
where $\cdots$ denotes terms that vanish upon pullback to $S'$ and are hence not relevant for the calculation of the charge. Note also that $\mu = -i_{n} \eta = \sqrt{q}\, d\theta \wedge d\phi$. Then, explicitly calculating $\mcov Q_{\xi}$, one finds that its pullback to $S'$ has a piece that diverges as $\Phi \to 0$ that is given by $\Phi^{-2} \mu\, \ms{D}_{A} Y^{A}-\Phi^{-1}\, \mu\, Y^{A} \ms{D}^{B} C_{AB}$. The first term is a total derivative which drops out of the integral over $S'$. Moreover, the second term cancels with \eqref{eq:contr-lambdap} up to a total derivative term. We therefore see that upon integrating over $S'$, the diverging piece drops out of $\int_{S'}Q_{\xi}'$. Taking the limit $S'\to S$ we then obtain
\be
\int_{S} Q_{\xi}'= -\int_{S} \mu \bigg[2 f \partial_{u} Z - 2 f M - \ms{U}^{A} \ms{D}_{A} f - 2 Y^{A} \{N_{A} -\frac{1}{4} C_{A}{}^{B} \ms{D}^{C} C_{BC} -\frac{1}{16} \ms{D}_{A} (C_{BC} C^{BC})\}\bigg]\,.
\ee
Using this in addition to \eqref{eq:bplp} and \eqref{eq:vec-field} to compute \eqref{eqn:Hxi} and dropping total derivative terms, we find that the final expression for $H_{\xi}$ is given by

\be \label{eq:finalGBMScharge}
H_{\xi} = \int_{S} \mu \bigg[  4 f M + 2 Y^{A} \{N_{A} -\frac{1}{4} C_{A}{}^{B} \ms{D}^{C} C_{BC} -\frac{1}{16} \ms{D}_{A} (C_{BC} C^{BC})\}\bigg]\,.
\ee

This expression is the same as the one derived for the (usual) BMS charge in, for example, \cite{FN,GPS,Barnich:2011mi} and is also consistent with the expression for the charge given in (9.21) of \cite{freidel2021weyl}. It was pointed out in \cite{Compere:2018ylh} that this expression diverges in limits to the end-points of $\scp$ (i.e as $u \to \pm \infty$) when one allows for the most general fall-offs in $u$ of  $C_{AB}$: $C_{AB} = O(u)$, that are compatible with the action of the GBMS group on the boundary fields. To cure these ``corner'' divergences, one would have to implement an additional  renormalization step, similar in spirit to the one discussed in section \ref{exparg}. Presumably, one would have to add to the expression for $\Lambda$ in \eqref{eq:lambdapvacGR} terms that are finite as $\conf \to 0$ but which diverge as $u \to \pm \infty$. This would modify the expression for $\beta'$ after which one would have to pick an expression for $\upsilon$ (see \eqref{eq:betpdecomp} and the discussion around it)  which would lead to a different expression for the charge, $H_{\xi}$. However, addressing this issue is beyond the scope of the present paper and we leave it to future work. Indeed, it would be interesting to carry out these steps to attempt to derive the GBMS charge expression in (5.49) of \cite{Compere:2018ylh} which does not have the aforementioned divergences.

We note that the decomposition we picked in \eqref{eq:bplp} was not unique even after having picked the expression for the flux, $\beom$, which we take to be given by \eqref{eq:flux-in-Dirichlet}. Instead of the choice made in \eqref{eq:bplp}, one could instead have picked
\be
\ell'=\eta \bigg[2 M + \ms{D}_{A} \ms{U}^{A}\bigg]\,, \, \,\beta' = 2 \delta Z\mu \,.
\ee
Also, because of the Bondi condition, $(\ms{D}_{A} \ms{U}^{A}) \eta = d(i_{\ms{U}} \eta)$ and therefore, one could also consider a decomposition of \eqref{eq:finpiece} whereby
\be
\ell' =\eta \bigg[2 M +2\partial_{u}Z\bigg]\,, \, \,\beta' = -\delta ( i_{\ms{U}} \eta) \,.
\ee
where we have defined a vector $\ms{U}^{i}$ on $\scp$ such that  $\ms{U}^{i} l_{i} =0$ and $\ms{U}^{A} = \ms{U}^{i} e_{i}{}^{A}$ where $e_{i}{}^{A}$ is a projector onto angular directions. Finally, one could also consider
\be
\ell'= 2 \eta M\,, \, \,\beta' = -\delta ( i_{\ms{U}} \eta - 2 Z \mu) \,.
\ee

  To resolve the ambiguity between these choices, as described in section.~\ref{sec:corners}, one needs to implement a corner improvement
where one looks for a decomposition of $\beta'-\lambda'$ of the form [see \eqref{eqn:bldecomp}]
\be \label{eq:betpdecomp}
\beta'-\lambda'= -\delta c' + d\gamma '+ \varepsilon  \,.
\ee
As described in \eqref{eqn:hxiimproved}, the improved expression for the charge density is given by 
$\tilde h_{\xi} =  h_{\xi} - \Delta_{\hat{\xi}} c'$.
To get unambiguous charges,\footnote{The $\gamma'$ term in \eqref{eq:betpdecomp} only contributes an exact piece to the charge density and therefore, in the present context, its choice does not effect the charge. We therefore pick $\gamma'=0$ here for convenience.}
one needs to fix an expression for $\varepsilon$ in addition to the expression for $\beom$ which we fixed to be given by \eqref{eq:flux-in-Dirichlet}. Here, we pick $\varepsilon=-\lambda'$.\footnote{Note from \eqref{eq:lambdapvacGR} that $\lambda'=\Lambda$ is actually divergent on $\scp$ and so really this decomposition is done on a cutoff surface \emph{near} $\scp$ after extending $\beta'$ arbitrarily away from $\scp$. Obtaining a finite corner flux \emph{on} $\scp$ would entail a more careful analysis of corner terms in the action of vacuum general relativity which we leave to future work.}\textsuperscript{,}\footnote{Note also that to ensure finiteness of the charge in the $u \to \pm \infty$ limits described earlier (an issue we have chosen to ignore here), one would need to pick a different expression for $\varepsilon$. Presumably, this would follow from a more careful analysis of the corner terms in the action and the resulting charge expression \textit{will} indeed be modified in that case.} It is then easy to check that calculating the charge in the same way as before but with $h_{\xi}$ replaced with $\tilde h_{\xi}$, the final charge expression remains unchanged and is (still) given by \eqref{eq:finalGBMScharge}. This demonstrates that as highlighted earlier in the paper, fixing an expression for the flux terms in the problem, on the boundary as well as the corners, ($\beom$ and $\varepsilon$ in this case), gives us an unambiguous expression for the charge.

\section{Discussion}
\label{sec:discussion}

We conclude with a discussion of a number of future directions and applications 
of the present work.

\subsection{More general asymptotic symmetries}

The holographic renormalization argument presented in section
\ref{sec:hr} demonstrates that all asymptotic charges can be rendered 
finite once appropriate counterterms have been found to produce 
a finite renormalized gravitational action.  This holds without
imposing asymptotic boundary conditions on the dynamical fields, and hence motivates
exploring formulations of the theory in which the standard boundary
conditions are relaxed.  Indeed, the arguments of section \ref{sec:hr}
were inspired by similar considerations for asymptotically anti-de Sitter
spacetimes \cite{Compere:2008us} in which the standard Dirichlet boundary
condition was relaxed. This produces an enlarged asymptotic symmetry
algebra for these spacetimes, which have been 
further explored in the works on the $\Lambda$-BMS group 
\cite{Compere:2019bua, Compere:2020lrt}.  
In the past, finiteness of the action and charges has been
suggested as a reason for selecting  boundary conditions for the theory, but 
the analysis of the present work suggests that this is unnecessary, since 
finiteness can instead be achieved through holographic renormalization.  
The only reason for imposing boundary conditions should be to obtain a 
well-defined variational principle, or, equivalently, to ensure the phase 
space describes a closed system that does not lose symplectic flux through its boundary.

Relaxing the standard boundary conditions of four dimensional asymptotically
flat spacetimes leads to the enlarged symmetry algebras discussed in section
\ref{sec:compendium}.  Each of the symmetry groups described there still fixes some
structure at null infinity, but since holographic renormalization applies in the 
absence of any such boundary condition, it is tempting to propose an even more general
set of symmetries.  These would be obtained by relaxing the final condition
leading to the Weyl BMS configuration space (\ref{eq:fcs11}), namely, not imposing
$n^i$ be fixed.  We would expect to obtain in this manner all diffeomorphisms of 
$\scp$ as asymptotic symmetries, and it would be interesting to compute 
expressions for the associated charges.\footnote{The appearance of 
$\text{Diff}(\scri)$
has also been suggested to appear in the context of asymptotically de-Sitter and 
anti-de Sitter spaces in \cite{Ashtekar:2014zfa}.}  
The enlarged algebra may also be related to the extended symmetries 
of finite null surfaces explored in \cite{Adami:2020amw}.

Another context in which  extended symmetries can arise
is in higher dimensional asymptotically flat spacetimes.  In higher than four 
spacetime dimensions, there exist consistent boundary conditions that eliminate
the supertranslations as asymptotic symmetries.  However, in light of the relation
between supertranslations and the Weinberg soft graviton theorems 
\cite{Weinberg:1965nx, Strominger:2013jfa, He:2014laa}, which hold
in all dimensions, it is desirable to find a phase space in higher dimensions that 
admits a nontrivial action of supertranslations.  Such relaxed boundary conditions
have been explored in 
\cite{Kapec:2015vwa, Pate:2017fgt, Campoleoni:2020ejn}, 
and the general holographic renormalization argument suggests that 
a phase space can be constructed in which these transformations produce finite 
charges.\footnote{These relaxed boundary conditions have been questioned in
\cite{Hollands:2016oma} on the grounds of not leading to 
finite fluxes through $\scp$, but such divergences can always be handled by
the procedure of holographic renormalization, at the expense of introducing
some dependence on a background structure. 
See, for example, \cite{Aggarwal:2018ilg}.}
It would be interesting to carry out the holographic renormalization procedure 
in these higher dimensional cases and to construct the phase space on which the 
renormalized BMS charges are defined, as well as to 
obtain charges associated with higher dimensional
versions of the symmetry algebras described in section \ref{sec:compendium}.
Some ideas in this direction have been explored in \cite{Colferai:2020rte,
Fiorucci:2020xto, Capone:2019aiy, Capone:2021ouo}.

A final application would be to investigate the recently proposed $w_{1+\infty}$ 
symmetry of 4D asymptotically flat gravity, which was derived at the level
of celestial amplitudes
\cite{Strominger:2021lvk}.  An interesting question to address is whether the charge
generators of this algebra arise from asymptotic diffeomorphisms, to give 
a spacetime interpretation of the symmetry transformations.  The holographic
renormalization procedure in the present work provides an ideal framework for
investigating this question.

\subsection{Gluing and quantization}
\label{sec:gluing}

One of the main motivations for considering localized charges is to understand 
the embedding of the localized phase spaces and their observables into the global
phase space of the theory.  In the classical context, understanding this embedding
can help give meaning to quasilocal notions of energy, which are relevant in 
practice since astrophysical processes are usefully described using local descriptions
of objects' locations and momenta, despite the fact that local observables are
nonperturbatively ill-defined in a diffeomorphism-invariant theory.  
There is a natural construction known as Marsden-Weinstein symplectic reduction 
\cite{Marsden:1974dsb}
by which localized phase spaces can be assembled into
a global phase space, ensuring in the process that the localized charges 
become trivial, as would be expected for charges associated with 
a gauge symmetry.  
This application of symplectic reduction to the gluing of local phase spaces
was discussed in the work of Donnelly and Freidel \cite{Donnelly2016a}.
The idea is to take two adjacent localized
phase spaces $\ps_1$ and $\ps_2$, each containing a set of charges $H_\xi^i$, $i=1,2$, 
associated with diffeomorphisms that act at their common boundary.  One then
constructs the product phase space $\ps_{12} = \ps_1 \times \ps_2$, which 
also admits an action of the boundary symmetry, generated 
by the sum of the  the individual charges, $H_\xi^\text{tot}
= H_\xi^1 + H_\xi^2$.
The reduced phase space is obtained by then restricting to the submanifold of zero
total charge $H_\xi^\text{tot} =0$, and further quotienting by the flow generated by
the charges within this submanifold.  This two-step process results in 
a new phase space $\ps_\text{red} = \ps_{12}/\!/G$, with $G$ the group
of boundary symmetries.  The fact that the boundary symmetries should act
trivially on the global phase space is now encapsulated by the restriction to the 
zero charge submanifold and further quotienting by the group action.  This process
thus gives a way of realizing the individual phase spaces $\ps_i$ within the 
global phase space, although they are not symplectic submanifolds due to the 
quotient procedure needed in the construction.

The work of Donnelly and Freidel focused on 
symmetries in general relativity
that preserve the codimension-2 boundary of a subregion Cauchy
surface \cite{Donnelly2016a}.
Such symmetries are simpler
to handle since the flux term in Hamilton's equations 
(\ref{eqn:IxiOm}) identically vanishes (assuming covariant $\beta'-\lambda'$ and 
field-independent generators), and 
the Wald-Zoupas procedure is not needed in order to construct localized charges.  The more 
general localized charges considered in the present work are defined even when there are 
nonzero fluxes, and in section \ref{sec:algebra} 
we showed that their Poisson brackets on the 
localized phase space are given by the BT bracket, which reproduces the 
diffeomorphism algebra of the vector fields (or a suitable modification when generators
are field-dependent) whenever the extension term $K_{\xi,\zeta}$ can be shown to vanish.  
This is enough to apply the Marsden-Weinstein reduction procedure, since the localized charges
generate an action of the boundary symmetry group on the localized phase space, even though
this action does not generically act like a diffeomorphism on all observables, due to the 
failure of such a transformation to satisfy Hamilton's equation.  
It would be very interesting to carry out this procedure in more detail in order to better
understand the relevance of localized charges within the full global phase space.

An even more interesting question is to understand how to apply the reduction in the case
of nonvanishing extension terms in the algebra of localized charges,
as in equation (\ref{eqn:HHbrack}).  The 
extensions $K_{\xi,\zeta}$ represent additional independent charges,
and together with the $H_\xi$ generators they produce an algebra that is larger than the 
original set of boundary symmetries.  There is a question of how to interpret these additional
charges, and how to interpret the reduction with respect to the additional generators.  
The mathematical machinery for handling such situations is called symplectic reduction
by stages \cite{Marsden2006}, and 
it would be worth investigating 
whether the reduced phase space obtained using this procedure reproduces the expected
global phase space.  

Another major motivation for carrying out this reduction procedure is in the applications
to the quantum theory of subregions in a gravitational theory.  
There is an analogous procedure to Marsden-Weinstein reduction 
whereby the physical Hilbert space $\mathcal{H}_\text{phys}$ is realized
as a subspace of the tensor product $\mathcal{H}^1\otimes \mathcal{H}^2$,
where $\mathcal{H}^i$ are the Hilbert spaces constructed via quantization
of the localized phase spaces $\ps^i$\cite{Donnelly:2014gva, Donnelly2016a}.  
This subspace is defined as the zero charge
eigenspace associated with the boundary symmetries in the localized phase spaces,
and restricting to this physical Hilbert space has the interpretation of imposing the 
constraints associated with diffeomorphism invariance.  There are a number of results 
beginning with the works of Guillemin and Sternberg that show in certain situations that 
the process of quantization commutes with symplectic reduction \cite{Guillemin1982}.  
Hence, we should expect that the localized 
phase spaces $\ps^i$ provide useful semiclassical descriptions of the local
Hilbert spaces $\mathcal{H}^i$.  

These local phase spaces are important when addressing questions regarding entanglement
entropy for subregions in gravitational theories and the entropy associated with 
black hole horizons.  It has long been appreciated that black holes possess an entropy
proportional to their area 
\cite{Bekenstein:1972tm, Bekenstein:1973ur, Hawking:1974rv}, and in a variety of contexts, this entropy
can be usefully interpreted as entanglement entropy 
\cite{Sorkin:2014kta, Bombelli1986, Srednicki1993a, Frolov1993}.  
Even more generic subregions 
in gravity are expected to possess a finite entropy
\cite{Susskind:1994sm, Jacobson:1994iw, Bousso:2015mna};
for example, in holography, subregions bounded by extremal codimension-2 surfaces 
have an entropy given by the Ryu-Takayanagi formula, which is interpreted
as the entanglement entropy of a subregion of the boundary conformal field theory
\cite{ Ryu:2006bv}.  The construction of localized Hilbert spaces as described above
is then crucial for giving a bulk Hilbert space interpretation of this entropy.  
The larger Hilbert space $\mathcal{H}^1\otimes\mathcal{H}^2$ in which the 
physical Hilbert space is embedded is known as the extended Hilbert space, and 
contains additional edge mode degrees of freedom that contribute to the entanglement entropy
\cite{Donnelly:2014gva}.   These edge modes can be viewed as objects charged under 
the boundary symmetries considered in this present work, and hence the localized charges 
play a central role in characterizing edge mode degrees of freedom.  
In some cases, considerations of boundary symmetries can in fact be shown to determine the 
entropy given some reasonable assumptions on the quantization of the localized phase space.
The best examples of this often involve a set of Virasoro symmetries or a related centrally
extended algebra acting on a Killing horizon 
\cite{Strominger1998, Carlip_1999, Haco:2018ske, Chen:2020nyh, Chandrasekaran:2020wwn}.  In this case, 
the quantization is conjectured to involve a CFT, and the Cardy formula for such a theory
then is able to reproduce the Bekenstein-Hawking entropy of the horizon.  
It is interesting that the central extension in these examples seems to play an 
important role in determining the entropy, and this may be related to 
interesting properties of the reduction procedure for algebras involving nonzero
extensions.

\subsection{Corner improvements} \label{subsec:corner-imprv}

In section \ref{sec:corners}, we described an additional correction that must be added to 
the localized charges to arrive at an expression that is fully invariant
under the extra ambiguities mentioned in that section.  This correction was first 
described in appendix C of \cite{Chandrasekaran:2020wwn}, and the present work generalizes
the proposal to allow for noncovariances in $L'$ and $\theta'$.  As mentioned
in the text, the correction to the charge density involve a quantity $c'$ which appears
as a contribution to the subregion action from codimension-2 corners.  
Note there are additional questions involving the precise relation between the 
full corner contribution to the action and the $c'$ appearing in the charge, 
since,
as discussed in footnote \ref{ftn:corneraction}, there are independent contributions
to the corner action coming from the boundary of each hypersurface $\ns^\pm$ ending 
at the corner.  Spelling out the precise relation between these contributions to the 
action and the localized charges would be an interesting future direction to explore.  

Ambiguities of the type described in section \ref{sec:corners} arose in the 
construction of GBMS charges in section \ref{expcalc}, where there could have been other possible choices for the form of the corner flux than the one we picked.  
It would be interesting to relate the choice made there to 
a more careful analysis of boundary terms needed to obtain a finite variational principle
for subregions bounded by $\scp$, and to carefully derive these terms from
a corner Dirichlet principle as well as a corner-improved holographic renormalization
procedure, as described in section 
\ref{sec:hr1}.  A
possible result of such an analysis would be to obtain GBMS charges that are finite in 
the limit to either end of $\scp$.  
This would  allow comparison to the expression
obtained by Compere, Fiorucci, and Ruzziconi in Eq.~5.49 of \cite{Compere:2018aar}, 
which does satisfy this finiteness property but was derived 
somewhat indirectly by using input from soft theorems.

Finally, we  mention that localized charges constructed via the Brown-York
procedure, as described recently in \cite{Chandrasekaran:2021hxc}, 
also enjoy the property of being 
free of the ambiguities discussed in section \ref{sec:corners}, since these 
charges only depend on the form of the codimension-1 flux $\beom$.  On the other
hand, these charges can differ from the canonical charges for transformations
that act anomalously on the boundary structures, and hence may yield different expressions
than the corner-improved charges.  It would be useful to carry out this comparison
in detail.

\subsection{Alternative resolutions of the ambiguity}
In this work, we have emphasized that resolving the ambiguities in the covariant phase 
space construction amounts to choosing a preferred form of the flux.  
Following \cite{Chandrasekaran:2020wwn}, we advocated for the use of 
a Dirichlet form of the flux, given is close  connection to 
standard holographic constructions, junction conditions,
and the Brown-York formulation
of localized charges recently explored in \cite{Chandrasekaran:2021hxc}.
Additional intrinsic counterterms preserving the Dirichlet form of the flux are necessary for 
asymptotic symmetries, where they are needed to ensure a finite flux through the boundary,
and were related to the holographic renormalization of the action in section
\ref{genarg}.
Previously, there have been other proposals for resolving the ambiguities, and we take a moment
to briefly comment on these alternative approaches.

The approach initially advocated by Wald and Zoupas
\cite{Wald:1999wa}, and employed in subsequent work, for example
\cite{CFP}, fixes some ambiguities using a stationarity 
condition, although for sufficiently permissive
boundary conditions, this requirement either does not yield a unique result, or else
fails to hold.  
A different approach is to focus on the properties of a given Lagrangian, and to extract 
a preferred symplectic potential using homotopy operators of the variational
bicomplex
\cite{Anderson1989, Barnich:2001jy, Freidel:2020xyx}.  While this certainly
yields an unambiguous result, there is still a degree of arbitrariness in the fact that 
homotopy operators for a given complex in general are not unique.  
In fact, the original formulas by Iyer and Wald \cite{Iyer:1994ys}
for the symplectic potential are completely unambiguous. The ambiguity instead 
arises in 
addressing why one particular formula for the symplectic potential is preferred over 
another.  In this regard, we find the resolving the ambiguity by focusing 
on properties of the flux yields a clearer explanation of what choices have been made 
in finding the resolution.  It would still be interesting to carefully relate the resolutions 
we explore in the present work to those involving the variational bicomplex, and understand the 
extent to which these two approaches can be made equivalent.  Finally, we mention the 
work of Kirklin 
\cite{Kirklin:2019xug}, 
who uses a construction based on the path integral for a subregion, and 
extracts a manifestly unambiguous symplectic potential using ideas closely related to the 
Peierls bracket construction \cite{Peierls:1952cb}.  
This procedure has a number of advantages beyond being manifestly
unambiguous, including making a more direct connection to the quantum description of the 
subregion, and being completely covariant with respect the codimension-2 corner 
of the subregion; i.e., it does not require a preferred codimenion-1 hypersurface $\ns$
bounding the subregion.  Unfortunately, the construction is sufficiently different from
the standard covariant phase space that it is not immediately clear what the specific form
of the corner contribution to the symplectic potential is in Kirklin's construction.  It would 
be very interesting to make this comparison, and determine whether his construction
is related to the Dirichlet form of the flux that was the focus of the present work.

\subsection{Casimir energy of vacuum AdS}

A byproduct of the localized charge construction in section \ref{sec:WZcharges} 
is that the resulting charges are largely free
from the usual ambiguity to be shifted by phase space 
constants.  The reason for this is that there are 
fewer quantities that qualify as true constants when
no boundary condition is imposed on the intrinsic boundary data.  
The requirement that the charges satisfy equation (\ref{eqn:Hammod})
is therefore a stronger condition than occurs in 
standard canonical frameworks
in which boundary conditions are imposed to ensure the flux
$\flx_\h\xi$ vanishes.  The additional content in 
equation (\ref{eqn:Hammod}) is that the charge $H_{\xi}$ 
must satisfy this equation  even for 
variations that violate the boundary conditions.  For example,
when taking $\beom$ to be in Dirichlet form, and choosing 
$\xi^a$ such that $\Delta_\h\xi(\beta'-\lambda') + h_{\delta \xi}$
vanishes, one would find that imposing a Dirichlet boundary
condition causes the entire flux $\flx_\h\xi$ to vanish, and $H_\xi$
is then the charge that integrates Hamilton's equation for the
transformation.  However, any other quantity $ H_{\xi}'$
that differs
from $H_{\xi}$ by a functional of the intrinsic quantities
on the boundary would also satisfy Hamilton's equation, since 
such intrinsic functionals are phase space constants once the 
Dirichlet boundary condition is imposed.  On the other hand,
these intrinsic functionals  have a nontrivial variation
for fluctuations that do not hold the intrinsic data fixed, 
in which case $ H_{\xi}'$ will fail to satisfy (\ref{eqn:Hammod})
in the larger phase space considered in this work where such variations 
are permitted. 
This allows us to conclude that the charge $H_\xi$ is unique
up to an overall constant that is independent of the bulk and 
boundary geometry.  The expression (\ref{eqn:Hxiint}) represents
a valid choice for fixing this constant, and allows for 
meaningful comparison of the values of the charges in different
spacetimes. 

An important context in which such a comparison arises is in 
odd-dimensional asymptotically AdS spaces,
where, depending on the choice of boundary conformal frame,
the charges in vacuum AdS can take on nonzero values.  In particular,
for asymptotic time translations, the nonzero charge is interpreted 
as the Casimir energy for the dual CFT 
\cite{Balasubramanian:1999re}.  
This result crucially relies on the ability to compare the charges 
in different conformal frames, and for the alternative definition
of canonical charges proposed by Ashtekar, Magnon, and Das (AMD)
\cite{Ashtekar:1984zz, Ashtekar:1999jx}, 
the energy vanishes for vacuum 
AdS, regardless of the choice of conformal frame.  
The resolution of this discrepancy lies in the fact that the 
AMD charges differ from the charges constructed from
a holographic stress tensor by an intrinsic functional of the 
boundary geometry \cite{Hollands:2005wt,
Papadimitriou:2005ii}.  This intrinsic functional has the effect
of subtracting off the value of the charge of vacuum AdS in 
the appropriate conformal frame, so that the AMD charges always 
vanish in vacuum AdS.  

This raises the question as to which definition of charge 
coincides with the expression (\ref{eqn:Hxiint}) in the context 
of asymptotically AdS spacetimes.  The answer can be inferred 
from the results of \cite{Chandrasekaran:2021hxc}
(see also \cite{Hollands:2005ya, Papadimitriou:2005ii, Harlow:2019yfa}), 
which showed that when the flux is chosen 
to be of Dirichlet form, $H_\xi$ agrees with the Brown-York charges 
constructed from the boundary stress tensor obtained by 
varying the subregion action with the respect to the intrinsic 
boundary variables \cite{Brown:1992br}.\footnote{More precisely,
the equivalence between the Brown-York and canonical definitions 
of charges was shown to hold for transformations that act 
covariantly on the intrinsic geometry of the boundary.  In the 
case of a null boundary, we showed in
\cite{Chandrasekaran:2021hxc} 
that for transformations that act anomalously on the 
null generator $n^i$, in the sense $\Delta_{\h\xi} n^i = w_\xi n^i$
for some function $w_\xi$, the two definitions of charges 
differ by an intrinsic functional constructed from $w_\xi$.
In the asymptotically AdS context, a similar anomaly should 
arise for asymptotic symmetries associated with conformal 
isometries of the boundary metric with nontrivial conformal 
factors.  In these cases, the holographic charges and canonical 
charges $H_\xi$ likely differ, and it would be interesting 
to investigate whether this difference has any physical
interpretation.  Note that this subtlety does not 
affect the discussion of the Casimir energy, since that involves 
charges assoicated with time translation, which is 
a boundary isometry with vanishing conformal factor.}  
Since the Casimir energy is obtained from holographic
charges constructed using the Brown-York method, it is immediately
apparent that the charges $H_\xi$ considered here will 
reproduce the Casimir energy of asymptotically AdS spacetimes, and 
therefore differ from the AMD charges.  It is important to emphasize
that, like the holographic charges, any shifts in the localized 
charges $H_\xi$ are derived from a corresponding change in the subregion
action, since the action principle completely determines the expression
for the charges.  This property is not shared by the AMD charges, and 
there does not appear to be any action principle that would yield
the AMD formula for the charges via the method of section 
\ref{sec:WZcharges}.  
Our construction 
thus provides a novel means of obtaining this Casimir energy 
from canonical methods that does not suffer from ambiguities 
associated with shifting the charges by intrinsic functionals.

\subsection{Implications for holography}
There are a number of potential applications of the present work to various aspects of 
holography.  The arguments of section
\ref{sec:hr1}
on holographic renormalization
of the symplectic potential are largely motivated by well-known constructions that 
originated in AdS/CFT 
\cite{Witten:1998qj, 
Henningson1998a, Balasubramanian:1999re, DeHaro2001,
Papadimitriou:2005ii}.  Although Dirichlet
boundary conditions were initially thought to be necessary in order to obtain a 
finite symplectic form, it was pointed out in the work of Comp\`ere and Marolf that in fact
the holographic renormalization procedure also yields a finite boundary symplectic form,
after taking into account the appropriate corner contributions 
\cite{Compere:2008us}.  This then motivates
definitions of a wide class of charges associated with all boundary diffeomorphisms, 
instead of focusing only on the subalgebra of conformal Killing vectors of the boundary 
metric.  For example, in the context of asymptotically de Sitter or anti-de Sitter spaces, 
such considerations led to the identifications of the $\Lambda$-BMS symmetry algebras, which 
are useful in obtaining the BMS symmetries upon taking a flat space limit 
\cite{Compere:2019bua,Compere:2020lrt,Fiorucci:2020xto}.
The general proof in section \ref{genarg} that such renormalization is always possible, independent
of the details of the spacetime asymptotics, suggests that the associated generalized charges are 
always present, and hence should have an interpretation in the dual holographic description.  

One puzzling aspect of interpreting these charges holographically is that the symmetry 
algebras constructed in this way are much larger than the algebras typically
encountered in standard examples of AdS/CFT.  For example, in asymptotically AdS spaces,
the dual quantum theory is a conformal field theory,
where the only conserved diffeomorphism charges are those associated with conformal
isometries.  On the other hand, the charges considered in the present work are  generically
not conserved, due to the presence of nonzero fluxes through the boundary, and hence there 
is no immediate contradiction with standard holographic considerations.  
The existence of these charges appears to be most closely tied to the ability to 
define a local stress tensor operator in the dual theory.  
As recently reviewed in \cite{Chandrasekaran:2021hxc}, the entire set of 
localized charges can be constructed using the Brown-York stress tensor
on the subregion boundary.  
Although each individual charge may not be conserved,
the stress tensor itself satisfies a covariant conservation equation as a consequence 
of the gravitational constraints.
In a holographic dual picture,
the dictionary relates the Brown-York stress tensor to the local 
stress tensor of the dual field theory.  
Because the continuity equation relating the nonconservation of the charges to the flux
is intimately related to the covariant conservation equation of the stress tensor,
one could speculate that the diffeomorphism charges become important when characterizing 
the theory in a hydrodynamical regime,  which gives a coarse-grained, effective description 
of the quantum theory in which the important degrees of freedom are those 
associated with conserved quantities, such as the stress tensor.
This connection between gravity and hydrodynamics has been noted in holography in the fluid-gravity
correspondence
\cite{Bhattacharyya:2007vjd, Rangamani:2009xk}, 
and has also appeared in various other contexts including the membrane paradigm
of black holes \cite{Damour1982, Thorne:1986iy} and considerations of the Einstein equation of state
\cite{Jacobson:1995ab}.

There are a number of other possible holographic applications of the present work.  
The considerations of localized charges are well-adapted to describing gravitational theories
in local subregions, and in some cases these subregions can be given a holographic interpretation in 
terms of a CFT deformed by an irrelevant $T\bar T$ or $T^2$ deformation
\cite{McGough:2016lol,Hartman:2018tkw}.  Some  ideas relating the $T\bar T$
deformation to covariant phase space constructions were recently considered in \cite{Kraus:2021cwf}.  

Another area of interest to which the localized charges may be relevant is in the 
recent models of black hole evaporation that reproduce the Page curve
\cite{Almheiri:2019psf, Penington:2019npb, Almheiri:2020cfm}, where outgoing 
Hawking radiation in an asymptotically flat AdS black hole is collected in a non-gravitational theory
on flat space, in order induce evaporation.  This gluing construction is similar in spirit to 
the reduction procedure described in section \ref{sec:gluing} for combining subregions, and 
hence it may be worthwhile to understand the evaporation models from that perspective.
Furthermore, the gluing construction should in principle be possible in setups where both subregions
are gravitational, and hence may yield a useful way of understanding black hole evaporation
models without restricting one of the subregions to be nongravitational.  This may help address
recent criticisms of applicability of the evaporation models to genuine asymptotically 
flat gravitational systems raised in \cite{Raju-1,Raju-2}.  

Finally, the considerations 
of null surfaces and holographic renormalization is particularly well-adapted to applications
in celestial holography, which seeks to find a dual of asymptotically flat space in terms of 
a celestial CFT 
\cite{Aneesh:2021uzk, Raclariu:2021zjz, Pasterski:2021rjz}.   
In particular, it would be worthwhile to understand the covariant counterterms
needed to renormalize the action and the associated 
null Brown-York stress tensor recently considered in \cite{Chandrasekaran:2021hxc},
without explicitly employing the auxiliary
rigging vector used in section 
\ref{exparg}.

\section*{Acknowledgments}

We thank Luca Ciambelli, Geoffrey Comp\`ere, Laurent Freidel, Rob Leigh, 
Don Marolf, David Nichols, Ioannis Papadimitriou, 
Kartik Prabhu, Daniele Pranzetti, Romain Ruzziconi, and 
Kostas Skenderis
for helpful discussions.  E.F. and I.S. are supported in part by NSF
grants PHY-1707800 and PHY-2110463. I.S. also acknowledges support from the John and David Boochever prize fellowship in fundamental theoretical physics.
AJS is supported by the Air Force Office of Scientific Research under award number FA9550-19-1-036.
Research at Perimeter Institute is supported in part by the Government of Canada through the Department of Innovation, Science and Economic Development and by the Province of Ontario through the Ministry of Colleges and Universities. V.C. is supported in part by the Berkeley Center for Theoretical Physics; by the Department of Energy, Office of Science, Office
of High Energy Physics under QuantISED Award DE-SC0019380 and under contract DEAC02-05CH11231; by the National Science Foundation under grant PHY-1820912; and by a grant from the Simons Foundation (816048, VC).

\appendix

\section{Field space calculations}
\label{app:fs}
Here we collect some identities satisfied by various operators on field
space.
Given a vector field $V$ on $\fs$, its action on differential forms via the Lie
derivative is given by Cartan's magic formula
\beq
L_V = I_V\delta + \delta I_V.
\eeq
More generally, if $\nu$ is a vector-valued one-form on $\fs$, we can define 
a derivation of the exterior algebra of degree
$0$ denoted $I_\nu$ which is given by contraction
on the vector index and then antisymmetrization of the remaining covariant indices;
on a $p$-form $\alpha$, this is given by
\cite{Michor1993}
\beq
(I_\nu \alpha)_{\n A_1 \n A_2\ldots \n A_p} =  \nu\ind{^{\n B}_{\un{\n A_1}}}
\alpha\ind{_{\n B}_{\un{\n A_2 \ldots \n A_p}}},
\eeq
where the underline denotes antisymmetrization of the indices.
The graded commutator of $I_\nu$ with the exterior derivative $\delta$ 
defines a new derivation of degree $1$ denoted $L_\nu$,
\beq
L_\nu = I_\nu \delta - \delta I_\nu.
\eeq
In particular, a field dependent vector field $\xi^a$ has nontrivial variation $\delta
\xi^a$ which is a one form on field space.  The map $\xi^a\mapsto \hat \xi$
extends to $\delta \xi^a$, producing a vector-valued one form on $\fs$ denoted 
$\wh{\delta\xi}$.  This object then defines derivations $I_{\wh{\delta\xi}}$ and 
$L_{\wh{\delta\xi}}$ by the above definitions.
A vector valued differential form $\rho$ of higher degree defines derivations 
$I_\rho$ and $L_\rho$ in a similar manner. 

\begin{lemma} \label{lem:1} The various derivations defined above satisfy 
\begin{align}
[L_{\h\xi}, \lie_\zeta] &= \lie_{\left(I_{\h\xi} \delta \zeta \right)}\label{eqn:Llie}\\
[\lie_\xi, I_{\wh{\delta \zeta}} ] & = 0 \label{eqn:lieIdel}\\
[I_{\h\xi}, I_{\wh{\delta\zeta}}] &= I_{\wh{I_{\h\xi}\delta\zeta}} \label{eqn:IIdel} \\
[I_{\wh{\delta\xi}}, I_{\wh{\delta\zeta}}] &= I_{\hat\sigma}; 
\qquad \sigma^a = I_{\wh{\delta\xi}}\delta \zeta^a - I_{\wh{\delta\zeta}}\delta\xi^a
\label{eqn:IdelIdel}\\
[L_{\h\xi}, I_{\wh{\delta\zeta}}] &=
I_{\h\tau}; \qquad \tau^a = [\delta\zeta,\xi]^a +\delta I_{\h\xi}\delta \zeta^a-
I_{\wh{\delta\zeta}}\delta\xi^a 
\label{eqn:LIdel}\\
[L_{\h\xi}, L_{\h\zeta}] &= -L_{\wh{\brmod{\xi}{\zeta}}};
\qquad \brmod{\xi}{\zeta}^a = [\xi,\zeta]^a -I_{\h\xi}\delta \zeta^a + I_{\h\zeta}\delta\xi^a\label{eqn:LL}.
\end{align}
In particular, (\ref{eqn:LL}) implies that the field space Lie bracket is given by 
\beq\label{eqn:fslb}
[\h\xi,\h\zeta]_\fs = -\wh{\brmod{\xi}{\zeta}}.
\eeq
\end{lemma}
\begin{proof}
For (\ref{eqn:Llie}), we compute
\begin{align}
[L_\h\xi,\lie_\zeta] &=I_{\h\xi}\delta \lie_\zeta + \delta I_{\h\xi}\lie_\zeta
-\lie_\zeta I_{\h\xi}\delta - \lie_\zeta\delta I_{\h\xi}\nonumber \\
&=I_{\h\xi}\lie_{\delta\zeta} +\lie_\zeta I_\h\xi \delta + \lie_{\delta\zeta}I_\h\xi
-\lie_\zeta I_\h\xi\delta -\lie_\zeta \delta I_{\h\xi}
\nonumber \\
&=\lie_{I_\h\xi \delta \zeta} - \lie_{\delta \zeta}I_{\h\xi} + \lie_{\delta \zeta}I_\h\xi
\end{align}
yielding the identity.  Equation (\ref{eqn:lieIdel}) is identically true
from the definition of how $\lie_\xi$ and $I_{\wh{\delta\zeta}}$ act on 
field space differential forms.  
Equations (\ref{eqn:IIdel}) and (\ref{eqn:IdelIdel}) follow from the for the 
Nijenhuis-Richardson bracket for two algebraic derivations
\cite{Michor1993}.  

For equation (\ref{eqn:LIdel}), we know from the 
general structure of brackets of derivations that $[L_{\h\xi}, I_{\wh\delta \zeta}]$
must be an algebraic derivation, and hence is determined by its action on a basis
of one-forms $\delta \phi$.  This then produces
\begin{align}
[L_\h\xi, I_{\wh{\delta\zeta}} ] \delta \phi &=
I_\h\xi\delta \lie_{\delta\zeta}\phi + \delta \lie_{I_{\h\xi}\delta \zeta}\phi
-I_{\wh{\delta\zeta}}\delta\lie_\xi \phi \nonumber \\
&=
-\lie_{I_\h\xi\delta \zeta}\delta \phi+\lie_{\delta\zeta}\lie_\xi\phi
+\lie_{\delta(I_\h\xi\delta\zeta)}\phi +\lie_{I_\h\xi\delta \zeta}\delta \phi
-\lie_{I_\wh{\delta\zeta}\delta\xi}\phi -\lie_\xi \lie_{\delta\zeta}\phi
\nonumber \\
&=
\lie_{\big([\delta\zeta, \xi]+\delta I_{\h\xi}\delta \zeta 
- I_\wh{\delta\zeta}\delta \xi\big)} \phi
\end{align}
which then reproduces the RHS of (\ref{eqn:LIdel}).

Finally, for the commutator $[L_\h\xi, L_\h\zeta]$, we know that the resulting 
derivation will be a Lie derivative, and hence it is determined by 
its action on the scalars $\phi$.  We can therefore compute
\begin{align}
[L_\h\xi, L_\h\zeta]\phi &=
L_\h\xi \lie_\zeta\phi - L_\h\zeta \lie_\xi\phi \nonumber \\
&=
\lie_{I_\h\xi\delta\zeta} \phi +\lie_\zeta \lie_\xi \phi 
-\lie_{I_\h\zeta\delta\xi} \phi -\lie_\xi\lie_\zeta \phi \nonumber \\
&= -\lie_{\brmod{\xi}{\zeta}} \phi 
\end{align}
\end{proof}

\begin{lemma}  The operator $\Delta_{\h\xi}$ satisfies the following identities
\begin{align}
[\delta, \Delta_{\h\xi}] &= \Delta_{\wh{\delta\xi}} \label{eqn:delDel} \\
[\Delta_{\h\xi}, \Delta_{\h\zeta}] &= \Delta_{[\h\xi,\h\zeta]_\fs} 
= -\Delta_{\wh{\brmod{\xi}{\zeta}}} \label{eqn:WZ}\\
[\Delta_\h\xi, I_\h\zeta] &= I_{\wh {\Delta_\h\xi \zeta}}
= -I_{\wh{\brmod{\xi}{\zeta}} }  +I_{\wh{ I_{\h\zeta}\delta\xi}} 
\label{eqn:DelIz}
\end{align}
\end{lemma}
\begin{proof}
Equation (\ref{eqn:delDel}) follows from
\begin{align}
\delta\Delta_{\h\xi} &= \delta(L_{\h\xi}-\lie_\xi - I_{\delta\xi}) \nonumber \\
&= L_{\h\xi}\delta -\lie_{\delta\xi} -\lie_\xi\delta+L_{\wh{\delta\xi}} 
-I_{\wh{\delta\xi}}\delta \nonumber\\
&= \Delta_{\h\xi}\delta +\Delta_{\wh{\delta\xi}}
\end{align}
since $I_{\wh{\delta\delta\xi}} = 0$.

To derive equation (\ref{eqn:WZ}), we can use the identities in Lemma \ref{lem:1}
to derive
\begin{align}
[\Delta_\h\xi, \Delta_\h\zeta] &= [(L_{\h\xi}-\lie_\xi -I_\wh{\delta\xi}),
(L_{\h\zeta}-\lie_\zeta-I_\wh{\delta\zeta} ) ] \nonumber \\
&=
-L_{\wh{\brmod{\xi}{\zeta}}} +\lie_{[\xi,\zeta]} 
-\lie_{I_\h\xi\delta\zeta} + \lie_{I_\h\zeta\delta\xi} -[L_\h\xi, I_\wh{\delta\zeta}]
-[I_{\wh{\delta\xi}}, L_{\h\zeta} ]+ [I_\wh{\delta\xi}, I_{\wh{\delta\zeta}}]
\label{eqn:DelDelcalc}
\end{align}
The last three commutators all combine into a single contraction $I_{\hat\alpha}$,
and using (\ref{eqn:IdelIdel}) and (\ref{eqn:LIdel}) we find
\begin{align}
\alpha^a &= -[\delta\zeta,\xi]^a-\delta I_\h\xi\delta\zeta^a+I_\wh{\delta\zeta}\delta\xi^a
+[\delta\xi,\zeta]^a+\delta I_\h\zeta\delta\xi^a-I_\wh{\delta\xi}\delta\zeta^a
+I_\wh{\delta\xi}\delta\zeta^a + I_\wh{\delta\zeta}\delta\xi^a\nonumber \\
&=
\delta\brmod{\xi}{\zeta}^a.
\end{align}
Hence, equation (\ref{eqn:DelDelcalc}) becomes
\beq
[\Delta_\h\xi, \Delta_\h\zeta]
=
-L_{\wh{\brmod{\xi}{\zeta}}} + \lie_{\brmod{\xi}{\zeta}}
+I_{\wh{\delta\brmod{\xi}{\zeta}}}
=-\Delta_{\wh{\brmod{\xi}{\zeta}}}
\eeq

Finally, for equation (\ref{eqn:DelIz}), we apply equations 
(\ref{eqn:IIdel}) and (\ref{eqn:fslb}) to compute 
\begin{align}
[\Delta_\h\xi, I_\h\zeta] &= [L_\h\xi -\lie_\xi - I_\wh{\delta\xi}, I_\h\zeta] \\
&= I_{[\h\xi,\h\zeta]_\fs} + I_{\wh{I_\h\zeta\delta\xi}} \\
&= -I_{\wh{\brmod{\xi}{\zeta}} } + I_{\wh{I_\h\zeta\delta\xi}} 
= I_{\wh{\Delta_\h\xi \zeta}}
\end{align}

\end{proof}

\section{Phase space calculations}\label{sec:phase-space-calc}

The standard Iyer-Wald 
identity \cite{Wald:1993nt, Iyer:1994ys} for computing the contraction of a vector field into
the symplectic current receives modifications when $\theta'$ contains
noncovariances.  Making generous use of Cartan's magic formula in addition to Eqs.~(\ref{eqn:noncov-defn}), (\ref{eqn:noncovtheta}), (\ref{eqn:noethercurrent}), (\ref{eqn:Jxi'}), (\ref{eqn:delDel}), as well as the fact that on-shell, $\delta L^{\prime} = d \theta^{\prime}$, 
we find that 
\begin{align}
-I_{\h\xi}\omega' &= -L_{\h\xi} \theta' + \delta I_{\h\xi}\theta' \nonumber \\
&= -\lie_\xi \theta' -\Delta_{\h\xi}\theta'-I_{\wh{\delta\xi}}\theta' + \delta(J_\xi' +i_\xi L' +\Delta_{\h\xi}b')
\nonumber \\
&=-i_\xi d\theta' -di_\xi\theta' 
-\Delta_{\h\xi} \delta b' -d\Delta_{\h\xi}\lambda' 
-J_{\delta\xi}' -\Delta_{\wh{\delta\xi}} b'
+ d\delta Q_\xi'  + i_\xi \delta L' +\delta\Delta_{\h\xi}b' \nonumber \\
&= 
d(\delta Q_\xi' -Q_{\delta \xi}'-i_\xi\theta'-\Delta_{\h\xi}\lambda')
\label{eqn:Ixiom}
\end{align}
where we used $
\delta i_\xi L' = i_{\delta \xi} L' + i_\xi \delta L'
$ in the third line. This is then used in determining the charges and fluxes that appear upon
contracting $-I_{\h\xi}$ into the symplectic form.  Taking into account the 
additional boundary contribution to $\Omega$ [Eq.~(\ref{eqn:Omcorner})], the result localizes to a boundary
integral, whose integrand, using Eqs.~(\ref{eqn:noncov-defn}) and (\ref{eqn:thdecomp}), is given by 
\begin{align}
& \;\delta Q_\xi' - Q_{\delta \xi}' - i_\xi \p\theta'
-\Delta_{\h\xi}\lambda'+I_{\h\xi}\delta\beta' \nonumber \\
=&\;
\delta Q_{\xi}'-Q_{\delta \xi}' +i_\xi\delta \ell' 
-\lie_\xi \beta' + di_\xi\beta' -i_\xi\beom -\Delta_{\h\xi}\lambda'
+L_{\h\xi}\beta'-\delta I_{\h\xi}\beta' \nonumber \\
=& \;
\delta(Q_\xi' +i_\xi\ell' - I_{\h\xi}\beta') - Q_{\delta\xi}'-i_{\delta\xi}\ell'
+\Delta_{\h\xi}(\beta'-\lambda') +I_{\wh{\delta\xi}}\beta'-i_\xi\beom +di_\xi\beta'
\nonumber\\
=& \;
\delta h_\xi - h_{\delta\xi} -i_\xi\beom + \Delta_{\h\xi}(\beta'-\lambda') +di_\xi\beta'
\end{align}
where we recall the definition of the charge density
\beq
h_\xi = Q_\xi' +i_\xi\ell' -I_{\h\xi}\beta'.
\eeq
Integrating this expression over the boundary of a Cauchy surface then
yields Eq.~(\ref{eqn:IxiOm}).

The exterior derivative of $h_\xi$ can then be explicitly computed, using Eqs.~\eqref{eqn:noncov-defn} and \eqref{eqn:thdecomp},
\begin{align}
dh_\xi &= J_\xi' + \lie_\xi \ell'-i_\xi d\ell' -I_{\h\xi}d\beta' \nonumber \\
&=
I_{\h\xi}\p\theta' -i_\xi L'-\Delta_{\h\xi}b'+I_{\h\xi}\delta\ell' -\Delta_{\h\xi}\ell'
-i_\xi d\ell' -I_{\h\xi}d\beta'
\nonumber \\
&=
I_{\h\xi}\beom -\Delta_{\h\xi}(\ell' +b') - i_{\xi}(L' +d\ell') \label{eqn:dhxiapp}
\end{align}
which verifies Eq.~(\ref{eqn:dhxi}).

When computing the bracket between the localized charges, it is helpful to have
an expression for the anomaly of the charge density.  First, we note using the 
expression (\ref{eq:Qprimexi}) for $Q_\zeta'$, the transformation property 
(\ref{eqn:Delvc}) satisfied by the covariant part $\mcov{Q}_\zeta$, and 
the identity (\ref{eqn:DelIz}), that the anomaly of $Q_\zeta'$ is given by
\beq
\Delta_\h\xi Q_\zeta' = -Q'_{\brmod{\xi}{\zeta}} +Q'_{I_\h\zeta\delta\xi} +i_\zeta
\Delta_\h\xi\ell' -I_\h\zeta\Delta_\h\xi \beta',
\eeq
and similarly it follows that the anomaly of the charge density is 
\beq \label{eqn:Delxihzeta}
\Delta_\h\xi h_\zeta = -h_{\brmod{\xi}{\zeta}} +h_{I_\h\zeta\delta\xi}
+i_\zeta\Delta_\h\xi(\ell' + b') -I_\h\zeta \Delta_\h\xi(\beta'-\lambda').
\eeq
The bracket (\ref{eqn:BTbrack}) of the charges is then given by
\beq
\{H_\xi, H_\zeta\} = -I_\h\xi \delta H_\zeta + I_\h\zeta \flx_\h\xi =
\int_{\partial\Sigma}m_{\xi,\zeta},
\eeq
and by applying the definition (\ref{eqn:Axi}) of $\flx_\h\xi$ and using 
(\ref{eqn:dhxiapp}) and (\ref{eqn:Delxihzeta}), the integrand can evaluates to
\begin{align}
m_{\xi,\zeta} &= -\lie_\xi h_\zeta - \Delta_\h\xi h_\zeta +I_\h\zeta\left(i_\xi\beom
-\Delta_\h\xi(\beta'-\lambda') +h_{\delta\xi}\right) \\
&=h_{\brmod{\xi}{\zeta}}- i_\zeta\Delta_\h\xi(\ell' +b') + i_\xi \Delta_\h\zeta(\ell'+b')
+i_\xi i_\zeta(L'+d\ell')-di_\xi h_\zeta. \label{eqn:mxizeta}
\end{align}
Integrating this over the surface $\partial\Sigma$ yields 
the charge representation theorem quoted in equation (\ref{eqn:HHbrack}), 
using that $\xi^a$ and 
$\zeta^a$ are both tangent to $\ns$ which causes the term $i_\xi i_\zeta(L'+d\ell')$
to pull back to zero.

A similar computation yields the bracket for the corner-improved charges 
constructed in section \ref{sec:corners}.  Working with an improved 
charge density $\tilde h_\xi$ defined by dropping the final exact term in equation
(\ref{eqn:htilde2}) which integrates to zero in the charge, 
\beq
\tilde h_\xi = \mcov{Q}_\xi +i_\xi(\ell'+b'+dc') - I_\h\xi \cflx,
\eeq
we find that its exterior derivative is given by
\beq \label{eqn:dhtilde}
d\tilde h_\zeta = I_\h\zeta \beom -\Delta_\h\zeta(\ell'+b'+dc') - i_\zeta(L'+d\ell'),
\eeq
and its anomaly by
\beq\label{eqn:Delhtilde}
\Delta_\h\xi \tilde h_\zeta = \tilde h_{\Delta_\h\xi\zeta} 
+ i_\zeta\Delta_\h\xi(\ell'+b'+dc')-I_\h\zeta \Delta_\h\xi \cflx.
\eeq
The bracket of the charges is 
\beq
\{\tilde H_\xi, \tilde H_\zeta\} = -I_\h\xi \delta \tilde H_\zeta+I_\h\zeta
\tilde \flx_\h\xi = \int_{\partial\Sigma}\tilde m_{\xi,\zeta}
\eeq
with $\tilde \flx_\h\xi$ defined in (\ref{eqn:tildeflx}).  Then applying 
(\ref{eqn:dhtilde}) and (\ref{eqn:Delhtilde}), the integrand evaluates to
\begin{align}
\tilde m_{\xi,\zeta} &=
-\lie_\xi \tilde h_\zeta -\Delta_\h\xi \tilde h_\zeta +I_\h\zeta\left(i_\xi \beom
-\Delta_\h\xi \cflx+\tilde h_{\delta\xi}\right) \\
&=
\tilde h_{\brmod{\xi}{\zeta}} + i_\xi \Delta_\h\zeta(\ell'+b'+dc') 
- i_\zeta\Delta_\h\xi(\ell'+b'+dc') + i_\xi i_\zeta (L' +d\ell') -di_\xi \tilde h_\zeta
\end{align}
which by the same arguments as above yields the corner-improved charge representation
theorem, equations (\ref{eqn:tilHtilH}) and (\ref{eqn:tilK}).

\section{Scaling transformations on a null surface}
\label{app:nullsurf}

Consider a spacetime $(\stm, g_{ab})$ containing a null surface ${\cal N}$.
In this appendix we review the various geometric quantities that are naturally
defined on ${\cal N}$ (see for example Sec.\ 3 of
Ref.\ \cite{CFP} for more details), and how they transform under rescalings of the
null normal and under conformal transformations of the metric.  We restrict to 
$4$-dimensional spacetimes in this section.

We pick a smooth future-directed normal covector $n_a$ on ${\cal N}$, and define the
inaffinity $\nonaffinity$, a function on ${\cal N}$, by\footnote{
If the extension of $n_a$ away from $\ns$ is chosen to satisfy
$\nabla_{[a} n_{b]} = 0$, the quantity $\kappa$ is equivalently given by 
the relation $\nabla_a (n_b n^b) \heq 2 \kappa n_a$ which is the usual definition of surface
gravity for a horizon when $n_a$ is a Killing vector field.  Thus the
inaffinity is sometimes called surface gravity for general normals $n_a$, 
although in the most general case where $\nabla_{[a} n_{b]} \neq 0$, these 
two definitions of $\kappa$ will not agree.}
\be
n^a \nabla_a n^b \hateq \nonaffinity n^b,
\label{kappadef}
\ee
where we are using $\hateq$ to mean equality when evaluated on
${\cal N}$.
The contravariant normal $n^a = g^{ab} n_b$, when evaluated on $\scp$, can be
viewed as an intrinsic vector $n^i$, since $n^a n_a =0$.  We denote by
$\inducedmetric_{ij}$ the degenerate induced metric, and by
$\volume_{ijk}$ the 3-volume form on ${\cal N}$ given by taking the pullback
of $\volume_{abc}$ where $\volume_{abc}$ is any three form with $4 \volume_{[abc} n_{d]} = \varepsilon_{abcd}$.
Finally we define a 2-volume form by
\be
\label{volumesmall}
\volumesmall_{ij} =  -\volume_{ijk} n^k.
\ee

Next, we take the pullback on the index $a$ of $\nabla_a n^b$, which
is then orthogonal to $n_b$ on the index $b$.  This quantity therefore
defines an intrinsic tensor ${\shape}_i^{\ j}$ called the Weingarten
map \cite{Gourgoulhon:2005ng}.  The second fundamental form or shape
tensor is $K_{ij} = {\shape}_i^{\ k} \inducedmetric_{kj}$, which can be decomposed
as
\be
K_{ij} = \frac{1}{2} \expansion \inducedmetric_{ij} + \sigma_{ij}
\ee
in terms of an expansion\footnote{The relation of the expansion $\expansion$ to the divergence $\nabla_a n^a$ of the normal depends on how one extends the definition of $n^a$ off the null surface.  If that extension satisfies $n_a n^a=0$, then $\Theta = \nabla_a n^a - \nonaffinity$.  If that extension satisfies $\nabla_{[a} n_{b]}=0$, then we have instead $\Theta = \nabla_a n^a - 2 \nonaffinity$ \cite{Jacobson:1993pf}.} $\expansion$ and a symmetric traceless shear tensor $\sigma_{ij}$.

These fields on a null surface obey the relations \cite{Gourgoulhon:2005ng,CFP}
\begin{subequations}
  \label{grelations1}
  \begin{eqnarray}
\label{grelations1a}
\inducedmetric_{ij} n^j &=& 0, \\
\label{grelations1a1}
K_{ij} n^j &=& 0, \\
\label{grelations1a2}
{\shape}_i^{\ j} n^i &=& \nonaffinity n^j,\\
\label{grelations1b}
( \lie_n - \expansion) \inducedmetric_{ij} &=& 2 \sigma_{ij}, \\
\label{grelations1c}
( \lie_n - \expansion) \volume_{ijk} &=& 0, \\
\label{grelations1d}
( \lie_n - \expansion) \volumesmall_{ij} &=& 0, \\
( \lie_n - \nonaffinity) \expansion &=& - \frac{1}{2} \expansion^2 -  \sigma_{ij} \sigma_{kl} q^{ik} q^{kl} - R_{ab} n^a n^b,
  \end{eqnarray}
\end{subequations}
where $q^{ij}$ is any tensor that satisfies $q_{ij} q^{jk} q_{kl} = q_{il}$.

Consider now rescaling the normal according to
\be
n^i \to e^\sigma n^i,
\ee
where $\sigma$ is a smooth function on ${\cal N}$.  We can also
perform a conformal transformation on the metric,
\be
g_{ab} \to e^{2 \Upsilon} g_{ab}.
\ee
Here $\Upsilon$ is a smooth function on a neighborhood of ${\cal N}$,
but we will be interested only in $\Upsilon$ restricted to ${\cal N}$.
Under the combined effect of these transformations the various fields
transform as
\begin{subequations}
  \label{contransformc}
  \begin{eqnarray}
  \label{contransformc2}  
  n_a &\to& e^{\sigma + 2 \Upsilon} n_a, \\
    \label{contransformc3}
    \inducedmetric_{ij} &\to& e^{2 \Upsilon} \inducedmetric_{ij}, \\
    \label{contransformc4a}
      \volumesmall_{ij} &\to& e^{2 \Upsilon} \volumesmall_{ij}, \\
    \label{contransformc4}
      \volume_{ijk} &\to& e^{2 \Upsilon - \sigma} \volume_{ijk}, \\
        \label{contransformc5}
    \nonaffinity &\to& e^{\sigma} (\nonaffinity + \lie_n \sigma + 2 \lie_n
    \Upsilon),\\
            \label{contransformc6}
    \expansion &\to& e^{\sigma} (\expansion + 2 \lie_n
    \Upsilon),\\
            \label{contransformc7}
    K_{ij} &\to& e^{\sigma + 2 \Upsilon} (K_{ij} + \inducedmetric_{ij} \lie_n \Upsilon),\\
    {\shape}_i^{\ j} &\to& e^{\sigma} \left[{\shape}_i^{\ j} +
      D_i(\sigma + \Upsilon) n^j + \lie_n \Upsilon \delta_i^j \right],
    \end{eqnarray}
\end{subequations}
where $D_i$ is any derivative operator on ${\cal N}$.
These transformation laws preserve the relations (\ref{grelations1}).

In applying this framework to null surfaces ${\cal N}$ at a finite location in
spacetime \cite{CFP}, the metric $g_{ab}$ is the physical metric.
Hence there is no freedom to conformally rescale the metric, and we
must take $\Upsilon =0$.  In this case the scaling laws
(\ref{contransformc}) reduce to the scaling laws\footnote{Note that the quantities denoted here by $W_i^{\ j}, \expansion, \inducedmetric_{ij}, \volumesmall_{ij}, \volume_{ijk}$ and $n^i$ were denoted there ${\cal K}_i^{\ j}, \theta, h_{ij}, \varepsilon_{ij}, \varepsilon_{ijk}$ and $\ell^i$, respectively.} given in Eq.\ (3.3)
of Ref.\ \cite{CFP}.  By contrast, in applying the framework to future null
infinity ${\cal N} = \scp$, the metric $g_{ab}$ is the unphysical
metric and is subject to the conformal rescaling freedom
(\ref{conffreedom}), which also includes a rescaling of the normal.
In this case we must take $\Upsilon = - \sigma$, and with this specialization the scaling laws
(\ref{contransformc}) reduce to the laws (\ref{contransform}) given Appendix \ref{sec:conf}.

\section{Asymptotically flat spacetimes: notations and conventions}
\label{sec:conf}

In this appendix we review the definition of asymptotically flat
spacetimes in 3+1 dimensions, and define the notations we use for the conformal
completion framework used to describe them.

Consider vacuum spacetimes that are asymptotically flat at
null infinity, $\scri$, in the sense of \cite{Wald:106274}.
This means that we have a manifold $\stm$ with boundary $\scri$ which is
topologically $\bb R \times  S^2$, and an
unphysical metric $\gunphys_{ab}$ which is smooth on $\stm$ for which
$\scri$ is null.  We also have a smooth conformal
factor $\conf$ on $\stm$ which satisfies $\conf = 0$
on $\scri$ and
for which
\be
n_a = \nabla_a \conf
\label{normaldef}
\ee
vanishes nowhere on $\scri$.  Finally the
physical metric 
\be
\gphys_{ab} = \conf^{-2} \gunphys_{ab}
\label{eqn:physmetric}
\ee
satisfies the
vacuum Einstein equation ${\tilde G}_{ab} =0$ on $\stm \setminus \scri$.
The conformal transformation
\be
\left( \gunphys_{ab}, \conf \right) \to \left( e^{-2 \sigma}
\gunphys_{ab}, e^{-\sigma} \conf \right),
\label{conffreedom}
\ee
where $\sigma$ is a smooth function on $M$, preserves the physical
metric.
Although normally one would expect the theory to be invariant under this conformal freedom,
it is possible in general contexts
for the definitions of gravitational charges to depend on background structures
like the choice of conformal frame (as it does in AdS), as we argued in Sec.\ \ref{sec:charges} above.  The following discussion can be easily adapted to past null infinity but we will focus on future null infinity just for simplicity.

As for any null surface, the metric $\gunphys_{ab}$ and normal $n_a$ determine a
number of geometric quantities on $\scp$, reviewed in Appendix
\ref{app:nullsurf}.  These include the inaffinity $\nonaffinity$, 
the expansion $\expansion$, the shear tensor $\sigma_{ij}$, the induced metric
$\inducedmetric_{ij}$, the 3-volume form $\volume_{ijk}$, the 2-volume form $\volumesmall_{ij}$,
the second fundamental form or shape tensor $K_{ij}$, and the Weingarten map
${\shape}_i^{\ j}$.  For general null surfaces these quantities obey a
number of identities given in Eqs.\ (\ref{grelations1}).  We now review
properties of these quantities that are specific to $\scp$.

First, the normal $n_a$ is a pure gradient from Eq.\ (\ref{normaldef}), and so
$\nabla_{[a} n_{b]}=0$.  Since $\scp$ is null we have $n_a n^a =
\conf g + O(\conf^2)$ for some function $g$ on $\scp$.
Taking a gradient, evaluating at $\conf=0$, using the symmetry of
$\nabla_a n_b$ and using the definition (\ref{kappadef}) of the inaffinity $\nonaffinity$ now yields that 
\be
\label{eees1}
g^{ab} n_a n_b = 2 \nonaffinity \conf + O(\conf^2).
\ee
Second, it follows from the vacuum Einstein equation satisfied by the
physical metric that 
\be
\nabla_{(a} n_{b)}  \hateq f  g_{ab}
\ee
for some function $f$ on $\scp$;
see, e.g., Eq.\ (2.6) of Ref.\ \cite{Flanagan:2019vbl}.
As a reminder we are using $\hateq$ to mean equality when evaluated on
$\scp$. Combining this with
Eq.\ (\ref{normaldef}) yields $\nabla_a n_b \hateq f g_{ab}$, from
which we obtain $f = \nonaffinity$ and 
\bes
\label{scriidentities}
\bea
\label{scriidentities1}
\nabla_a \nabla_b \conf &\hateq& \nonaffinity g_{ab},\\
\label{scriidentities2}
\expansion &=& 2 \nonaffinity, \\
\label{scriidentities3}
\sigma_{ij} &=& 0, \\
\label{scriidentities4}
      {\shape}_i^{\ j} &=& \nonaffinity \delta_i^j.
      \eea
\ees
Inserting Eqs.\ (\ref{scriidentities2}) and (\ref{scriidentities3}) into the general identities
(\ref{grelations1}) for  any null surface yields
the relations
\begin{subequations}
  \label{relations1}
  \begin{eqnarray}
\label{relations1a}
\inducedmetric_{ij} n^j = 0, \\
\label{relations1b}
( \lie_n - 2 \nonaffinity) \inducedmetric_{ij} = 0, \\
\label{relations1c}
( \lie_n - 2 \nonaffinity) \volume_{ijk} = 0, \\
\label{relations1d}
( \lie_n - 2 \nonaffinity) \volumesmall_{ij} = 0.
  \end{eqnarray}
\end{subequations}
 Under the conformal transformation (\ref{conffreedom}) the
transformation laws for the various
fields on $\scp$ are given by the special case $\Upsilon = -\sigma$
of the transformation laws (\ref{contransformc})
discussed in Appendix \ref{app:nullsurf}, and are given by
\begin{subequations}
  \label{contransform}
  \begin{eqnarray}
  \label{contransform1}
  n^i &\to& e^\sigma n^i, \\
  \label{contransform2}  
  n_a &\to& e^{-\sigma} n_a, \\
    \label{contransform3}
    \inducedmetric_{ij} &\to& e^{-2 \sigma} \inducedmetric_{ij}, \\
    \label{contransform4a}
      \volumesmall_{ij} &\to& e^{-2 \sigma} \volumesmall_{ij}, \\
    \label{contransform4}
      \volume_{ijk} &\to& e^{-3 \sigma} \volume_{ijk}, \\
        \label{contransform5}
    \nonaffinity &\to& e^{\sigma} (\nonaffinity - \lie_n \sigma).
  \end{eqnarray}
\end{subequations}
These transformation laws preserve the relations (\ref{scriidentities}) and (\ref{relations1}).    
Using the freedom (\ref{contransform5}) one can enforce the Bondi
condition
\be
\nonaffinity = 0.
\label{bondicondition}
\ee
However in most of our analysis in this
paper we will not make this specialization and will allow $\nonaffinity$ to be nonzero.

\section{Symmetry groups at future null infinity in vacuum general relativity}
\label{sec:derivegroups}

In this appendix we derive the symmetry groups that correspond to the
three different field configuration spaces defined in Sec.\ \ref{sec:compendium} in the
body of the paper.  Rather than proceeding directly, it will be more
convenient to proceed in three stages, following  the universal intrinsic structure approach of Ashtekar
\cite{2014arXiv1409.1800A} and the techniques of Ref.\ \cite{CFP}:
\begin{itemize}

  \item We define universal intrinsic structures in each of the three
    cases, and derive the corresponding group of diffeomorphisms of
    $\scp$ that preserve these structures.

  \item We define boundary structures on $\scp$ in each of the three
    cases, and define associated field configuration spaces.  These
    configuration spaces are related to those given in
    Sec.\ \ref{sec:compendium} by taking orbits under the conformal
    transformations.

    \item Finally, we show that the symmetry groups of the intrinsic
      structures coincide with those of the field configuration spaces
      associated with the boundary structures, and with the symmetry
      groups of the spaces of Sec.\ \ref{sec:compendium}
      
\end{itemize}
We first explain these steps in detail in the BMS context,
and then outline the extensions to the generalized BMS and Weyl BMS
contexts.

\subsection{Bondi-Metzner-Sachs case}
\label{sec:bmscase}

\subsubsection{Definition of intrinsic structure}

Consider triplets of tensor fields $(n^i, \inducedmetric_{ij}, \nonaffinity)$ defined
on $\scp$ that satisfy the relations (\ref{relations1a}) and
(\ref{relations1b}) for which the vector field $n^i$ is complete.  We
define  
any two such triplets to be equivalent if they are related by a
rescaling of the form given by Eqs.\ (\ref{contransform1}),
(\ref{contransform3}),  and  (\ref{contransform5}):
\be
\label{eqr1}
(n^i, \inducedmetric_{ij}, \nonaffinity) \sim (e^\sigma n^i, e^{-2 \sigma} \inducedmetric_{ij},
e^\sigma \nonaffinity - e^\sigma \lie_n \sigma).
\ee
We denote the
equivalence class associated with a given triple as
\be
\mathfrak{u}_{21} = [n^i,\inducedmetric_{ij},\nonaffinity].
\label{eq:u21}
\ee
We call
the quantity $\mathfrak{u}_{21}$ an intrinsic geometric structure on
$\scp$.  These structures are universal in the sense that given any
two such structures on $\scp$, there exists a diffeomorphism $\varphi :
\scp \to \scp$ which maps one onto the other via
pullback\footnote{This can be shown by an argument similar to that given in
Sec.\ 4.1 of \cite{CFP}.}.

We will be defining a number of similar equivalence classes throughout
this appendix, and our notational conventions for these objects are
as follows.
In the symbol $\mathfrak{u}_{AB}$, $A$ can be $2$ (if the induced metric $\inducedmetric_{ij}$
is present in the set of fields), $1$ (if the volume form $\volume_{ijk}$ is instead present), or
$0$ (if neither $\inducedmetric_{ij}$ nor $\volume_{ijk}$ is
present).  The second index $B$
can be $1$ (if the inaffinity $\nonaffinity$ is present in the set of
fields) or $0$ (if $\nonaffinity$ is absent).  Thus there will
be six types of equivalence class, $\mathfrak{u}_{21}$, $\mathfrak{u}_{11}$, $\mathfrak{u}_{01}$,
$\mathfrak{u}_{20}$, $\mathfrak{u}_{10}$ and $\mathfrak{u}_{00}$.
Additionally, we will consider structures in which the normal covector $n_a$ is 
also present in the set of fields.  When this is the case, we will use the 
notation $\mf p_{AB}$, while the notation $\mf u_{AB}$ is reserved for structures
in which $n_a$ is absent.  
Finally tensor fields in the equivalence
classes are barred (eg.\ ${\bar n}^i, {\bar \inducedmetric}_{ij}, \ldots$) when $\nonaffinity$
is absent, and are not barred (eg.\ $n^i, \inducedmetric_{ij}, \ldots$)
when $\nonaffinity$ is present.

A given asymptotically flat spacetime $(\stm,\gphys_{ab})$ determines a
unique intrinsic structure $\mathfrak{u}_{21} = \mathfrak{u}_{21}[ \gphys_{ab}]$, as follows.  Choose an
unphysical metric $\gunphys_{ab}$ and conformal factor $\conf$ for
which $\gphys_{ab} = \conf^{-2} \gunphys_{ab}$.  Compute the
quantities $\inducedmetric_{ij}$, $n^i$ and $\nonaffinity$ from the unphysical metric and
conformal factor, and take the equivalence class (\ref{eq:u21}).  The
result is independent of which conformal factor and unphysical metric
within the equivalence class is chosen, by the equivalence relation (\ref{eqr1}) and the scaling
laws (\ref{contransform}).

We can define a different type of universal intrinsic structure \cite{2014arXiv1409.1800A}, without the
inaffinity $\nonaffinity$, as follows.  Consider pairs
$({\bar n}^i, {\bar \inducedmetric}_{ij})$ that satisfy Eqs.\ (\ref{relations1a}) and
(\ref{relations1b}) with $\nonaffinity = 0$:
\be
\label{rels1}
   {\bar n}^i {\bar \inducedmetric}_{ij} = 0, \ \ \ \lie_{\bar n} {\bar \inducedmetric}_{ij} = 0.
   \ee
We define two such pairs to be equivalent if they are related by a
transformation of the form (\ref{contransform}) that preserves $\nonaffinity = 0$, that
is,
\be
\label{eqr2}
( {\bar n}^i, {\bar \inducedmetric}_{ij}) \sim (e^\sigma {\bar n}^i, e^{-2 \sigma}
   {\bar \inducedmetric}_{ij} )
   \ee
with $\lie_{\bar n} \sigma = 0$.  We denote the equivalence class
associated with a given pair as
\be
\mathfrak{u}_{20} = \left[ {\bar n}^i, {\bar \inducedmetric}_{ij} \right].
\label{eq:u20}
\ee

There is a one-to-one correspondence between intrinsic structures of
the type $\mathfrak{u}_{21}$ and those of the type $\mathfrak{u}_{20}$.  Given
an intrinsic structure $\left[ n^i, \inducedmetric_{ij}, \nonaffinity\right]$, if we consider
the set of representative triples $( {\bar n}^i, {\bar \inducedmetric}_{ij}, 0)$
with vanishing inaffinity, the result is the equivalence class
$\left[{\bar n}^i, {\bar \inducedmetric}_{ij} \right]$.  Conversely, given the
equivalence class $\left[{\bar n}^i, {\bar \inducedmetric}_{ij} \right]$, we can
take any element $({\bar n}^i, {\bar \inducedmetric}_{ij})$, consider the
corresponding triple $({\bar n}^i, {\bar \inducedmetric}_{ij},0)$, and then take
the equivalence class under the equivalence relation (\ref{eqr1}) to
generate the intrinsic structure of type $\mathfrak{u}_{21}$.
We will denote this one-to-one correspondence as $\mathfrak{u}_{21} =
\mathfrak{u}_{21} (\mathfrak{u}_{20}) $.

\subsubsection{Symmetry group of intrinsic structure}
\label{sec:sgis}

Consider now diffeomorphisms $\varphi: \scp \to \scp$.  We
define the action of the pullback $\varphi^*$ on an intrinsic structure
$\mathfrak{u}_{21} = \left[ n^i, \inducedmetric_{ij}, \nonaffinity \right]$
by acting with the pullback on a representative of the equivalence
class:
\be
\label{pullbackdefine}
\varphi^* \left[ n^i, \inducedmetric_{ij}, \nonaffinity \right] = \left[ \varphi^* n^i,
  \varphi^* \inducedmetric_{ij}, \varphi^* \nonaffinity \right].
\ee
This action is well defined, since if $(n^i, \inducedmetric_{ij}, \nonaffinity)$ and
$({\hat n}^i$, ${\hat \inducedmetric}_{ij}$, ${\hat \nonaffinity})$
are two triples related
by a rescaling function $\sigma$, then the pullbacks of these triples
are related by the rescaling function $\varphi^* \sigma$.
Now given an intrinsic structure $\mathfrak{u}_{21}$, we define the
corresponding symmetry group to be the group of diffeomorphisms which
preserves the intrinsic structure:
\be
\label{group21def}
{\cal D}_{\mathfrak{u}_{21}} = \left\{ \varphi: \scp \to \scp \right|
  \left. \varphi^* \mathfrak{u}_{21} = \mathfrak{u}_{21} \right\}.
\ee
From the definition (\ref{pullbackdefine}) and the equivalence
relation (\ref{eqr1}), given a diffeomorphism $\varphi$ in this group
and a representative $(n^i, \inducedmetric_{ij}, \nonaffinity)$ of the intrinsic
structure, the action of the diffeomorphism is that of a rescaling by
some smooth function $\alpha = \alpha(\varphi, n^i)$:
\bes
\label{actionvp0}
\bea
\label{actionvp0a}
\varphi^* n^i &=& e^{-\alpha} n^i, \\
\label{actionvp0b}
\varphi^* \inducedmetric_{ij} &=& e^{2 \alpha} \inducedmetric_{ij}, \\
\label{actionvp0c}
\varphi^* \nonaffinity &=& e^{-\alpha} ( \nonaffinity + \lie_n \alpha).
\eea
\ees
The dependence of the function $\alpha$ on the choice
of representative (or equivalently on the normalization of the normal)
is given by
\be
\alpha(\varphi, e^\sigma n^i) = \alpha(\varphi, n^i) + \sigma -
\varphi^* \sigma,
\label{alphascaling}
\ee
from Eqs.\ (\ref{eqr1}) and (\ref{actionvp0}).

We  similarly define the symmetry group ${\cal D}_{\mathfrak{u}_{20}}$ to be the
group of diffeomorphisms that preserves a given intrinsic structure
$\mathfrak{u}_{20} = \left[ {\bar n}^i, {\bar \inducedmetric}_{ij}\right]$:
\be
\label{group20def}
{\cal D}_{\mathfrak{u}_{20}} = \left\{ \varphi: \scp \to \scp \right|
  \left. \varphi^* \mathfrak{u}_{20} = \mathfrak{u}_{20} \right\}.
\ee
Because of the one-to-one correspondence discussed above, this group
coincides with the group (\ref{group21def}), in the sense that
\be
   {\cal D}_{\mathfrak{u}_{21}(\mathfrak{u}_{20})} = {\cal D}_{\mathfrak{u}_{20}},
   \label{groupsthesame}
\ee
where the notation is defined after Eq.\ (\ref{eq:u20}).
To see this in more detail, 
if $\varphi \in {\cal D}_{\mathfrak{u}_{20}}$ then $\varphi^*
\mathfrak{u}_{20} = \mathfrak{u}_{20}$, and so
$
\varphi^* \mathfrak{u}_{21} ( \mathfrak{u}_{20}) = \mathfrak{u}_{21} (
\varphi^* \mathfrak{u}_{20}) = \mathfrak{u}_{21} ( \mathfrak{u}_{20}) $, where we
have used covariance, and so 
$\varphi \in {\cal D}_{\mathfrak{u}_{21}(\mathfrak{u}_{20})}$.  The converse
uses
the fact that the mapping $\mathfrak{u}_{20} \to
\mathfrak{u}_{21}(\mathfrak{u}_{20})$ is a bijection.

Because of the equality (\ref{groupsthesame}), we can give an
alternative characterization of the symmetries in the group.
From the definition (\ref{pullbackdefine}) and the equivalence
relation (\ref{eqr2}), 
given a representative $({\bar n}^i, {\bar \inducedmetric}_{ij})$ of the intrinsic
structure $\mathfrak{u}_{20}$ and a diffeomorphisms $\varphi$ in ${\cal D}_{\mathfrak{u}_{20}}$, 
the action of the diffeomorphism is that of a rescaling by
some smooth function $\alpha = \alpha(\varphi, {\bar n}^i)$:
\bes
\label{actionvp}
\bea
\label{actionvpa}
\varphi^* {\bar n}^i &=& e^{-\alpha} {\bar n}^i, \\
\label{actionvpb}
\varphi^* {\bar \inducedmetric}_{ij} &=& e^{2 \alpha} {\bar \inducedmetric}_{ij},
\eea
\ees
where
\be
\label{actioncpc}
\lie_{\bar n} \alpha =0.
\ee
The dependence of the function $\alpha$ on the choice
of representative 
is given by
\be
\label{alphatransform}
\alpha(\varphi, e^\sigma {\bar n}^i) = \alpha(\varphi, {\bar n}^i) + \sigma -
\varphi^* \sigma,
\ee
from Eqs.\ (\ref{eqr1}) and (\ref{actionvp}), which coincides with the dependence (\ref{alphascaling}) except
that here we must have $\lie_{\bar n} \sigma = 0$ from Eq.\ (\ref{eqr2}).
Equations (\ref{rels1}), (\ref{actionvp}) and (\ref{actioncpc}) are the usual
definition\footnote{The induced metric ${\bar \inducedmetric}_{ij}$ induces a unique two-dimensional Riemannian metric on
the space of generators of $\scp$, from Eqs.\ (\ref{rels1}).  One can specialize
the choice of representative in the equivalence class $\left[ {\bar n}^i, {\bar \inducedmetric}_{ij} \right]$,
using the freedom (\ref{eqr2}), to make this metric have constant scalar curvature (i.e.\ be a round two metric).  While this
specialization is often used to simplify the presentation of the BMS group, it is not necessary to do so.}
of the BMS group.
The linearized versions of Eqs.\ (\ref{actionvp}) and (\ref{alphatransform}) are
\bes
\label{linearizedsyms}
\bea
\lie_\xi {\bar n}^i &=& - \alpha {\bar n}^i, \\
\lie_\xi {\bar \inducedmetric}_{ij} &=& 2 \alpha {\bar \inducedmetric}_{ij},
\eea
\ees
and
\be
\alpha(\xi^i, e^\sigma {\bar n}^i) = \alpha(\xi^i, {\bar n}^i) -
\lie_\xi \sigma,
\ee
where the infinitesimal diffeomorphism is represented by the vector
field $\xi^i$ on $\scp$.

\subsubsection{Definition of field configuration space}

We now turn to the definition of a field configuration space whose
symmetry group matches that of the intrinsic structures discussed above.  We start by defining
a geometric structure on $\scp$ which we call a {\it boundary
  structure}, which is an extension of our previous definition of
intrinsic structure. 
We consider sets of tensor fields on $\scp$ of the form
\be
( n^i, \inducedmetric_{ij}, \nonaffinity, n_a ),
\ee
where $n_a$ is a choice of normal covector, the remaining fields
satisfy the relations (\ref{relations1a}) and
(\ref{relations1b}), and the vector field $n^i$ is complete.
We
define  
any two such sets to be equivalent if they are related by a
rescaling of the form (\ref{contransform}) for some smooth function $\sigma$:
\be
\label{eqr3}
(n^i, \inducedmetric_{ij}, \nonaffinity, n_a) \sim (e^\sigma n^i, e^{-2 \sigma} \inducedmetric_{ij},
e^\sigma \nonaffinity - e^\sigma \lie_n \sigma, e^{-\sigma} n_a).
\ee
We denote the
equivalence class associated with a given set as
\be
\mathfrak{p}_{21} = [n^i,\inducedmetric_{ij},\nonaffinity,n_a].
\label{eq:p21}
\ee
A choice of equivalence class is the desired boundary structure on $\scp$.

A choice of boundary structure
$\mathfrak{p}_{21} = [n^i,\inducedmetric_{ij},\nonaffinity,n_a]$ determines a unique
intrinsic structure $\mathfrak{u}_{21}$:  choose a
representative $(n^i,\inducedmetric_{ij},\nonaffinity,n_a)$, discard $n_a$, and
form the equivalence class
$\mathfrak{u}_{21} = [n^i,\inducedmetric_{ij},\nonaffinity]$ under the equivalence relation
(\ref{eqr1}). The result is independent of the representative
initially chosen, from Eqs.\ (\ref{eqr1}) and 
(\ref{eqr3}).
We will denote this induced intrinsic structure by $\mf u_{21}(\mf p_{21})$.
The boundary structure contains more information than the intrinsic
structure, which is necessary for the definition of the field
configuration space.

Just as for intrinsic structures, a given asymptotically flat spacetime $(\stm,\gphys_{ab})$ determines a
unique boundary structure $\mathfrak{p}_{21} = \mathfrak{p}_{21}[ \gphys_{ab}]$, as follows.  Choose an
unphysical metric $\gunphys_{ab}$ and conformal factor $\conf$ for
which $\gphys_{ab} = \conf^{-2} \gunphys_{ab}$.  Compute the
quantities $\inducedmetric_{ij}$, $n^i$, $\nonaffinity$ and $n_a$ from the unphysical metric and
conformal factor, and take the equivalence class (\ref{eq:p21}).  The
result is independent of which conformal factor and unphysical metric
is chosen, by the equivalence relation (\ref{eqr3}) and the scaling laws (\ref{contransform}).

Just as for intrinsic structures, we can define a different type of boundary structure, without the
inaffinity $\nonaffinity$, as follows.  Consider triplets $({\bar n}^i, {\bar
  q}_{ij}, {\bar n}_a)$ 
that satisfy Eqs.\ (\ref{rels1}) for which ${\bar n}_a$ is a complete normal covector.
We define two such triplets to be equivalent if they are related by a
transformation of the form (\ref{contransform}) that preserves $\nonaffinity = 0$, that
is,
\be
\label{eqr4}
( {\bar n}^i, {\bar \inducedmetric}_{ij}, {\bar n}_a) \sim (e^\sigma {\bar n}^i, e^{-2 \sigma}
   {\bar \inducedmetric}_{ij} , e^{-\sigma} {\bar n}_a)
   \ee
with $\lie_{\bar n} \sigma = 0$.  We denote the equivalence class
associated with a given triplet as
\be
\mathfrak{p}_{20} = \left[ {\bar n}^i, {\bar \inducedmetric}_{ij} , {\bar n}_a \right].
\label{eq:p20}
\ee
Just as above, a boundary structure $\mf p_{20}$ determines a unique
intrinsic structure $\mf u_{20} = \mf u_{20}(\mf p_{20})$ by dropping
the normal covector ${\bar n}_a$.
Also, just as for intrinsic structures, 
there is a one-to-one correspondence between boundary structures of
the type $\mathfrak{p}_{21}$ and those of the type $\mathfrak{p}_{20}$,
which we will denote as $\mathfrak{p}_{21} =
\mathfrak{p}_{21} (\mathfrak{p}_{20}) $ and $\mathfrak{p}_{20} =
\mathfrak{p}_{20} (\mathfrak{p}_{21}) $.  Given an asymptotically flat
spacetime $(\stm, \gphys_{ab})$, we define the corresponding boundary
structure of the new type
to be
\be
\mf p_{20}(\gphys_{ab}) = \mf p_{20}( \mf p_{21}(\gphys_{ab})).
\label{pequal}
\ee

Next, given a boundary structure $\mf p_{21}$, we define the
corresponding field configuration space to be the set of all
unphysical metrics and conformal factors that are compatible with that
boundary structure:
\be
\Gamma_{\mf p_{21}} = \left\{ (\stm, \gunphys_{ab}, \conf) \in \Gamma_0
\right| \left. \mf p_{21}(\gunphys_{ab},\conf) = \mf p_{21}  \right\}.
\label{Gamma21def}
\ee
Similarly, given a boundary structure $\mf p_{20}$, 
we define the field configuration space
\be
\Gamma_{\mf p_{20}} = \left\{ (\stm, \gunphys_{ab}, \conf) \in \Gamma_0
\right| \left. \mf p_{20}(\gunphys_{ab},\conf) = \mf p_{20}  \right\}.
\ee
These two spaces coincide, from Eq.\ (\ref{pequal}), in the sense that
\be
\Gamma_{\mf p_{21}(\mf p_{20})} = \Gamma_{\mf p_{20}}.
\ee
An argument analogous to that given in Appendix B of Ref.\ \cite{CFP}
can be used to show that the orbit of $\Gamma_{\mf p_{21}}$ under
diffeomorphisms of $\stm$ is the entire space $\Gamma_0$ defined in Eq.\ (\ref{Gamma0def}).

\subsubsection{Symmetry group of field configuration space}

We now turn to a discussion of the symmetry group of diffeomorphisms
that preserve the configuration phase space,
\be
{\cal G}_{\mf p_{21}} = \left\{ \psi: \stm \to \stm \right|
\left. \, \psi(\scp) = \scp \, , \, \psi^* \Gamma_{\mf p_{21}} =
\Gamma_{\mf p_{21}} \right\}.
\label{calG21def}
\ee
These diffeomorphisms induce diffeomorphisms of $\scp$: for any
$\psi$ in ${\cal G}_{\mf p_{21}}$ we define
\be
\varphi = \psi |_{\scp},
\label{induced}
\ee
and since $\psi$ preserves the boundary, $\varphi$ is a diffeomorphism
from $\scp$ to $\scp$.  Next, since $\psi$ preserves $\scp$,
the pullback of any normal covector $n_a$ evaluated on $\scp$ must be a rescaling of that
normal, so we have 
\be
\psi^* n_a \hateq e^{\gamma} n_a,
\label{gammadef}
\ee
where $\gamma = \gamma(\psi,n_a)$ is a smooth function on $\scp$
which depends on the diffeomorphism and on the normalization of the
normal.  The dependence on the normalization of the normal is given by
\be
\gamma(\psi,e^{-\sigma} n_a) = \gamma(\psi,n_a) +
\sigma - \varphi^* \sigma,
\label{gammascaling}
\ee
from Eqs.\ (\ref{induced}) and (\ref{gammadef}).

The physical asymptotic symmetry group is given by modding out by trivial
diffeomorphisms whose asymptotic charges vanish:
\be
{\cal D}_{\mf p_{21}} = {\cal G}_{\mf p_{21}} / \sim.
\ee
Here the equivalence relation $\sim$ is defined so that two
difeomorphisms are equivalent if they are related by a trivial
diffeomorphism.  For spacetime boundaries that are null surfaces at a
finite location, the trivial diffeomorphisms are those with \cite{CFP}
\be
\varphi = {\rm identity}, \ \ \ \ \gamma = 0.
\ee
This is also true in the BMS context, and we will assume it remains
true for the more general symmetry groups discussed below, pending the
explicit computation of the corresponding charges.
It follows that the group ${\cal D}_{\mf p_{21}}$ is in one-to-one
correspondence with the set of pairs $(\varphi, \gamma)$:
\be
{\cal D}_{\mf p_{21}} \simeq \left\{ (\varphi,\gamma) \right|
\left. \psi \in {\cal G}_{\mf p_{21}} \right\}.
\label{ddw}
  \ee

We now argue that the group ${\cal D}_{\mf p_{21}}$ coincides with the
symmetry group ${\cal D}_{\mf u_{21}}$ of the intrinsic structure
discussed in Sec.\ \ref{sec:sgis}.  From the condition $\psi^*
\Gamma_{\mf p_{21}} = \Gamma_{\mf p_{21}}$ in the definition (\ref{calG21def}),
we obtain that for any $(\stm, \gunphys_{ab}, \conf)$ in $\Gamma_{\mf p_{21}}$
we have $\mf p_{21} = \psi^* \mf p_{21}(\gunphys_{ab}, \conf) =
 \mf p_{21}(\psi^* \gunphys_{ab}, \psi^* \conf) $.  Using 
 Eqs.\ (\ref{induced}) and (\ref{gammadef}) we can rewrite this as
 \be
 \left[ \varphi^* n^i, \varphi^* \inducedmetric_{ij}, \varphi^* \nonaffinity, e^\gamma
   n_a \right] =  \left[ n^i, \inducedmetric_{ij}, \nonaffinity, n_a \right].
 \ee
Using the equivalence relation (\ref{eqr3}) it follows that there exists a
scaling function $\alpha$ on $\scp$ for which
\bes
\label{actionvp1}
\bea
\varphi^* n^i &=& e^{-\alpha} n^i, \\
\label{actionvp1b}
\varphi^* \inducedmetric_{ij} &=& e^{2 \alpha} \inducedmetric_{ij}, \\
\varphi^* \nonaffinity &=& e^{-\alpha} ( \nonaffinity + \lie_n \alpha), \\
e^\gamma n_a &=& e^\alpha n_a.
\eea
\ees
The first three equations here coincide with Eqs.\ (\ref{actionvp}), which
imply that $\varphi$ lies in ${\cal D}_{\mf u_{21}}$.  The last
equation implies that $\alpha = \gamma$, which is compatible with the
scaling laws (\ref{alphascaling}) and (\ref{gammascaling}). 
In particular this implies that $\gamma$ is determined by $\varphi$,
$\gamma = \gamma(\varphi)$, which implies from Eq.\ (\ref{ddw}) that
${\cal D}_{\mf p_{21}}$ and ${\cal D}_{\mf u_{21}}$ are isomorphic.

\subsubsection{Alternative definition of field configuration space with conformal freedom fixed}

The literature has often used an alternative definition of the field
configuration space, which differs from the definition (\ref{Gamma21def}) given
above only in that the conformal freedom is fixed \cite{Wald:1999wa,Flanagan:2019vbl}.
This configuration space $\Gamma_{\rm BMS}$ is defined in Eq.\ (\ref{eq:fcs}) above, and depends on 
a choice of conformal factor $\conf_0$ on a neighborhood
${\cal D}$ of $\scp$ and a choice of unphysical metric
$\gunphys_{0\,ab}$ on $\scp$.

We now show that the orbit of $\Gamma_{\rm BMS}$ under
conformal transformations is a particular space
$\Gamma_{\mf p_{21}}$, where $\mf p_{21} = \left[ {\bar n}_0^i, {\bar
    q}_{0\,ij}, 0, {\bar n}_{0\,a} \right]$
and ${\bar n}_0^i$, ${\bar \inducedmetric}_{0\,ij}$ and ${\bar n}_{0\,a}$ are computed from the given
data $\conf_0$ on ${\cal D}$ and $\gunphys_{0\,ab}$ on $\scp$.
First, it follows from the definitions
(\ref{Gamma21def}) and (\ref{eq:fcs}) that 
$\Gamma_{\rm BMS} \subset \Gamma_{\mf p_{21}}$.
Next, suppose that
$(\stm, \gunphys_{ab}, \conf)$ lies in $\Gamma_{\mf p_{21}}$.  It follows
that
\be
\left[ n^i, \inducedmetric_{ij}, \nonaffinity, n_a \right] = \left[ {\bar n}_0^i, {\bar \inducedmetric}_{0\,ij}, 0,
  {\bar n}_{0 a} \right],
\ee
where the fields on the left hand side are computed from
$\gunphys_{ab}, \conf$.  From the equivalence relation (\ref{eqr3})
there exists a scaling function $\alpha$ on $\scp$ so that 
\be
\left( n^i, \inducedmetric_{ij}, \nonaffinity, n_a \right) = \left( e^{-\alpha} {\bar n}_0^i,
e^{2 \alpha} {\bar \inducedmetric}_{0\,ij}, e^{-\alpha} \lie_{{\bar n}_0} \alpha,
  e^\alpha {\bar n}_{0 a} \right).
\ee
By suitably extending the definition of $\alpha$ from $\scp$ into the interior
of the spacetime we can
make $e^{-\alpha} \conf$ coincide with $\conf_0$ on ${\cal D}$, since
the gradients of these functions agree on $\scp$.
It then follows that
\be
\left( \stm, e^{-2
  \alpha} \gunphys_{ab}, e^{-\alpha} \conf \right)
\ee
lies in $\Gamma_{\rm BMS}$.

From this relation between $\Gamma_{\rm BMS}$ and $\Gamma_{\mf
  p_{21}}$, it follows that the asymptotic symmetry group of
$\Gamma_{\rm BMS}$ coincides with ${\cal D}_{\mf p_{21}}$.
Note however that a bulk diffeomorphism $\psi$ acts differently on the
two spaces.  On $\Gamma_{\rm BMS}$ it acts in tandem with a
conformal transformation to preserve the conformal factor,
\be
\left( \gunphys_{ab}, \, \conf \right) \to \left[ \left( \frac{
  \conf} {\psi^* \conf} \right)^2 \psi^* \gunphys_{ab}, \, \conf
\right],
\ee
while on $\Gamma_{\mf p_{21}}$ it acts simply as 
\be
\left( \gunphys_{ab}, \, \conf \right) \to \left( \psi^*
\gunphys_{ab}, \, \psi^* \conf \right).
\ee

\subsection{Generalized BMS field configuration space and symmetry group}

We now turn to the generalized BMS field configuration space and
generalized BMS group of \cite{Compere:2018ylh,CL,Campiglia:2014yka}.
The discussion in this case mirrors exactly the discussion of the BMS
case given in the previous section, with the following modifications:

\begin{itemize}

\item The induced metric $\inducedmetric_{ij}$ is replaced everywhere by the volume
  form $\volume_{ijk}$.  Thus we use Eq.\ (\ref{relations1c}) instead of
  Eqs.\ (\ref{relations1a}) and (\ref{relations1b}), and use the
  scaling relation (\ref{contransform4}) everywhere instead of the relation
  (\ref{contransform3}).
  
\item The equivalence relation (\ref{eqr1}) is replaced with
\be
\label{eqr5}
(n^i, \volume_{ijk}, \nonaffinity) \sim (e^\sigma n^i, e^{-3 \sigma} \volume_{ijk},
e^\sigma \nonaffinity - e^\sigma \lie_n \sigma),
\ee
and the definition (\ref{eq:u21}) of intrinsic structure is replaced
by
\be
\mathfrak{u}_{11} = [n^i,\volume_{ijk},\nonaffinity].
\label{eq:u11}
\ee

\item Similarly the equivalence relation (\ref{eqr2}) is replaced by
\be
\label{eqr6}
( {\bar n}^i, {\bar \volume}_{ijk}) \sim (e^\sigma {\bar n}^i, e^{-3 \sigma}
   {\bar \volume}_{ijk} )
   \ee
with $\lie_{\bar n} \sigma = 0$ and $\lie_{\bar n} {\bar \volume}_{ijk} =
0$.  The 
definition (\ref{eq:u20}) of intrinsic structure is replaced
\be
\mathfrak{u}_{10} = \left[ {\bar n}^i, {\bar \volume}_{ijk} \right].
\label{eq:u10}
\ee

\item The corresponding symmetry groups ${\cal D}_{\mf u_{11}}$ and
  ${\cal D}_{\mf u_{10}}$ are defined as in Sec.\ \ref{sec:sgis}, and again
  coincide in the appropriate sense.  The  relations (\ref{actionvp})
  that define the symmetries are replaced by
\bes
\label{actionvp11}
\bea
\varphi^* {\bar n}^i &=& e^{-\alpha} {\bar n}^i, \\
\varphi^* {\bar \volume}_{ijk} &=& e^{3 \alpha} {\bar \volume}_{ijk},
\eea
\ees
where $\lie_{\bar n} \alpha =0$, whose linearized versions are\cite{Flanagan:2019vbl}
\bes
\label{linearizedsyms1}
\bea
\lie_\xi {\bar n}^i &=& - \alpha {\bar n}^i, \\
\lie_\xi {\bar \volume}_{ijk} &=& 3 \alpha {\bar \volume}_{ijk}.
\eea
\ees

\item The definitions (\ref{eq:p21}) and (\ref{eq:p20}) of boundary structures are
  replaced by the analogous definitions
\be
\mathfrak{p}_{11} = [n^i,\volume_{ijk},\nonaffinity,n_a].
\label{eq:p11}
\ee
and
\be
\mathfrak{p}_{10} = \left[ {\bar n}^i, {\bar \volume}_{ijk} , {\bar n}_a \right].
\label{eq:p10}
\ee

\item The corresponding field configuration spaces $\Gamma_{\mf p_{11}}$ and
$\Gamma_{\mf p_{10}}$ are defined as before, and the argument that the
  corresponding symmetry groups ${\cal D}_{\mf p_{11}}$ and ${\cal
    D}_{\mf p_{10}}$ coincide with those of the intrinsic structures
  is unchanged.

\item The definition (\ref{eq:fcs}) of the conformal-freedom-fixed field
  configuration space is replaced with the definition (\ref{eq:fcs1}) of the space $\Gamma_{\rm GBMS}$.
As before, one can show that taking the orbit of $\Gamma_{\rm GBMS}$ under conformal
transformations yields a particular space $\Gamma_{\mf p_{11}}$, with
${\mf p_{11}} = \left[ {\bar n}^i_0, \, {\bar \volume}_{0\,ijk}, \, 0,
  \, {\bar n}_{0\,a} \right]$, and that the asymptotic symmetry group of
$\Gamma_{\rm GBMS}$ coincides with ${\cal D}_{\mf p_{11}}$.

\end{itemize}

\subsection{Weyl BMS field configuration space and symmetry group}

The field configuration space can be further expanded by omitting
both the induced metric $\inducedmetric_{ij}$ and volume form $\volume_{ijk}$
from the definitions.
We call the resulting space the {\it Weyl BMS}
field configuration space, following Ref.\ \cite{freidel2021weyl},
since the extra symmetries correspond to conformal
transformations
of
the form (\ref{contransform}) that are independent of other pieces of
the symmetry generator.
The resulting symmetry group then coincides with the symmetry group of
general null surfaces at finite locations derived in Ref.\ \cite{CFP}.
This coincidence of symmetry groups should facilitate understanding how the asymptotic symmetry
group is obtained from a limit of symmetry groups on finite null
boundaries.  It will also be important in future derivations of global
conservation laws in black hole spacetimes, where analyses
analogous to \cite{Prabhu_2019,Prabhu:2021cgk} at future
timelike infinity will be needed in order to determine the appropriate
matching of symmetry generators on the future horizon with those on
future null infinity; see for example the discussion in Sec.\ 7
of \cite{CFP}.

For the Weyl BMS field configuration space, the required modifications to the discussion
of the BMS case of Sec. \ref{sec:bmscase} are:

\begin{itemize}

\item The induced metric $\inducedmetric_{ij}$ is omitted everywhere. Thus the 
equivalence relation (\ref{eqr1}) is replaced with
\be
\label{eqr7}
(n^i, \nonaffinity) \sim (e^\sigma n^i, e^\sigma \nonaffinity - e^\sigma \lie_n \sigma),
\ee
and the definition (\ref{eq:u21}) of intrinsic structure is replaced
by
\be
\mathfrak{u}_{01} = [n^i,\nonaffinity].
\label{eq:u01}
\ee
Similarly the equivalence relation (\ref{eqr2}) is replaced by
\be
\label{eqr8}
( {\bar n}^i ) \sim (e^\sigma {\bar n}^i)
   \ee
with $\lie_{\bar n} \sigma = 0$.  The 
definition (\ref{eq:u20}) of intrinsic structure is replaced
\be
\mathfrak{u}_{00} = \left[ {\bar n}^i \right].
\label{eq:u00}
\ee

\item The corresponding symmetry groups ${\cal D}_{\mf u_{01}}$ and
  ${\cal D}_{\mf u_{00}}$ are defined as in Sec.\ \ref{sec:sgis}, and again
  coincide in the appropriate sense.  The  relations (\ref{actionvp})
that define the symmetries are replaced by
\be
\label{actionvp111}
\varphi^* {\bar n}^i = e^{-\alpha} {\bar n}^i, 
\ee
where $\lie_{\bar n} \alpha =0$, whose linearized versions is
\be
\label{linearizedsyms2}
\lie_\xi {\bar n}^i = - \alpha {\bar n}^i.
\ee
The symmetry group (\ref{actionvp111}) coincides with that of general
finite null surfaces, given by Eqs.\ (4.4) of Ref.\ \cite{CFP} specialized to
$\nonaffinity =0$.

\item The definitions (\ref{eq:p21}) and (\ref{eq:p20}) of boundary structures are
  replaced by the analogous definitions
\be
\mathfrak{p}_{01} = [n^i,\nonaffinity,n_a].
\label{eq:p01}
\ee
and
\be
\mathfrak{p}_{00} = \left[ {\bar n}^i , {\bar n}_a \right].
\label{eq:p00}
\ee

\item The corresponding field configuration spaces $\Gamma_{\mf p_{01}}$ and
$\Gamma_{\mf p_{00}}$ are defined as before, and the argument that the
  corresponding symmetry groups ${\cal D}_{\mf p_{01}}$ and ${\cal
    D}_{\mf p_{00}}$ coincide with those of the intrinsic structures
  is unchanged.

\item The definition (\ref{eq:fcs}) of the conformal-freedom-fixed field
  configuration space is replaced with the definition (\ref{eq:fcs11})
  of the space $\Gamma_{\rm WBMS}$. As before, one can show that
  taking the orbit of $\Gamma_{\rm WBMS}$ under conformal 
transformations yields a particular space $\Gamma_{\mf p_{01}}$, with
${\mf p_{01}} = \left[ {\bar n}^i_0, \, 0,
  \, {\bar n}_{0\,a} \right]$, and that the asymptotic symmetry group of
$\Gamma_{\rm WBMS}$ coincides with ${\cal D}_{\mf p_{01}}$.

\end{itemize}

\subsection{Properties of the asymptotic symmetry groups}

We now turn to a characterization of the structure of the symmetry
groups discussed in the previous sections and the corresponding
algebras.

For convenience, we will specialize to the definitions
${\cal D}_{\mf  u_{20}}$,  ${\cal D}_{\mf  u_{10}}$ and  ${\cal
  D}_{\mf  u_{00}}$ of these groups in which the inaffinity $\nonaffinity$
has been set to zero, given by Eq.\ (\ref{group20def}) and its avatars.
For each universal structure $\mf u_{20}$, $\mf u_{10}$ or $\mf u_{00}$, we
pick a corresponding representative $({\bar n}^i, {\bar \inducedmetric}_{ij})$,
$({\bar n}^i, {\bar \volume}_{ijk})$, or $( {\bar n}^i)$.
The null generator ${\bar n}^i$ is common to all of these representatives,
and we construct a coordinate system $(u,\theta^A)$ on $\scp$ using
this normal as described in Sec.\ \ref{sec:symgroups} in the body of
the paper.  The symmetry transformations are then given by
Eqs.\ (\ref{diffeogen}).

To derive these transformations, we start
with Eq.\ (\ref{actionvpa}), which is common to all three groups.
Combining this with Eq.\ (\ref{ucoorddef}) yields 
\be
\varphi^* \partial_u = e^{-\alpha} \partial_u = \frac{\partial
  u}{\partial {\hat u}} \partial_u + \frac{\partial \theta^A}{\partial
  {\hat u}} \partial_A,
\ee
which yields $\theta^A = \theta^A({\hat \theta}^B)$, and inverting
yields Eq.\ (\ref{barthetadef}).
It also yields
\be
\label{condt322}
\partial_u {\hat u}(u,\theta^A) = e^{\alpha(u,\theta^A)}.
\ee
Also from Eq.\ (\ref{actioncpc}) which applies to all three groups we
obtain that $\partial_u e^\alpha =0$, so that
$\alpha = \alpha(\theta^A)$, and now integrating Eq.\ (\ref{condt322})
yields Eq.\ (\ref{barudef}).

This completes the derivation for the Weyl BMS case, where the
functions $\chi$, $\alpha$ and $\gamma$ are unconstrained.
For the generalized BMS case, it follows from the conditions (\ref{actionvp11})
and the definition (\ref{volumesmall}) of $\volumesmall_{ij}$ that the function
$\alpha$ is given by Eq.\ (\ref{alpha333}).  Similarly, for the BMS
case, it follows from the condition (\ref{actionvp1b}) 
that the function
$\alpha$ is given by Eq.\ (\ref{alpha444}).

Finally, it can be useful to understand the action of the groups on
representatives of the universal structures for which $\nonaffinity \ne 0$.
We specialize for simplicity to linearized supertranslations of the
form
\be
\xi^i = f n^i.
\ee
Note that the symmetry generator $\xi^i$ is invariant under the conformal
rescalings (\ref{contransform}) by definition, but that the
coefficient $f$ has a nonzero conformal weight, transforming as $f \to
e^{-\sigma} f$ from Eq.\ (\ref{contransform1}).
For the BMS group the coefficient $f$ satisfies
the conformally invariant equation
\be
\label{stchar}
(\lie_n - \nonaffinity) f = 0,
\ee
from Eqs.\ (\ref{relations1a}), (\ref{relations1b}), (\ref{actionvp0a}) and (\ref{actionvp0b}).
This equation is also valid for the generalized BMS group, from
Eqs.\ (\ref{relations1c}) together with the unbarred version of Eqs.\ (\ref{linearizedsyms1}).
Finally, for the Weyl BMS group, Eq.\ (\ref{stchar}) is replaced
with the conformally invariant equation\footnote{This equation differs
from the corresponding equation (4.17) of Ref.\ \cite{CFP},
despite the fact that the underlying algebras of infinitesimal
diffeomorphisms $\vec \xi$ on the
null surfaces coincide.  The difference arises from the fact that 
the scaling properties of $\nonaffinity$ and $n^i$ differ in the two cases
(see Appendix \ref{app:nullsurf}).}
\be
\lie_n (\lie_n - \nonaffinity) f = 0,
\ee
from Eqs.\ (\ref{actionvp0a}) and (\ref{actionvp0c}) which apply to the group ${\cal
  D}_{\mf u_{01}}$. This equation now admits the two different kinds
of supertranslations as solutions.

\section{Details of holographic renormalization with a rigging vector
  field}
\label{app:hr}

In this appendix we derive some of the results on holographic
renormalization which were discussed in Sec.\ \ref{sec:hr}.    

We start by inserting the coordinate expansions 
(\ref{fieldsexpand}) of the Lagrangian
and symplectic form into the identity (\ref{eqn:dL}), and specializing to on
shell field configurations.  This yields
\be
\theta^{\prime\,0}_{,0} + \theta^{\prime\,i}_{,i} = \delta {\cal L}.
\label{eq:onshell}
\ee
Similarly the boundary canonical transformation (\ref{JKM11}) can be
written in terms of coordinate components as 
\bes
\bea
    {\cal L}_{\rm ren} &=& {\cal L} + B^{i}_{,i} + B^{0}_{,0},   \\
    \label{theta0primeans}
  \theta^{\prime\,0}_{\rm ren} &=& \theta^{\prime\,0} + \delta B^{0} - \Lambda^{i}_{,i}, \\
  \theta^{\prime\,i}_{\rm ren} &=& \theta^{\prime\,i} + \delta B^{i} + \Lambda^{i}_{,0} + 2 \Lambda^{ij}_{,j}.
\eea
\ees
It follows that the choices (\ref{translated}) of $B$ and $\Lambda$ yield ${\cal L}_{\rm ren}=\theta^{\prime\,i}_{\rm ren}=0$,
cf.\ Eq.\ (\ref{thetaifinite}).  Also differentiating
Eq.\ (\ref{theta0primeans}) with respect to $x^0$ and combining with
Eq.\ (\ref{translated}) and (\ref{eq:onshell}) gives
$\theta^{\prime\,0}_{\rm ren\,,0} =0$.  Hence to evaluate
$\theta^{\prime\,0}_{\rm ren}$ at $x^0=0$ we can evaluate it at $x^0 = \upsilon_0$,
at which value it reduces to $\theta^{\prime\,0}(\upsilon_0)$, from
Eqs.\ (\ref{theta0primeans}) and (\ref{translated}).  This yields the
result (\ref{theta0finite}).

We next turn to the computation of the anomalies (\ref{conj}).  Given
the prescription (\ref{bprimelambdaprime}) for $B[L',{\vec v}]$,
the anomaly is given by, from the definition (\ref{eqn:noncov-defn}),
\be
\Delta_{\hat \xi} B = B[\psi^* L', {\vec v} ] - B[\psi^* L', \psi^*
  {\vec v}].
\ee
Here it is understood that the right hand side is to be linearized in
the diffeomorphism $\psi$, whose linear part is parameterized by the
vector field ${\vec \xi}$.  (It will be convenient to initially work with the
full nonlinear diffeomorphism rather than its linearized version).  
Acting on both sides with $\psi^{-1\,*}$, we see that the right hand
side is proportional to $\xi$, and so we can drop the $\psi^{-1\,*}$ on
the left hand side when working to linear order.  This gives
\be
\Delta_{\hat \xi} B = B[L', \psi^{-1\,*} {\vec v} ] - B[ L', 
  {\vec v}].
\label{anomdef1}
\ee
Defining $B_1 = B[L',\psi^{-1\,*} {\vec v}]$, we have from
Eqs.\ (\ref{bprimelambdaprime}) that $B_1$ is given by
\bes
\bea
\label{anom0}
i_{\tilde v} B_1 &=& 0, \\
\label{anom1}
{\tilde \pi}_{\upsilon }^* B_1 &=&
\int_{\upsilon}^{\upsilon_0} d {\bar \upsilon} \, {\tilde \pi}_{{\bar \upsilon}}^*
i_{\tilde v} L',
\eea
\ees
where ${\tilde {\vec v}} = \psi^{-1\,*} {\vec v}$ and
\be
{\tilde \pi}_\upsilon = \psi \circ \pi_{\upsilon} \circ
\varphi^{-1},
\ee
with $\varphi : {\cal N} \to {\cal N}$ being the restriction of $\psi$
to the boundary ${\cal N}$.  Acting on both sides of Eq.\ (\ref{anom1})
with $\varphi^*$ now gives
\be
\pi_{\upsilon}^* \psi^* B_1 = \int_\upsilon^{\upsilon_0} d
   {\bar \upsilon} \pi_{{\bar \upsilon}}^* \psi^* i_{\tilde v}
   L'.
   \ee
Using $\psi^* = 1 + \lie_\xi + \ldots$ together with Eq.\ (\ref{anomdef1}) this can be rewritten as
\be
\label{semifinal}
\pi_{\upsilon}^* \Delta_{\hat \xi} B = - \pi_{\upsilon}^*
\lie_\xi B + \int_\upsilon^{\upsilon_0} d
   {\bar \upsilon} \pi_{{\bar \upsilon}}^* i_v \lie_\xi L'.
\ee

We now switch to using the coordinate notation of
Sec.\ \ref{sec:rspL}.  First, for any vector field ${\vec w} = w^0
\partial_0 + w^i \partial_i$ and any $d$-form $\chi = \chi^0 \varpi +
\chi^i dx^0 \wedge \varpi_i$, the Lie derivative is given by
\be
\lie_w \chi = \left[ w^0 \chi^0_{,0} - w^0_{,i} \chi^i + ( w^i
  \chi^0)_{,i} \right] \varpi
+ \left[ - w^i_{,0} \chi^0 - w^i_{,j} \chi^j + (w^0 \chi^i)_{,0} +
  (w^j \chi^i)_{,j} \right] dx^0 \wedge \varpi_i.
\ee
Using this formula together with Eqs.\ (\ref{translated}), we find
Eq.\ (\ref{semifinal})
reduces to
\be
( \Delta_{\hat \xi} B)^0=
(\xi^0 {\cal L})(\upsilon_0) - \left(
\xi^i \int_{\upsilon}^{\upsilon_0} d{\bar \upsilon} {\cal L}
- \int_{\upsilon}^{\upsilon_0} d{\bar \upsilon} \xi^i {\cal
  L} \right)_{,i}.
\label{anomr1}
\ee
The other component of $\Delta_{\hat \xi} B$ is given by combining
Eq.\ (\ref{anom0}) with ${\tilde {\vec v}} = {\vec v} - \lie_\xi {\vec v}$
and Eq.\ (\ref{anomdef1}), which gives
\be
\label{semifinal1}
i_v \Delta_{\hat \xi} B = i_{\lie_\xi v} B.
\ee
Using ${\vec v} = \partial_0$ this yields
\be
( \Delta_{\hat \xi} B)^i = \xi^i_{,0}
\int_{\upsilon}^{\upsilon_0} d {\bar \upsilon} {\cal L}.
\label{anomr2}
\ee
Combining the results (\ref{anomr1}) and (\ref{anomr2}) with
Eqs.\ (\ref{wq1}), (\ref{wq2}) and (\ref{kappaxic}) now shows consistency with the
identity (\ref{conj1}).

The derivation of $\Delta_{\hat \xi} \Lambda$ is exactly analogous.
Equations (\ref{semifinal}) and (\ref{semifinal1}) are replaced by
\bes
\bea
\label{semifinallambda}
\pi_{\upsilon}^* \Delta_{\hat \xi} \Lambda &=& - \pi_{\upsilon}^*
\lie_\xi \Lambda - \int_\upsilon^{\upsilon_0} d
   {\bar \upsilon} \pi_{{\bar \upsilon}}^* i_v \lie_\xi \theta', \\
i_v \Delta_{\hat \xi} \Lambda &=& i_{\lie_\xi v} \Lambda,
\eea
\ees
which together yield
\be
\Delta_{\hat \xi} \Lambda = -\left[ (\xi^0 \theta^{\prime\,i})(\upsilon_0) -
  \int_{\upsilon}^{\upsilon_0} d {\bar \upsilon} \xi^i_{,0}
  \theta^{\prime\,0} + \xi^j   \int_{\upsilon}^{\upsilon_0} d {\bar
    \upsilon} \theta^{\prime\,i}_{,j}
-    \int_{\upsilon}^{\upsilon_0} d {\bar \upsilon}
\xi^j \theta^{\prime\,i}_{,j} \right] \varpi_i + \xi^i_{,0}
\int_{\upsilon}^{\upsilon_0} d {\bar \upsilon} \theta^{\prime\,j} \,
dx^0 \wedge \varpi_{ij}.
\ee
Combining this with Eqs.\ (\ref{wq3}), (\ref{wq4}), (\ref{kappamuans})
and (\ref{eq:onshell}) now shows consistency with the identity (\ref{conj2}).

\bibliographystyle{JHEPthesis}
\bibliography{asycps}

\end{document}